\documentclass[a4paper,10pt]{article}
\usepackage[utf8]{inputenc}
\usepackage{amsthm}
\usepackage{amsfonts}
\usepackage{amssymb}	
\usepackage{amsmath}
\usepackage{mathtools}
\usepackage[english]{babel}
\usepackage{color}
\usepackage{cite}
\usepackage[numbers]{natbib}
\usepackage{slashed}
\usepackage{enumerate}
\usepackage{graphicx}
\usepackage{systeme}
\usepackage{dsfont}
\usepackage{relsize}
\usepackage[multiple]{footmisc}
\RequirePackage{geometry}
\usepackage{float}
\usepackage{subcaption}
\usepackage{hyperref}
\usepackage[labelformat=simple]{subcaption}

\newcommand{\R}{\mathbb{R}}
\newcommand{\C}{\mathbb{C}} 
\newcommand{\eps}{ {\varepsilon} }
\renewcommand{\d}{\partial}
\renewcommand{\div}{{\rm div}}

\newcommand{\kabs}{|k|}

\renewcommand{\Im}{{\rm Im}}
\renewcommand{\Re}{{\rm Re}}
\newcommand{\idmat}{\mathds{1}}

\newcommand{\op}{(1)}
\newcommand{\twop}{(2)}
\newcommand{\thp}{(3)}
\newcommand{\fp}{(4)}
\newcommand{\fip}{(5)}

\newcommand{\mue}{\mu_e}
\newcommand{\muc}{\mu_c}
\newcommand{\mumic}{\mu_{\text{micro}}}
\newcommand{\lame}{\lambda_e}
\newcommand{\lammic}{\lambda_{\text{micro}}}
\newcommand{\Lc}{L_c^2}

\newcommand{\norm}[1]{\left\lVert#1\right\rVert}

\newcommand{\sigmatil}{\widetilde{\sigma}}

\usepackage{empheq}

\usepackage{mdframed}
\usepackage{lipsum} 
\usepackage[many]{tcolorbox}

\theoremstyle{remark}

\newtheorem{lemma}{Lemma}

\DeclareMathOperator{\tr}{tr}

\DeclareMathOperator{\curl}{curl}
\DeclareMathOperator{\Div}{Div}
\DeclareMathOperator{\Curl}{Curl}
\DeclareMathOperator{\sym}{sym}
\let\skew\relax

\DeclareMathOperator{\skew}{skew}

\numberwithin{equation}{section}

\makeatletter

\let\@fnsymbol\@arabic

\makeatother
\title{Low-and high-frequency Stoneley waves, reflection and transmission at a Cauchy/relaxed micromorphic interface}
\author{Alexios Aivaliotis\thanks{Alexios Aivaliotis, corresponding author, alexios.aivaliotis@insa-lyon.fr, GEOMAS, INSA-Lyon, Universit\'e de Lyon, 20 avenue Albert Einstein,
		69621, Villeurbanne cedex, France}~, Ali Daouadji\thanks{Ali Daouadji, ali.daouadji@insa-lyon.fr, GEOMAS, INSA-Lyon, Universit\'e de Lyon, 20 avenue Albert Einstein,
		69621, Villeurbanne Cedex, France}~, Gabriele Barbagallo\thanks{Gabriele Barbagallo, gabriele.barbagallo@insa-lyon.fr, GEOMAS, INSA-Lyon, Universit\'e de Lyon, 20 avenue Albert Einstein,
		69621, Villeurbanne Cedex, France}~, Domenico Tallarico\thanks{Domenico Tallarico, domenico.tallarico@insa-lyon.fr, GEOMAS, INSA-Lyon, Universit\'e de Lyon, 20 avenue Albert Einstein,
		69621, Villeurbanne Cedex, France},\\Patrizio Neff \thanks{Patrizio Neff, patrizio.neff@uni-due.de, Head of Chair for Nonlinear Analysis and Modelling, Fakultät für Mathematik, Universität Duisburg-Essen,
		Mathematik-Carrée, Thea-Leymann-Straße 9, 45127 Essen, Germany}~ and Angela Madeo\thanks{Angela Madeo, angela.madeo@insa-lyon.fr, GEOMAS, INSA-Lyon, Universit\'e
		de Lyon, 20 avenue Albert Einstein, 69621, Villeurbanne Cedex, France}}
\date{}

\begin{document}
\maketitle
\small
\begin{abstract}
In this paper we study the reflective properties of a $2$D interface separating a homogeneous solid from a band-gap metamaterial by modeling it as an interface between a classical Cauchy continuum and a relaxed micromorphic medium. We show that the proposed model is able to predict the onset of Stoneley interface waves at the considered interface both at low and high-frequency regimes. More precisely, critical angles for the incident wave can be identified, beyond which classical Stoneley waves, as well as microstructure-related Stoneley waves appear. We show that this onset of Stoneley waves, both at low and high frequencies, strongly depends on the relative mechanical properties of the two media. We suggest that a suitable tailoring of the relative stiffnesses of the two media can be used to conceive ``smart interfaces'' giving rise to wide frequency bounds where total reflection or total transmission may occur.
\end{abstract}
\normalsize
\textbf{Keywords:} enriched continuum mechanics, metamaterials, band-gaps, wave-propagation, Rayleigh waves, Stoneley waves, relaxed micromorphic model, interface, total reflection and transmission.

\vspace{2mm}
\textbf{}\\
\textbf{AMS 2010 subject classification:} 74A10 (stress), 74A30 (nonsimple
materials), 74A60 (micromechanical theories),
74B05 (classical linear elasticity), 74J05 (linear waves), 74J10 (bulk waves), 75J15 (surface waves), 74J20 (wave scattering), 74M25 (micromechanics), 74Q15
(effective constitutive equations).
\newpage 
\tableofcontents
\section{Introduction}
In this paper, we rigorously study the two-dimensional problem of reflection and transmission of elastic waves at an interface between a homogeneous solid and a band-gap metamaterial modeled as a Cauchy medium and a relaxed micromorphic medium, respectively. 

We propose a systematic revision of the reflection/transmission problem at a Cauchy/Cauchy interface in order to set up our notation and to carry out clearly all analytical results concerning the critical incident angles, giving rise to Stoneley interface weaves. This systematic presentation allows us to readily generalize the adopted techniques to the case of reflection and transmission at an interface separating a Cauchy medium from a relaxed micromorphic one. We observe in this latter case the existence of high-frequency critical angles of incidence, which are able to determine the onset of microstructure-related Stoneley interface waves. We clearly show that both low and high-frequency interface waves are directly related to the relative stiffnesses of the two media and can discriminate between total reflection and total transmission phenomena.

\textcolor{black}{Recent years have seen the abrupt development of acoustic metamaterials whose mechanical properties allow unorthodox material behaviors such as band-gaps (\cite{wang2014harnessing,xiao2011longitudinal,liu2000locally}), cloaking (\cite{misseroni2016cymatics,chen2007acoustic,valentine2009optical,bueckmann2015mechanical}), focusing (\cite{guenneau2007acoustic,cummer2016controlling,tallarico2017tilted}), wave-guiding  (\cite{kaina2014slow,tallarico2017edge}) etc. The bulk behavior of these metamaterials has gathered the attention of the scientific community via the refinement of mathematical techniques, such as Bloch-Floquet analysis or homogenization techniques (\cite{lions1978asymptotic,craster2010high}). More recently, filtering properties of bulk periodic media, i.e. the onset of so-called band-gaps and non-linear dispersion, has been given attention in the framework of enriched continuum models (\cite{neff2014unifying,dagostino2018effective,madeo2014band,madeo2015wave,madeo2016complete,madeo2016first,madeo2016reflection,madeo2017modeling,madeo2017relaxed,madeo2017review,madeo2017role,barbagallo2017transparent,neff2015relaxed,neff2017real}). Although bulk media are investigated in detail, the study of the reflective/refractive properties at the boundary of such metamaterials is far from being well understood. Good knowledge of the reflective and transmittive  properties of such interfaces could be a key point for the conception of metamaterials systems, which would completely transform the idea we currently have about reflection and transmission of elastic waves at the interface between two solids. It is for this reason that many authors convey their research towards so-called ``metasurfaces'' (\cite{xie2014wavefront, liang2018wavefront, li2014acoustic}), i.e. relatively thin layers of metamaterials whose microstructure is able to interact with the incident wave-front in such a way that the resulting reflection/transmission patterns exhibit exotic properties, such as total reflection or total transmission, conversion of a bulk incident wave in interface waves, etc.}

Notwithstanding the paramount importance these metasurfaces may have for technological advancements in the field of noise absorption or stealth, they show limitations in the sense that they work for relatively small frequency ranges, for which the wavelength of the incident wave is comparable to the thickness of the metasurface itself. This restricts the range of applicability of such devices, above all for what concerns low frequencies which would result in very thick metasurfaces.

In this paper, we choose a different approach for modeling the reflective and diffractive  properties of an interface which separates a bulk homogeneous material from a bulk metamaterial. This interface does not itself contain any internal structure, but its refractive properties can be modulated by acting on the relative stiffnesses of the two materials on each side. The homogeneous material is modeled via a classical linear-elastic Cauchy model, while the metamaterial is described by the linear relaxed micromorphic model, an enriched continuum model which already proved its effectiveness in the description of the bulk behavior of certain metamaterials (\cite{madeo2016first,madeo2017modeling,madeo2017relaxed}).

We are able to clearly show that when the homogeneous material is ``stiffer'' than the considered metamaterial, zones of very high (sometimes total) transmission can be found both at low and high frequencies. More precisely, we find that high-frequency total transmission is discriminated by a critical angle, beyond which total transmission gradually shifts towards total reflection. Engineering systems of this type could be fruitfully exploited for the conception of wave filters and polarizers, for non-destructive evaluation or for selective cloaking.

On the other hand, we show that when the homogeneous system is ``softer'' than the metamaterial, broadband total reflection can be achieved for almost all frequencies and angles of incidence. This could be of paramount importance for the conception of wave screens able to isolate from noise and/or vibrations. We are also able to show that such total reflection phenomena are related to the onset of classical Stoneley interface waves\footnote{Interface waves propagating at the interface between an elastic solid and air are called Rayleigh waves, after Lord Rayleigh, who was the first to show their existence (\cite{rayleigh1885onwaves}). Interface waves propagating at the surface between two solids are called Stoneley waves.} at low frequencies (\cite{stoneley1924elastic}) and of new microstructure-related interface waves at higher frequencies. 

We explicitly remark that no precise microstructure is targeted in this paper, since the presented results could be generalized to any specific metamaterial without changing the overall results. This is due to the fact that the properties we unveil here only depend on the ``relative stiffnesses'' of the considered media and not on the absolute stiffness of the metamaterial itself.

The considerations drawn in this paper thus open the way for the conception of new metastructures for wave front manipulation with all the possible applications aforementioned applications. 

\subsection{Notational agreement}
Let $\R^{3\times 3}$ be the set of all real $3\times 3$ second order tensors (matrices) which we denote by capital letters. A simple and a double contraction between tensors of any suitable order is denoted by $\cdot$ and $:$ respectively, while the scalar product of such tensors by $\left\langle \cdot, \cdot \right\rangle$.\footnote{For example, $(A\cdot v)_i=A_{ij}v_j, (A\cdot B)_{ik}=A_{ij}B_{jk}, A:B=A_{ij}B_{ji},(C\cdot B)_{ijk}=C_{ijp}B_{pk},(C:B)_{i}=C_{ijp}B_{pj},\left\langle v,w\right\rangle = v\cdot w = v_i w_i, \left\langle A, B \right\rangle = A_{ij}B_{ij}$, etc.} The Einstein sum convention is implied throughout this text unless otherwise specified. The standard Euclidean scalar product on $\R^{3 \times 3}$ is given by $\left\langle X, Y \right\rangle_{\R^{3\times 3}} = \tr (X\cdot Y^T)$ and consequently the Frobenius tensor norm is $\norm{X}^2 = \left\langle X, X\right\rangle_{\R^{3\times 3}}$. From now on, for the sake of notational brevity, we omit the scalar product indices $\R^3, \R^{3\times 3}$ when no confusion arises. The identity tensor on $\R^{3\times 3}$ will be denoted by $\idmat$; then, $\tr(X)=\left\langle X, \idmat \right\rangle$.

Consider a body which occupies a bounded open set $B_L \subset \R^3$ and assume that the boundary $\d B_L$ is a smooth surface of class $C^2$. An elastic material fills the domain $B_L$ and we denote by $\Sigma$ any material surface embedded in $B_L$. The outward unit normal to $\d B_L$ will be denoted by $\nu$ as will the outward unit normal to a surface $\Sigma$ embedded in $B_L$ (see Fig. \ref{fig:BL}). 
\begin{figure}[H]
	\centering
	\includegraphics[scale=0.4]{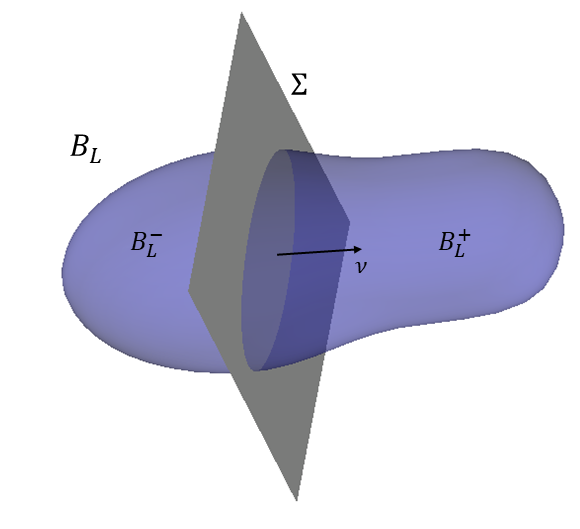}
	\caption{\small Schematic representation of the body $B_L$, the surface  $\Sigma$ and the sub-bodies $B_L^-$ and $B_L^+$.}
	\label{fig:BL}
\end{figure}
Given a field $a$ defined on the surface $\Sigma$, we set 
\begin{equation}
[[a]] = a^+ - a^-,
\end{equation}
which defines a measure of the jump of $a$ through the material surface $\Sigma$, where 
\begin{equation}
[\cdot ]^{-} := \lim_{\substack{x \in B_L^{-}\setminus \Sigma \\ x \to \Sigma}} [\cdot ], \quad [\cdot ]^{+} := \lim_{\substack{x \in B_L^{+}\setminus \Sigma \\ x \to \Sigma}} [\cdot ],
\end{equation}
with $B_L^{-}, B_L^{+}$ being the two subdomains which result from splitting $B_L$ by the surface $\Sigma$ (see again Fig. \ref{fig:BL}). 

The Lebesgue spaces of square integrable functions, vectors or tensors fields on $B_L$ with values on $\R, \R^3, \R^{3\times 3}$ respectively, are denoted by $L^2(B_L)$. Moreover we introduce the standard Sobolev spaces\footnote{The operators $\nabla$, $\curl$ and $\div$ are the classical gradient, curl and divergence operators. In symbols, for a field $u$ of any order, $(\nabla u)_i=u_{,i}$, for a vector field $v$, $(\curl v)_i = \eps_{ijk}v_{k,j}$ and for a field $w$ of order greater than $1$, $\div w = w_{i,i}$. }
\begin{align}
&H^1(B_L)=\left\lbrace u \in L^2(B_L)| \,\nabla u \in L^2(B_L), \norm{u}^2_{H^1(B_L)}:=\norm{u}^2_{L^2(B_L)} + \norm{\nabla u}^2_{L^2(B_L)}\right\rbrace, \nonumber \\ 
&H(\curl ;B_L)=\left\lbrace v \in L^2(B_L)|\, \curl v \in L^2(B_L), \norm{v}^2_{H(\curl ;B_L)}:=\norm{v}^2_{L^2(B_L)} + \norm{\curl v}^2_{L^2(B_L)}\right\rbrace, \nonumber \\
&H(\div ;B_L)=\left\lbrace v \in L^2(B_L)|\, \div v \in L^2(B_L), \norm{v}^2_{H(\div ;B_L)}:=\norm{v}^2_{L^2(B_L)} + \norm{\div v}^2_{L^2(B_L)}\right\rbrace, \label{SobolevSpaces}
\end{align}
of functions $u$ and vector fields $v$ respectively.


For vector fields $v$ with components in $H^1(B_L)$ and for tensor fields $P$ with rows in $H(\curl ; B_L)$ (resp. $H(\div ; B_L)$), i.e.
\begin{equation}
v=(v_1,v_2,v_3)^T, \text{ } v_i \in H^1(B_L), \quad P=(P_1^T,P_2^T,P_3^T)^T, \text{ } P_i\in H(\curl ; B_L) \text{ resp. }P_i \in H(\div ; B_L), \quad i=1,2,3, 
\end{equation}
we define 
\begin{align} 
\nabla v &=((\nabla v_1)^T,(\nabla v_2)^T,(\nabla v_3)^T)^T,\nonumber \\
 \Curl P &= ((\curl P_1)^T,(\curl P_2)^T,(\curl P_3)^T)^T, \label{eq:nablacurldiv}
 \\ \Div P &=(\div P_1, \div P_2, \div P_3)^T. \nonumber
\end{align}
The corresponding Sobolev spaces are denoted by $H^1(B_L)$, $H(\Div ; B_L)$, $H(\Curl ; B_L)$. 

\section{Governing equations and energy flux}\label{sec:governingeqs}
Let $t_0>0$ be a fixed time and consider a bounded domain $B_L \subset \R^3$. The action functional of the system at hand is defined as 
\begin{equation}\label{eq:actionfunct}
\mathcal{A}=\int_0^{t_0} \int_{B_L} (J-W) dX dt,
\end{equation}
where $J$ is the kinetic and $W$ the potential energy of the system.
\subsection{Governing equations for the classical linear elastic Cauchy model}
Let $u$ be the classical macroscopic displacement field and $\rho$ the macroscopic mass density. Then, the kinetic energy for a classical, linear-elastic, isotropic Cauchy medium takes the form\footnote{Here and in the sequel, we denote by the subscript $,	t$ the partial derivative with respect to time.} 
\begin{equation}\label{eq:Cauchykinetic}
J=\frac{1}{2}\rho \norm{u_{,t}}^2. 
\end{equation}
The linear elastic strain energy density  for the isotropic case is given by 
\begin{equation}\label{eq:Cauchyenergy}
W = \mu \norm{\sym \nabla u}^2 + \frac{\lambda}{2}(\tr(\sym \nabla u))^2,
\end{equation}
where $\mu$ and $\lambda$ are the classical Lam\'e parameters. Applying the least action principle allows us to derive the \emph{equations of motion in strong form} which, expressed equivalently in compact and index form, are respectively given by 
\begin{equation}\label{eq:Cauchystrong}
\rho \, u_{,tt} = \Div \sigma, \qquad \rho \, u_{i,tt} = \sigma_{ij,j},
\end{equation}
where 
\begin{equation}\label{eq:Cauchystress}
\sigma = 2 \mu \sym \nabla u + \lambda \tr(\sym \nabla u)\idmat, \qquad \sigma_{ij} = \mu(u_{i,j}+u_{j,i})+\lambda u_{k,k}\delta_{ij},
\end{equation}
is the classical symmetric Cauchy stress tensor for isotropic materials.


\subsection{Governing equations for the relaxed micromorphic model}
The kinetic energy in the isotropic relaxed micromorphic model takes the following form \cite{madeo2016complete,madeo2016reflection}
\begin{equation}\label{eq:relaxedkinetic}
J = \frac{1}{2}\rho \norm{u_{,t}}^2+\frac{1}{2}\eta\norm{P_{,t}}^2,
\end{equation}
where $u$ is the macroscopic displacement field and $P\in \R^{3\times 3}$ is the non-symmetric micro-distortion tensor which accounts for independent micro-motions at lower scales, $\rho$ is the apparent macroscopic density and $\eta$ is the micro-inertia density.

The strain energy for an isotropic relaxed micromorphic medium is given by \cite{madeo2016complete,madeo2016reflection}
\begin{align}
W = &\,\mue \norm{\sym(\nabla u - P)}^2 +\frac{\lame}{2}(\tr(\nabla u - P))^2 + \muc \norm{\skew(\nabla u - P)}^2 \nonumber \\
&+ \mumic \norm{\sym P}^2+\frac{\lammic}{2}(\tr P)^2 + \frac{\mue \Lc}{2}\norm{\Curl P}^2. \label{eq:relaxedenergy}
\end{align}
Minimizing the action functional (i.e. imposing $\delta \mathcal{A} = 0$, integrating by parts a suitable number of times and taking arbitrary variations $\delta u$ and $\delta P$ for the kinematic fields), allows us to obtain the \emph{strong form of the bulk equations of motion} (\cite{ghiba2014relaxed,madeo2015wave,madeo2014band,neff2014unifying})
\begin{equation}\label{eq:relaxedeqns}
\begin{array}{clccl}
\rho \, u_{tt} = & \Div \widetilde{\sigma}, & &\quad \rho \, u_{i,tt} = & \widetilde{\sigma}_{ij,j},\\
\eta \,P_{tt} = &  \widetilde{\sigma}-s-\Curl m, & & \quad \eta \, P_{ij,tt}  = &\widetilde{\sigma}_{ij} - s_{ij} - m_{ik,p}\eps_{jkp},
\end{array}
\end{equation}
where  
\begin{align}
\widetilde{\sigma}&=2 \mue \sym(\nabla u - P) + \lame \tr (\nabla u - P)\,\idmat + 2\muc \skew(\nabla u - P), \nonumber\\
s&=2\mumic \sym P + \lammic \left(\tr P\right) \, \idmat, \label{eq:relaxedRHS} \\
m&= \mue \Lc \Curl P, \nonumber
\end{align}
$\eps_{jkp}$ is the Levi-Civita tensor and the equivalent index form of the introduced quantities reads
\begin{align}
\widetilde{\sigma}_{ij} &= \mue (u_{i,j} - P_{ij} + u_{j,i} - P_{ji}) + \lame (u_{k,k} - P_{kk})\delta_{ij} + \muc(u_{i,j} - P_{ij} -u_{j,i} + P_{ji}),\nonumber\\
s_{ij} &= \mumic(P_{ij}+P_{ji}) + \lammic P_{kk}\delta_{ij}, \\
m_{ik} &= \mue \Lc P_{ia,b}\eps_{kba}.\nonumber
\end{align}
Here, $m$ is the moment stress tensor, $\muc\geq 0$ is called the Cosserat couple modulus and $L_c \geq 0$ is a characteristic length scale.

\subsection{Conservation of total energy and energy flux}
The mechanical system we are considering is conservative and, therefore, the energy must be conserved in the sense that the following differential form of a continuity equation must hold:
\begin{equation}
E_{,t}+\div H=0,\label{EnnergyConservation}
\end{equation}
where $E=J+W$ is the total energy of the system and $H$ is the energy flux vector, whose explicit expression is computed in the following subsections for a classical Cauchy medium and a relaxed micromorphic one. 


The conservation of energy \eqref{EnnergyConservation} plays an important role when considering a surface of discontinuity between two continuous media, i.e. it prescribes continuity of the normal component of the energy flux $H$ across the considered interface. The tangent part of the energy flux vector $H$ is not subjected to the same restriction, so that we can observe the onset of Stoneley waves along the interface. Such interface waves do not have to satisfy the conservation of energy at the interface. On the other hand, when considering waves with non-vanishing normal component, it is necessary to check that they contribute to the normal part of the energy flux in such a way that it is conserved at the interface. The definition of the normal part of the flux for different waves will allow us to introduce the concepts of reflection and transmission coefficients as a measure of the partition of the energy of the incident wave between reflected and transmitted waves. Nevertheless, one must keep in mind that such reflection and transmission coefficients do not contain any information about the possible presence of Stoneley waves.

\subsubsection{Classical Cauchy medium}
In the case of a Cauchy medium we have, differentiating expressions \eqref{eq:Cauchykinetic} and \eqref{eq:Cauchyenergy} with respect to time  and using definition \eqref{eq:Cauchystress}
\begin{equation*}
E_{,t}=J_{,t}+W_{,t} = \rho \left\langle u_{,t},u_{,tt}\right\rangle+\left\langle(2\mu \sym \nabla u + \lambda \tr (\sym \nabla u)\idmat),\sym \nabla (u_{,t})\right\rangle = \rho \left\langle u_{,t},u_{,tt}\right\rangle + \left\langle \sigma,\sym \nabla (u_{,t})\right\rangle.
\end{equation*}
We now replace $\rho \, u_{,tt}$ from the equations of motion \eqref{eq:Cauchystrong} and we use the fact that, given the symmetry of the Cauchy stress tensor $\sigma$, $\left\langle \sigma, \sym \nabla (u_{,t})\right\rangle = \left\langle \sigma, \nabla (u_{,t})\right\rangle = \Div (\sigma \cdot u_{,t})-\Div \sigma \cdot u_{,t}$ to get
\begin{equation}\label{eq:EdotCauchy}
E_{,t}=\Div \sigma \cdot u_{,t}+\Div(\sigma\cdot u_{,t})-\Div \sigma \cdot u_{,t} = \Div (\sigma \cdot u_{,t}).
\end{equation}
Thus, by comparing \eqref{eq:EdotCauchy} to the conservation of energy \eqref{EnnergyConservation}, we can deduce that the energy flux in a Cauchy continuum is given by (see e.g. \cite{achenbach1973wave})
\begin{equation}\label{eq:Cauchyflux}
H = -\sigma \cdot u_{,t}, \qquad H_k = - \sigma_{ik} u_{k,t}.
\end{equation}

\subsubsection{Relaxed micromorphic model}

In the case of a relaxed micromorphic medium we have, differentiating equations \eqref{eq:relaxedkinetic} and \eqref{eq:relaxedenergy} with respect to time\footnote{Let $\psi$ be a vector field and $A$ a second order tensor field. Then
	\begin{equation}\label{eq:vectorid1}
	\langle \nabla \psi, A\rangle  = \Div (\psi \cdot A) - \langle \psi, \Div A \rangle.
	\end{equation}
	Taking $\psi = u_{,t}$ and $A = \sigmatil$ we have
	\begin{equation}\label{eq:vectorid2}
	\left\langle \nabla u_{,t}, \sigmatil \right\rangle = \Div(u_{,t}\cdot \sigmatil) - \left\langle u_{,t},\Div \sigmatil \right\rangle.
	\end{equation}
	Furthermore, we have the following identity
	\begin{equation}\label{eq:vectorid3}
	\left\langle m, \Curl P_{,t} \right\rangle = \Div\left( (m^T\cdot P_{,t}):\eps\right)+\left\langle \Curl m , P_{,t}\right\rangle,
	\end{equation}
	which follows from the identity $\div( v \times w)=w \cdot\curl v - v \cdot \curl w$, where $v, w$ are suitable vector fields, $\times$ is the usual vector product and $:$ is the double contraction between tensors.} 
\begin{align}
E_{,t}=&\,\rho \left\langle u_{,t},u_{,tt}\right\rangle + \eta \left\langle P_{,t},P_{,tt}\right\rangle + \left\langle2\mue \sym (\nabla u - P),\sym(\nabla u_{,t}-P_{,t})\right\rangle + \left\langle \lame \tr (\nabla u - P) \idmat,\nabla u_{,t}-P_{,t}\right\rangle \nonumber \\
&+\left\langle 2\muc \skew(\nabla u - P),\skew (\nabla u_{,t}-P_{,t})\right\rangle + \left\langle 2\mumic \sym P, \sym P_{,t}\right\rangle + \left\langle \lammic (\tr P)\idmat,P_{,t}\right\rangle  \nonumber\\
&+ \left\langle \mue \Lc \Curl P, \Curl P_{,t}\right\rangle,\label{eq:EdotRMM1} 
\end{align}
or equivalently,\footnote{We explicitly recall that given a symmetric tensor $S$, a skew-symmetric tensor $A$ and a generic tensor $C$, we have $\left\langle S, C\right\rangle = \left\langle S, \sym C\right\rangle$ and $\left\langle A, C\right\rangle = \left\langle A, \skew C\right\rangle$.}  suitably collecting terms and using definitions \eqref{eq:relaxedRHS} for $\widetilde{\sigma}, s$ and $m$:
\begin{align}
E_{,t}=&\,\rho \left\langle u_{,t},u_{,tt}\right\rangle + \eta \left\langle P_{,t},P_{,tt}\right\rangle + \left\langle2\mue \sym (\nabla u - P)+\lame \tr (\nabla u - P),\sym(\nabla u_{,t}-P_{,t})\right\rangle \nonumber\\
&+\left\langle 2\muc \skew(\nabla u - P),\skew (\nabla u_{,t}-P_{,t})\right\rangle + \left\langle 2\mumic \sym P+\lammic (\tr P)\idmat, \sym P_{,t}\right\rangle + \left\langle \mue \Lc \Curl P, \Curl P_{,t}\right\rangle \nonumber\\
=&\,\rho \left\langle u_{,t},u_{,tt}\right\rangle + \eta \left\langle P_{,t},P_{,tt}\right\rangle + \left\langle \widetilde{\sigma},\nabla u_{,t}\right\rangle - \left\langle \widetilde{\sigma}-s,P_{,t}\right\rangle + \left\langle m,\Curl P_{,t}\right\rangle.\label{eq:EdotRMM2} 
\end{align}
We can now replace $\rho \, u_{,tt}$ and $\eta \, P_{,tt}$ by the governing equations \eqref{eq:relaxedeqns} and use \eqref{eq:vectorid2} and \eqref{eq:vectorid3} to finally get
\begin{align}
\hspace{-0.4cm}
E_{,t}&=\left\langle u_{,t},\Div \widetilde{\sigma}\right\rangle + \left\langle P_{,t}, \widetilde{\sigma} - s - \Curl m\right\rangle + \Div\left(u_{,t} \cdot \widetilde{\sigma}\right) -\left\langle u_{,t},\Div \widetilde{\sigma} \right\rangle- \left\langle \widetilde{\sigma}-s,P_{,t}\right\rangle + \left\langle m, \Curl P_{,t}\right\rangle \nonumber
\\
&=\Div (\widetilde{\sigma}^T\cdot u_{,t})-\left\langle \Curl m , P_{,t}\right\rangle + \Div \left(\left(m^T\cdot P_{,t}\right):\eps\right) + \left\langle\Curl m, P_{,t}\right\rangle = \Div\left(\widetilde{\sigma}^T\cdot u_{,t} + \left(m^T\cdot P_{,t}\right):\eps\right). \label{eq:EdotRMM3}
\end{align}
Thus, by comparison of \eqref{eq:EdotRMM3} with \eqref{EnnergyConservation}, the energy flux for a relaxed micromorphic medium is given by 
\begin{equation}\label{eq:relaxedflux}
H =-\widetilde{\sigma}^T\cdot u_{,t} - \left(m^T\cdot P_{,t}\right):\eps, \qquad H_{k} = -u_{i,t} \widetilde{\sigma}_{ik} - m_{ih}P_{ij,t}\eps_{jhk},
\end{equation}
where $:$ is the double contraction between tensors. 
\section{Boundary conditions}\label{sec:BoundCond}
\subsection{Boundary conditions on an interface between two classical Cauchy media}\label{sec:BoundCondCC}
As it is well known (see e.g. \cite{achenbach1973wave,graff1975wave,madeo2016reflection}), the boundary conditions which can be imposed at an interface between two Cauchy media are continuity of displacement or continuity of force.\footnote{It is also well known that if no surface forces are present at the considered interface, imposing continuity of displacement implies continuity of forces and vice versa. Both such continuity conditions must then be satisfied at the interface.} For the displacement, this means
\begin{equation}\label{displcont}
[[u]] =0 \Rightarrow u^{-} = u^{+},
\end{equation}
where $u^{-}$ is the macroscopic displacement on the ``minus'' side (the $x_1<0$ half-plane) and $u^{+}$ is the macroscopic displacement on the ``plus'' side (the $x_1>0$ half-plane).

As for the jump of force we have
\begin{equation}\label{jumpforcevec}
[[f]] = 0 \Rightarrow f^{-} = f^{+},
\end{equation}
where $f^{-}$ and $f^{+}$ are the surface force vectors on the ``minus'' and on the ``plus'' side, respectively. We recall that in a Cauchy medium, $f = \sigma \cdot \nu$, $\nu$ being the outward unit normal to the surface and $\sigma$ being the Cauchy stress tensor given by \eqref{eq:Cauchystress}.
\subsection{Boundary conditions on an interface between a classical Cauchy medium and a relaxed micromorphic medium}\label{sec:BoundCondCRM}
In this article we will be imposing two kinds of boundary conditions between a Cauchy and a relaxed micromorphic medium (see \cite{madeo2016reflection} for a detailed discussion). The first kind is that of a free microstructure, i.e. we allow the microstructure to vibrate freely. The second kind is that of a fixed microstructure, i.e. we do not allow any movement of the microstructure at the interface. For a derivation of these types of boundary conditions see \cite{madeo2016reflection}. 

As mentioned, we impose continuity of displacement, i.e. 
\begin{equation}
u^{-} = u^{+},
\end{equation}
where now the ``plus'' side is occupied by the relaxed micromorphic medium.\footnote{In the following, it will be natural to collect variables of the relaxed medium in two vectors $v_1$ and $v_2$, so that the components of the displacement $u^{+}$ will be denoted by $u^{+}=(v^1_1,v^1_2,v^1_1)^T$, see equations \eqref{eq:v1}, \eqref{eq:v2}.} Imposing continuity of displacement also implies continuity of force,\footnote{This fact can be deduced from the minimization of the action functional. Indeed, imposing the variation to be zero, we find the boundary condition $[[\langle F,\delta u\rangle]]=0\Rightarrow f\, u^{+} - t\, u^{-} = 0$, which implies $f=t$ since we impose continuity of displacement.} i.e. 
\[
f = t,
\]
where $f = \sigma \cdot \nu$ is the surface force calculated on the Cauchy side and the force for the relaxed micromorphic model is given by 
\begin{equation}\label{eq:forcemicromorphic}
t_i = \widetilde{\sigma}_{ij}\nu_j,
\end{equation}
where $\nu$ is the outward unit normal vector to the interface and $\widetilde{\sigma}$ is given in \eqref{eq:relaxedRHS}.

The least action principle provides us with another jump duality condition for the case of the connection between a Cauchy and a relaxed micromorphic medium. This extra jump condition specifies what we call ``free microstructure'' or ``fixed microstructure'' (see \cite{madeo2016reflection}).

In order to define the two types of boundary conditions we are interested in, we need the concept of double force $\tau$ which is the dual quantity of the micro-distortion tensor $P$ and is defined as \cite{madeo2016reflection}
\begin{equation}\label{eq:doubleforcedef}
\tau=-m\cdot \eps \cdot \nu,\quad \tau_{ij}=-m_{ik}\eps_{kjh}\nu_{h},
\end{equation}
where the involved quantities have been defined in \eqref{eq:relaxedRHS}.

\subsubsection{Free microstructure}
In this case, the macroscopic displacement is continuous while the microstructure of the relaxed micromorphic medium is free to move at the interface (\cite{madeo2016first,madeo2016reflection,madeo2017relaxed}). Leaving the interface free to move means that $P$ is arbitrary, which, on the other hand, implies the double force $\tau$ must vanish. We have then
\begin{equation}\label{eq:BCfreemicro}
[[u_i]]=0, \quad t_i -f_i = 0,\quad \tau_{ij} = 0, \quad i=1,2,3, \text{ } j=2,3.
\end{equation}
Figure \ref{fig:freemicro} gives a schematic representation of this boundary condition.
\\
\begin{figure}[H]
	\begin{centering}	
		\begin{picture}(248,125)
		
		\put(0,10){\line(0,1){124}}
		\put(0,10){\line(1,0){248}}
		\put(0,134){\line(1,0){248}}
		\linethickness{0.7 mm}
		\put(124,10){\line(0,1){124}}
		\thinlines
		\put(248,10){\line(0,1){124}}
		\multiput(131,16)(11,0){11}{
			\multiput(0,0)(0,11){11}{\circle{7.2}}}
		
		\put(35,140){\footnotesize{Cauchy medium $\ominus$}}		
		\put(126,140){\footnotesize{relaxed micromorphic medium $\oplus$}}

		\end{picture}
		\par\end{centering}
	
	\caption{\small Schematic representation of a macro internal clamp with free microstructure at a Cauchy/relaxed-micromorphic interface \cite{madeo2016first,madeo2016reflection}.}
	\label{fig:freemicro}
\end{figure}
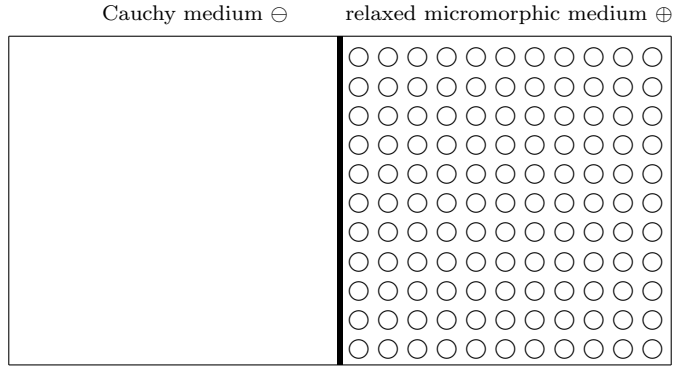

%
\subsubsection{Fixed microstructure}
This is the case in which we impose that the microstructure on the relaxed micromorphic side does not vibrate at the interface. The boundary conditions in this case are (see \cite{madeo2016reflection})\footnote{Let us remark again that the first condition (continuity of displacement) implies the second one when no surface forces are applied at the interface. On the other hand, imposing the tangent part of $P$ to be equal to zero implies that the double force $\tau$ is left arbitrary.}\footnote{We remark here that only the tangent part of the double force in \eqref{eq:BCfreemicro} or of the micro-distortion tensor in \eqref{eq:BCfixedmicro} must be assigned (\cite{neff2015relaxed,neff2014unifying}). This is peculiar of the relaxed micromorphic model and is related to the fact that only $\Curl P$ appears in the energy. In a standard Mindlin-Eringen model, where the whole $\nabla P$ appears in the energy, the whole double force $\tau$ (or alternatively the whole micro-distortion tensor $P$) must be assigned at the interface. Finally, in an internal variable model (no derivatives of $P$ in the energy), no conditions on $P$ or $\tau$ must be assigned at the considered interface.}
\begin{equation}\label{eq:BCfixedmicro}
[[u_i]]=0,\quad t_i-f_i=0,\quad P_{ij} = 0,\quad i=1,2,3, \text{ } j=2,3.
\end{equation}
Observe that in equations \eqref{eq:BCfreemicro} and \eqref{eq:BCfixedmicro} the components tangent to the interface do not have to be assigned when considering a relaxed micromorphic medium. This is explained in detail in \cite{madeo2016reflection,neff2014unifying,neff2015relaxed}, where the peculiarities of the relaxed micromorphic model are presented. In Figure \ref{fig:fixedmicro} we give a schematic representation of the realization of this boundary condition between a homogeneous material and a metamaterial.\\

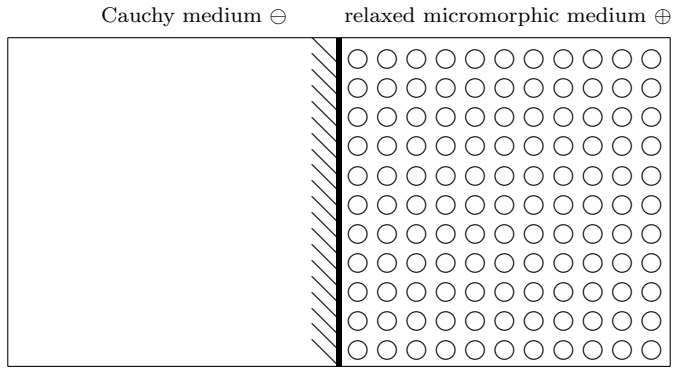
\begin{figure}[H]
	\begin{centering}	
		\begin{picture}(248,125)
		
		\multiput(124,10)(0,6){20}{\line(-1,1){10}} 
		\put(0,10){\line(0,1){124}}

		\put(0,10){\line(1,0){248}}		
		\put(0,134){\line(1,0){248}}
		\linethickness{0.7 mm}
		\put(124,10){\line(0,1){124}}
		\thinlines
		\put(248,10){\line(0,1){124}}
		\multiput(131,16)(11,0){11}{
			\multiput(0,0)(0,11){11}{\circle{7.2}}}
		
		\put(35,140){\footnotesize{Cauchy medium $\ominus$}}		
		\put(126,140){\footnotesize{relaxed micromorphic medium $\oplus$}}
		
		\end{picture}
		\par\end{centering}
	
	\caption{\small Schematic representation of a macro internal clamp with fixed microstructure at a Cauchy/relaxed-micromorphic interface \cite{madeo2016first,madeo2016reflection}.}
	\label{fig:fixedmicro}
\end{figure}

We explicitly remark (and we will show this fact in all detail in the remainder of this paper) that the boundary condition of Fig. \ref{fig:freemicro} is the only one which allows to obtain an equivalent Cauchy/Cauchy system when considering low frequencies. Indeed, in this case, since the tensor $P$ is left free, it can adjust itself in order to recover a Cauchy medium in the homogenized limit. On the other hand, the boundary condition of Fig. \ref{fig:fixedmicro} imposes an artificial value on $P$ along the interface, so that the effect of the microstructure is forced to be present in the considered system. It is for this reason that a Cauchy/Cauchy interface cannot be recovered as a homogenized limit of the system shown in Fig. \ref{fig:fixedmicro}.
\section{Wave propagation, reflection and transmission at an interface between two Cauchy media}
We now want to study wave propagation, reflection and transmission at the interface between two Cauchy media in two space dimensions. To that end, we agree on the following conventions. 

When we discuss reflection and transmission, we assume that the surface of discontinuity (the interface between the two media) from which the wave reflects and refracts is the $x_2$ axis ($x_1=0$). Furthermore, we assume that the incident wave hits the interface at the origin $O (0,0)$.

\begin{figure}[H]
	\centering
	\includegraphics[scale=0.4]{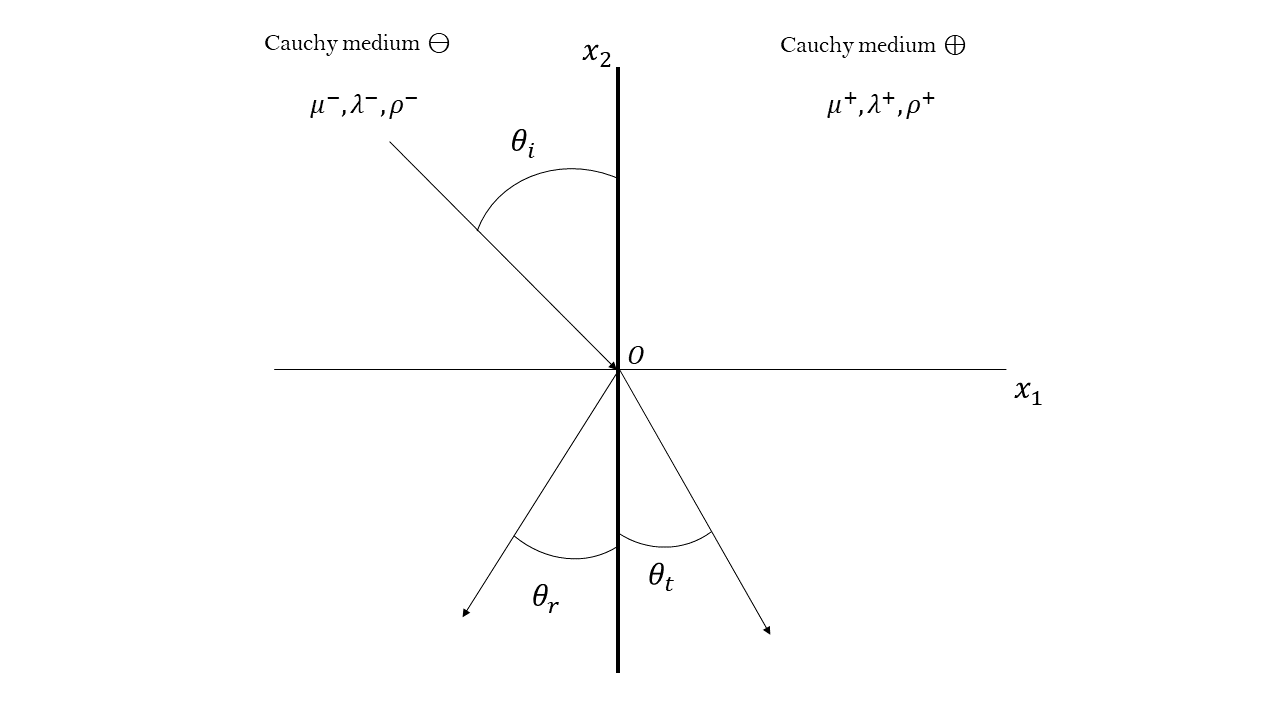}
	\caption{\small Schematic representation of a Cauchy/Cauchy interface with an incident, a reflected and a transmitted wave for the case of an incident out-of-plane SH wave. The indices $i,r,t$ stand for ``incident'', ``reflected'' and ``transmitted'' respectively.}
	\label{fig:CauchyCauchy}
\end{figure}
Incident waves propagate from $-\infty$ in the $x_1<0$ half-plane towards the interface, reflected waves propagate from the interface towards $-\infty$ in the $x_1<0$ half-plane and refracted (or, equivalently, transmitted) waves propagate from the interface towards $+\infty$ in the $x_1>0$ half-plane. Remark that, since we consider $2$D wave propagation, the incident wave can hit the interface at an arbitrary angle. In Figure \ref{fig:CauchyCauchy} we graphically present the schematics of the reflection and transmission of elastic waves at a Cauchy/Cauchy interface for the simple case of an incident out-of-plane SH wave.

As is classical, we will consider both in-plane (in the $(x_1 x_2)-$ plane) and out-of-plane (along the $x_3$ direction) motions. Nevertheless, all the considered components of the displacement, namely $u_1, u_2$ and $u_3$ will only depend on the $x_1, x_2$ space components (plane wave hypothesis). We hence write
\begin{equation}\label{eq:u2D}
u = (u_1(x_1,x_2,t),u_2(x_1,x_2,t),u_3(x_1,x_2,t))^T.
\end{equation}

As will be evident, depending on the type of wave (longitudinal, SV shear or SH shear)\footnote{Following classical nomenclature (\cite{achenbach1973wave}) we call longitudinal (or L) waves those waves whose displacement vector is in the same direction of the wave vector. Moreover SV waves are shear waves whose displacement vectors is orthogonal to the wave vector and lies in the $(x_1x_2)-$plane. SH waves are shear waves whose displacement vector is orthogonal to the wave vector and lies in a plane orthogonal to the $(x_1x_2)-$plane.}  some components of these fields will be equal to zero. 

We now make a small digression on wave propagation in classical Cauchy media. These results are of course well known (see e.g. \cite{achenbach1973wave,graff1975wave}), however, we present them here in detail following our notation, so that we can carry all the results and considerations over to the relaxed micromorphic model as a natural extension.

\subsection{Wave propagation in Cauchy media}\label{sec:Cauchywaveprop}

We start by writing the governing equations \eqref{eq:Cauchystrong} for the special case where the displacement only depends on the in-plane space variables $x_1$ and $x_2$. Plugging $u$ as in \eqref{eq:u2D} into \eqref{eq:Cauchystrong} gives the following system of PDEs:

  \begin{align}
 &\left.
 \begin{tabular}{cc}
 $\rho \, u_{1,tt}$ &$= (2\mu + \lambda)u_{1,11}  + (\mu + \lambda) u_{2,12} + \mu \, u_{1,22}$  \\
 $\rho \, u_{2,tt}$ &$= (2\mu + \lambda)u_{2,22}  + (\mu + \lambda) u_{1,12}+ \mu \, u_{2,11}$
 \end{tabular}
 \right\}\label{eq:comp12}\\
 &\hspace{2.8mm} \rho \, u_{3,tt} \hspace{3.6mm} = \mu (u_{3,11} + u_{3,22}) \label{eq:comp3} .
 \end{align}
We remark that the first two equations \eqref{eq:comp12} are coupled, while the third \eqref{eq:comp3} is uncoupled from the first two. 

We now formulate the plane wave ansatz, according to which the displacement vector $u$ takes the same value at all points lying on the same orthogonal line to the	 $(x_1x_2)-$plane (no dependence on $x_3$). Moreover, we also consider that the displacement field is periodic in space. In symbols, the plane-wave ansatz reads
\begin{align}
u &= \widehat{\psi}\, e^{i(\left\langle x,k \right\rangle- \omega t)} \hspace{1.5mm} =\widehat{\psi}\, e^{i(x_1k_1 + x_2k_2- \omega t)},\label{eq:planeansatzCauchy}\\
u_3 &= \widehat{\psi}_3\, e^{i\left(\langle x,k\rangle - \omega t\right)} = \widehat{\psi}_3 \,e^{i\left(x_1k_1 + x_2 k_2- \omega t\right)}, \label{eq:planewaveCauchy3}
\end{align}
where $\widehat{\psi}= (\widehat{\psi}_1,\widehat{\psi}_2)^T$ is the vector of amplitudes, $k = (k_1,k_2)^T$ is the wave-vector, which fixes the direction of propagation of the considered wave and $x=(x_1,x_2)^T$ is the position vector. Moreover, $\widehat{\psi}_3$ is a scalar amplitude for the third component of the displacement. We explicitly remark that in equation \eqref{eq:planeansatzCauchy} we consider (with a slight abuse of notation) $u=(u_1,u_2)^T$ to be the vector involving only the in-plane displacement components $u_1$ and $u_2$ which are coupled via equations \eqref{eq:comp12}. We start by considering the first system of coupled equations and we plug the plane-wave ansatz \eqref{eq:planeansatzCauchy} into \eqref{eq:comp12}, so obtaining
\begin{align}
(\omega^2 - c_L^2 k_1^2 -c_S^2 k_2^2)\widehat{\psi}_1 - c_V^2 k_1 k_2 \widehat{\psi}_2 &=0, \nonumber \\
-c_V^2 k_1 k_2 \widehat{\psi}_1 + (\omega^2 -c_L^2 k_2^2 - c_S^2 k_1^2)\widehat{\psi}_2 &=0, \label{eq:systemCauchy}
\end{align}
where we made the abbreviations
\begin{equation}\label{eq:speedsCauchy}
c_L^2 = \frac{2 \mu + \lambda}{\rho}, \quad c_S^2 = \frac{\mu}{\rho}, \quad c_V^2 = c_L^2 - c_S^2 = \frac{\mu + \lambda}{\rho},
\end{equation}
where $\mu, \lambda$ are the Lam\'e parameters of the material and $\rho$ is its apparent density. This system of algebraic equations can be written compactly as $A \cdot \widehat{\psi} =0 $, where $A$ is the matrix of coefficients
\begin{equation}\label{eq:coeffsmatrixCauchy}
A = \left(\begin{array}{cc}
\omega^2 - c_L^2 k_1^2 - c_S^2 k_2^2 & -c_V^2 k_1 k_2 \\
-c_V^2 k_1 k_2  & \omega^2 -c_L^2 k_2^2 - c_S^2 k_1^2
\end{array}\right). 
\end{equation}
For $A \cdot \widehat{\psi} = 0$ to have a solution other than the trivial one, we impose $\det A = 0$. We have (see Appendix \ref{appendixCauchy} for a detailed calculation of this expression)
\begin{equation}\label{eq:detACauchy}
\det A =\frac{\left((2\mu + \lambda) \left(k_1^2 + k_2^2\right) - \rho \omega^2\right)\left( \mu\left(k_1^2 + k_2^2\right) - \rho \omega^2\right)}{\rho^2}. 
\end{equation}

We now solve the algebraic equation $\det A = 0$ with respect to the first component $k_1$ of the wave-vector $k$ (as will be evident in section \ref{sec:ReflTransCC}, the second component of the wave-vector $k_2$ is always known when imposing boundary conditions)
\begin{align}
\left((2\mu + \lambda)  \left(k_1^2 + k_2^2\right) - \rho \omega^2\right)\left( \mu \left(k_1^2 + k_2^2\right) - \rho \omega^2\right) = 0, \nonumber\\
(2\mu + \lambda)  \left(k_1^2 + k_2^2\right) - \rho \omega^2 = 0 \quad \text{or} \quad \mu \left(k_1^2 + k_2^2\right) -\rho \omega^2 = 0,\nonumber\\
k_1^2 = \frac{\rho \,\omega^2}{2\mu + \lambda}-k_2^2 \quad \text{or} \quad k_1^2 = \frac{\rho}{\mu}\omega^2 - k_2^2, \label{eq:k1Cauchysquared}\\
k_1 =\pm \sqrt{\frac{\omega^2}{c_L^2}-k_2^2} \quad \text{or} \quad k_1 =\pm \sqrt{\frac{\omega^2}{c_S^2}-k_2^2}, \label{eq:k1Cauchy}
\end{align}

As we will show in the remainder of this subsection, the first or second solution in \eqref{eq:k1Cauchy} is associated to what we call a longitudinal or SV shear wave, respectively. The choice of sign for these solutions is related to the direction of propagation of the considered wave (positive for incident and transmitted waves, negative for reflected waves).\footnote{As a matter of fact, the sign $+$ or $-$ in expressions \eqref{eq:k1Cauchy} is uniquely determined by imposing that the solution preserves the conservation of energy at the considered interface. For Cauchy media, as well as for isotropic relaxed micromorphic media, it turns out that one must choose positive $k_1$ for transmitted and incident waves and negative $k_1$ for reflected ones (according to our convention). On the other hand, there exist some cases, as for example the case of anisotropic relaxed micromorphic media, for which these results are not straightforward and negative $k_1$ may appear also for transmitted waves, giving rise to so-called negative refraction phenomena. These unorthodox reflective properties are not the object of the present paper and will be discussed in subsequent works, where applications of the anisotropic relaxed micromorphic model will be presented.}

We will show later on in detail that, once boundary conditions are imposed at a given interface between two Cauchy media, the value of the component $k_2$ of the wave-vector $k$ can be considered to be known. We will see that $k_2$ is always real and positive, which means that, according to \eqref{eq:k1Cauchy}, the first component $k_1$ of the wave-vector can be either real or purely imaginary, depending on the values of the frequency and of the material parameters. Two scenarios are then possible:
\begin{enumerate}
\item both $k_1$ and $k_2$ are real: This means that, according to the wave ansatz \eqref{eq:planeansatzCauchy}, we have a harmonic wave which propagates in the direction given by the wave-vector $k$ lying in the $(x_1x_2)-$ plane. The wave-vector of the considered wave has two non-vanishing components in the $(x_1x_2)-$plane. A simplified illustration of this case is given in Fig. \ref{fig:propagative}.
\begin{figure}[H]
	\centering
	\includegraphics[scale=0.7]{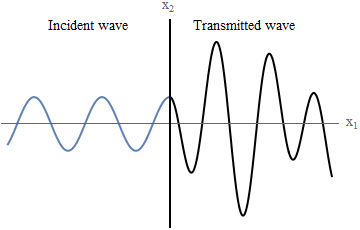}
	\caption{\small Schematic representation of an incident wave which propagates after hitting the interface. For graphical simplicity, only normal incidence and normal transmission are depicted.}
	\label{fig:propagative}
\end{figure}
\item $k_2$ is real and $k_1$ purely imaginary: In this case, according to equation \eqref{eq:planeansatzCauchy}, the wave continues to propagate in the $x_2$ direction (along the interface), but decays with a negative exponential in the $x_1$ direction (away from the interface). Such a wave propagating only along the interface is known as a Stoneley interface wave (\cite{stoneley1924elastic,auld1973acoustic2}). An illustration of this phenomenon is given in Fig. \ref{fig:evanescent}.
\begin{figure}[H]
	\centering
	\includegraphics[scale=0.7]{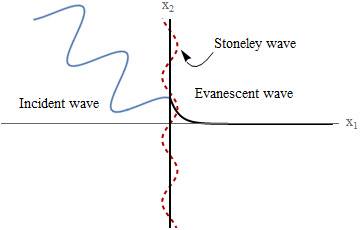}
	\caption{\small Schematic representation of an incident wave which is transformed into a interface wave along the interface and decays exponentially away from it (Stoneley wave).}
	\label{fig:evanescent}
\end{figure}
\end{enumerate}
In conclusion, we can say that when considering wave propagation in $2-$dimensional Cauchy media, it is possible that, given the material parameters, some frequencies exist for which all the energy of the incident wave can be redirected in Stoneley waves traveling along the interface, thus inhibiting transmission across the interface. Such waves are also possible at the interface between a Cauchy and a relaxed micromorphic medium, in a generalized setting.

Assuming the second component $k_2$ of the wave-vector to be known, we now calculate the solution $\widehat{\psi}$ to the algebraic equations $A \cdot \widehat{\psi} = 0$; 
these solutions make up the kernel (or nullspace) of the matrix $A$ and are essentially eigenvectors to the eigenvalue $0$. Hence, as is common nomenclature (\cite{madeo2016reflection}), we call them eigenvectors. 

Using the first solution of \eqref{eq:k1Cauchy}, we see that it also implies
\begin{equation}
k_1^2 = \frac{\omega^2}{c_L^2}-k_2^2.
\end{equation}
Using such relations between $k_1$ and $k_2$ in the first equation of \eqref{eq:systemCauchy}, we obtain

\begin{align}
\left(\omega^2 - c_L^2\left(\frac{\omega^2}{c_L^2}-k_2^2\right)-c_S^2k_2^2\right)\widehat{\psi}_1 - c_V^2 \sqrt{\frac{\omega^2}{c_L^2}-k_2^2}~k_2 \widehat{\psi}_2 &= 0, \nonumber \\
(\omega^2-\omega^2+c_L^2k_2^2-c_S^2k_2^2)\widehat{\psi}_1-c_V^2 \sqrt{\frac{\omega^2}{c_L^2}-k_2^2}~k_2 \widehat{\psi}_2 &= 0,\nonumber\\
(\underbrace{c_L^2-c_S^2}_{=c_V^2})k_2^2\widehat{\psi}_1-c_V^2 \sqrt{\frac{\omega^2}{c_L^2}-k_2^2}~k_2 \widehat{\psi}_2 &= 0, \label{eq:alphsa2long0}
\end{align}
this implies
\begin{equation}\label{eq:alpha2long}
\widehat{\psi}_2 =\frac{k_2}{\sqrt{\frac{\omega^2}{c_L^2}-k_2^2}}~\widehat{\psi}_1 \Leftrightarrow
\widehat{\psi}_2 =\frac{c_Lk_2}{\sqrt{\omega^2-c_L^2k_2^2}}~\widehat{\psi}_1.
\end{equation}

So, the eigenvector of the matrix $A$ in this case is given by\footnote{We explicitly remark that the definition of $\widehat{\psi}^L$ given by the first equality allows us to compute the vector $\widehat{\psi}^L$ once $k_2$ is known, i.e. when imposing boundary conditions. On the other hand, the equivalent definition $\widehat{\psi}^L=\left(1, k_2/k_1\right)^T$ allows us to remark that $k_1 \widehat{\psi}^L=\left(k_1,k_2\right)^T$ is still a solution of the equation $A\cdot \widehat{\psi}=0$. In this case, the vector of amplitudes is collinear with the wave-vector $k$. This allows us to talk about \textbf{longitudinal waves}, since the displacement vector $(u_1,u_2)^T$ given by expression \eqref{eq:planeansatzCauchy} (and hence the motion) is parallel to the direction of propagation of the traveling wave. We also notice that, given the eigenvector $\widehat{\psi}^L$, all vectors $a \widehat{\psi}^{L}$, with $a \in \C$, are solutions to the equation $A\cdot \widehat{\psi}=0$.}
\begin{equation}\label{eq:nullspacelong}
\widehat{\psi}^{L}
:=\left(\begin{array}{c}
1\\
\frac{c_Lk_2}{\sqrt{\omega^2-c_L^2k_2^2}}
\end{array}\right)=\left( \begin{array}{c}
1\\
\frac{k_2}{k_1}
\end{array}\right).
\end{equation}

Analogous considerations can be carried out when considering the second solution of \eqref{eq:k1Cauchy}, which also implies
\begin{equation}
k_1^2=\frac{\omega^2}{c_S^2}-k_2^2.
\end{equation}
Using this relation between $k_1$ and $k_2$ in the second equality of \eqref{eq:systemCauchy} gives

\begin{align}
\left(\omega^2 - c_L^2\left(\frac{\omega^2}{c_S^2}-k_2^2\right)-c_S^2k_2\right)\widehat{\psi}_1 - c_V^2 \sqrt{\frac{\omega^2}{c_S^2}-k_2^2}~k_2 \widehat{\psi}_2 &= 0, \nonumber\\
\left(\omega^2-\frac{c_L^2}{c_S^2}\omega^2+c_L^2k_2^2-c_S^2k_2^2\right)\widehat{\psi}_1-c_V^2\sqrt{\frac{\omega^2}{c_S^2}-k_2^2}~k_2\widehat{\psi}_2 &=0,\nonumber \\
\left(\omega^2\frac{c_S^2-c_L^2}{c_S^2}+(c_L^2-c_S^2)k_2^2\right)\widehat{\psi}_1-c_V^2\sqrt{\frac{\omega^2}{c_S^2}-k_2^2}~k_2\widehat{\psi}_2&=0,\nonumber\\
\left(-\omega^2\frac{c_V^2}{c_S^2}+c_V^2k_2^2\right)\widehat{\psi}_1-c_V^2\sqrt{\frac{\omega^2}{c_S^2}-k_2^2}~k_2\widehat{\psi}_2&=0,\label{eq:alpha2sv0}\\
\end{align}
this implies
\begin{equation}\label{eq:alpha2sv}
k_2\sqrt{\omega^2-k_2^2c_S^2}~\widehat{\psi}_2 = \frac{k_2^2c_S^2-\omega^2}{c_S}~\widehat{\psi}_1 \Leftrightarrow \widehat{\psi}_2=\frac{k_2^2c_S^2-\omega^2}{k_2c_S\sqrt{\omega^2-k_2^2c_S^2}}~\widehat{\psi}_1.
\end{equation}

Therefore, the eigenvector of the matrix $A$ in this case is given by\footnote{Analogously to the case of longitudinal waves, we can remark that the vector $\left(k_2,-k_1\right)^T$ is still a solution of the equation $A\cdot \widehat{\psi}=0$. This means that, in this case, the vector of amplitudes lies in the $(x_1x_2)-$plane and is orthogonal to the direction of propagation given by the wave-vector $k$. This allows us to introduce the concept of \textbf{transverse in-plane waves}, or ``shear vertical'' SV waves, since the displacement $(u_1,u_2)^T$ given by \eqref{eq:planeansatzCauchy} (and hence the motion) is orthogonal to the direction of propagation of the traveling wave, but lying in the $(x_1x_2)-$plane. We also remark that any vector $a\, \widehat{\psi}^{SV}$, with $a \in \C$, is solution to the equation $A\cdot \widehat{\psi}=0$. The first equality defining $\widehat{\psi}$ in \eqref{eq:nullspaceshear} allows to directly compute $\widehat{\psi}$ when $k_2$ is known, i.e. when imposing boundary conditions.}
\begin{equation}\label{eq:nullspaceshear}
\widehat{\psi}^{SV}
:=\left(\begin{array}{c}
1\\
\frac{k_2^2c_S^2-\omega^2}{k_2c_S\sqrt{\omega^2-k_2^2c_S^2}}
\end{array}\right)=\left( 
\begin{array}{c}
1\\
-\frac{k_1}{k_2}
\end{array}  \right).
\end{equation}

Finally, replacing the plane-wave ansatz \eqref{eq:planewaveCauchy3} for the third component $u_3$ of the displacement in equation \eqref{eq:comp3} gives\footnote{The component $u_3$ of the displacement being orthogonal to the $(x_1x_2)-$plane and thus to the direction of propagation of the wave, allows us to talk about \textbf{transverse, out-of-plane waves} or, equivalently, ``shear horizontal'' SH waves. Such waves have the same speed $c_S$ as the SV waves.}
\begin{align}
- \omega^2 \rho \, \widehat{\psi}_3 e^{i\left(\langle x,k\rangle - \omega t\right)} &= \mu(-k^2_1 -k_2^2) \widehat{\psi}_3 e^{i\left(\langle x,k\rangle - \omega t\right)}, \nonumber \\
\rho \,\omega^2 &= \mu (k_1^2+k_2^2), \nonumber \\
k_1 &= \pm \sqrt{\frac{\omega^2}{c_S^2}-k_2^2}. \label{eq:k1Cauchy3}
\end{align}
This relation, compared to the second of equations \eqref{eq:k1Cauchy}, tells us that the same relation between $k_1$ and $k_2$ exists when considering SV or SH waves.

\subsubsection*{Plane-wave ansatz: solution for the displacement field in a Cauchy medium}
We have seen in section \ref{sec:Cauchywaveprop} how to write the displacement field making use of the concepts of longitudinal, SV and SH waves.

According to equations \eqref{eq:planeansatzCauchy} and \eqref{eq:planewaveCauchy3} and considering the $2$D eigenvectors \eqref{eq:nullspacelong} and \eqref{eq:nullspaceshear}, we can finally write the solution for the displacement field $u=(u_1,u_2,u_3)^T$ as
\begin{equation}\label{eq:sollongsv} 
u = u^{L}+u^{SV} = a^L \psi^L e^{i\left(x_1 k_1^L + x_2 k_2^L - \omega t\right)} + a^{SV} \psi^{SV} e^{i\left(x_1 k_1^{SV} + x_2 k_2^{SV} - \omega t\right)}, 
\end{equation}
when we consider a longitudinal or an SV wave, or
\begin{equation} \label{eq:solsh}
u = u^{SH} = a^{SH} \psi^{SH} e^{i\left(x_1 k_1^{SH} + x_2 k_2^{SH} - \omega t\right)},
\end{equation}
when we consider an SH wave. In these formulas, starting from equations \eqref{eq:nullspacelong} and \eqref{eq:nullspaceshear}, we defined the unit vectors
\begin{equation}\label{eq:psiLSVSH}
\psi^L = \frac{1}{\left| \widehat{\psi}^L\right|}\left( \begin{array}{c}
\widehat{\psi}^L_1\\
\widehat{\psi}^L_2\\
0
\end{array}\right), \quad 
\psi^{SV} = \frac{1}{\left| \widehat{\psi}^{SV}\right|}\left( \begin{array}{c}
\widehat{\psi}^{SV}_1\\
\widehat{\psi}^{SV}_2\\
0
\end{array}\right), \quad \text{ and } \quad 
\psi^{SH} = \left( \begin{array}{c}
0\\
0\\
1
\end{array}\right).
\end{equation}
Finally, $a^L, a^{SV}, a^{SH} \in \C$ are arbitrary constants.

We also explicitly remark that in equations \eqref{eq:sollongsv} and \eqref{eq:solsh}, $k_1^L$ and $k_2^L$ are related via the first equation of \eqref{eq:k1Cauchy}, $k_1^{SV}$ and $k_2^{SV}$ via the second equation of \eqref{eq:k1Cauchy} and $k_1^{SH}$ and $k_2^{SH}$ via \eqref{eq:k1Cauchy3}.

As we already mentioned, $k_2$ will be known when imposing boundary conditions, so the only unknowns in the solution \eqref{eq:sollongsv} (resp \eqref{eq:solsh}) remain the scalar amplitudes $a^L, a^{SV}$ (resp. $a^{SH}$). We will show in the following subsection how, using the form \eqref{eq:sollongsv} (resp. \eqref{eq:solsh}) for the solution in the problem of reflection and transmission at an interface between Cauchy media, the unknown amplitudes can be computed by imposing boundary conditions.

	\subsection{Interface between two Cauchy media}\label{sec:ReflTransCC}
We now turn to the problem of reflection and transmission of elastic waves at an interface between two Cauchy media (cf. Fig. \ref{fig:CauchyCauchy}). We assume that an incident longitudinal wave\footnote{The exact same considerations can be carried over when the incident wave is SV transverse.} propagates on the ``minus'' side, hits the interface (which is the $x_2$ axis) at the origin of our co-ordinate system and is subsequently reflected into a longitudinal wave and an SV wave and is transmitted into the second medium on the ``plus'' side into a longitudinal wave and an SV wave as well, according to equations \eqref{eq:sollongsv}. If we send an incident SH wave, equation \eqref{eq:solsh} tells us that it will reflect and transmit only into SH waves.
\subsubsection{Incident longitudinal/SV-transverse wave}\label{sec:CCLSV}
According to our previous considerations, and given the linearity of the considered problem, the solution of the dynamical problem \eqref{eq:Cauchystrong} on the ``minus'' side can be written as\footnote{With a clear extension of the previously introduced notation, we denote by $a^{L,i}, a^{L,r}, a^{L,t},a^{SV,r},a^{SV,t}$ and $\psi^{L,i},
\psi^{L,r},\psi^{L,t},\psi^{SV,r},\psi^{SV,t}$ the amplitudes and eigenvectors of longitudinal incident, reflected, transmitted, in-plane transverse incident, reflected and transmitted waves respectively. Analogously, $a^{SV,i}$ and $\psi^{SV,i}$ will denote the amplitude and eigenvector related to in-plane transverse incident waves.}
\small
\begin{equation}\label{eq:solplus}
u^{-}(x_1,x_2,t) = a^{L,i} \psi^{L,i} e^{i\left(\left\langle x, k^{L,i}\right\rangle - \omega t\right)} + a^{L,r} \psi^{L,r} e^{i\left(\left\langle x, k^{L,r}\right\rangle - \omega t\right)} + a^{SV,r} \psi^{SV,r} e^{i\left(\left\langle  x,k^{SV,r}\right\rangle - \omega t\right)} =: u^{L,i} + u^{L,r} + u^{SV,r}.
\end{equation}
\normalsize
As for the ``plus'' side, the solution is
\begin{equation}\label{eq:solminus}
u^{+}(x_1,x_2,t) = a^{L,t} \psi^{L,t} e^{i\left(\left\langle x, k^{L,t} \right\rangle - \omega t\right)} + a^{SV,t} \psi^{SV,t}e^{i\left(\left\langle   x,k^{SV,t}\right\rangle - \omega t\right)} =: u^{L,t} + u^{SV,t}. 
\end{equation}
The vectors $\psi^{L,i}, \psi^{SV,r}$ and $\psi^{SV,t}$ are as in \eqref{eq:psiLSVSH}.

Now the new task is given an incident wave, i.e. knowing $a^{L,i}$ and $k^{L,i}$, to calculate all the respective parameters of the ``new'' waves. 

We see that the jump condition \eqref{displcont} can be further developed considering that $u^{-} = u^{L,i} + u^{L,r} + u^{SV,r}$ and $u^{+} = u^{L,t} + u^{SV,t}$. We equate the first components of \eqref{displcont}
\begin{equation*}
u_1^{L,i} + u_1^{L,r} +u_1^{SV,r} = u_1^{L,t} +u_1^{SV,t}, 
\end{equation*}
which, evaluating the involved expression at $x_1=0$ (the interface), implies
\small
\begin{equation*}
a^{L,i} \psi_1^{L,i} e^{i\left( x_2 k^{L,i}_2 - \omega t\right)} + a^{L,r}\psi^{L,r}_1 e^{i\left(x_2 k^{L,r}_2 - \omega t\right)} + a^{SV,r} \psi_1^{SV,r} e^{i\left(x_2 k^{SV,r}_2 - \omega t\right)} = a^{L,t} \psi_1^{L,t}e^{i\left( x_2 k_2^{L,t} - \omega t\right)} + a^{SV,t} \psi_1^{SV,t} e^{i\left(x_2 k_2^{SV,t} - \omega t\right)},
\end{equation*}
\normalsize
or, simplifying everywhere the time exponential
\begin{equation}
a^{L,i} \psi_1^{L,i} e^{i x_2 k^{L,i}_2 } + a^{L,r}\psi^{L,r}_1 e^{i x_2 k^{L,r}_2} + a^{SV,r} \psi_1^{SV,r} e^{i x_2 k^{SV,r}_2 }   
= a^{L,t} \psi_1^{L,t}e^{i x_2 k_2^{L,t}} + a^{SV,t} \psi_1^{SV,t} e^{i x_2 k_2^{SV,t}}. \label{expfactor}
\end{equation}
This must hold for all $x_2 \in \R$. The exponentials in this expression form a family of linearly independent functions and, therefore, we can safely assume that the coefficients $a^{L,i},a^{L,r},a^{SV,r},a^{L,t},a^{SV,t}$ are never all zero simultaneously. This means, than in order for \eqref{expfactor} to hold, we must require that the exponents of the exponentials are all equal to one another. Canceling out the imaginary unit $i$ and $x_2$, we deduce the fundamental relation\footnote{It is now clear why we said previously that the second components of all wave-vectors are known, since they are all equal to the second component of the wave-vector of the prescribed incident wave, which is known by definition.} 

\begin{equation}
\label{eq:SnellCauchy}
\boxed{
k_2^{L,i}=k_2^{L,r}=k_2^{SV,r}=k_2^{L,t}=k_2^{SV,t}
}\, ,
\end{equation}
which is the well-known \textbf{Snell's law} for in-plane waves (see \cite{achenbach1973wave,graff1975wave,zhu2015study,auld1973acoustic2}).
Using \eqref{eq:SnellCauchy} we see that the exponentials in \eqref{expfactor} can be canceled out leaving only
\begin{equation}\label{amplitudes1}
a^{L,i} \psi^{L,i}_1 + a^{L,r} \psi^{L,r}_1 + a^{SV,r} \psi^{SV,r}_1 = a^{L,t} \psi^{L,t}_1 +a^{SV,t} \psi^{SV,t}_1. 
\end{equation}
Analogously, equating the second components of the displacements in the jump conditions \eqref{displcont}, using \eqref{eq:SnellCauchy} and the fact that this must hold for all $x_2 \in \R$, gives 
\begin{equation}\label{amplitudes2}
a^{L,i} \psi^{L,i}_2 + a^{L,r} \psi^{L,r}_2 + a^{SV,r} \psi^{SV,r}_2 = a^{L,t} \psi^{L,t}_2 +a^{SV,t} \psi^{SV,t}_2. 
\end{equation}

As for the jump of force, we remark that the total force on both sides is $
F^{-} = f^{L,i} + f^{L,r} + f^{SV,r}$, $F^{+} = f^{L,t} + f^{SV,t}$,
with the $f$'s being evaluated at $x_1=0$. The forces are vectors of the form
\begin{equation}\label{forcevec}
f =  \left(f_1,f_2,0\right)^T,
\end{equation}
with 
\begin{equation}\label{forcecomp}
f_i = \sigma_{ij} \nu_j,
\end{equation}
where $\nu = (\nu_1,\nu_2,\nu_3)^T = (1,0,0)^T$ is the normal vector to the interface, i.e. to the $x_2$ axis.\footnote{We immediately see that the only components of the stress tensor which have a contribution in the calculation of the force jump are $\sigma_{11}$ and $\sigma_{21}$.} 

The force jump condition \eqref{jumpforcevec} can now be written component-wise as 
\begin{align}
&\sigma^{L,i}_{11} + \sigma^{L,r}_{11} + \sigma^{SV,r}_{11} = \sigma^{L,t}_{11} + \sigma^{SV,t}_{11}, \label{forcejumpone} \\ 
&\sigma^{L,i}_{21} + \sigma^{L,r}_{21} + \sigma^{SV,r}_{21} = \sigma^{L,t}_{21} + \sigma^{SV,t}_{21}. \label{forcejumptwo}
\end{align}
Calculating the stresses according to eq. \eqref{eq:Cauchystress}, where we use the solutions \eqref{eq:solplus} and \eqref{eq:solminus} for the displacement and again using \eqref{eq:SnellCauchy} gives
\small
\begin{align}
a^{L,i} &\left( (2\mu+\lambda)\psi_1^{L,i} k_1^{L,i} + \lambda\psi_2^{L,i}k_2^{L,i}\right) + a^{L,r} \left( (2\mu+\lambda)\psi_1^{L,r} k_1^{L,r} + \lambda\psi_2^{L,r}k_2^{L,r}\right) + a^{SV,r} \left( (2\mu+\lambda)\psi_1^{SV,r} k_1^{SV,r} + \lambda\psi_2^{SV,r}k_2^{SV,r}\right)\nonumber \\
&= a^{L,t}  \left( (2\mu^{+}+\lambda^{+})\psi_1^{L,t} k_1^{L,t} + \lambda^{+}\psi_2^{L,t}k_2^{L,t}\right) + a^{SV,t}  \left( (2\mu^{+}+\lambda^{+})\psi_1^{SV,t} k_1^{SV,t} + \lambda^{+}\psi_2^{SV,t}k_2^{SV,t}\right), \label{eq:jumpforce1}
\end{align}
\normalsize
and
\begin{align}
a^{L,i} \mu \left(\psi_1^{L,i} k_2^{L,i} + \psi_2^{L,i} k_1^{L,i}\right) &+ a^{L,r} \mu \left(\psi_1^{L,r} k_2^{L,r} + \psi_2^{L,r} k_1^{L,r}\right) + a^{SV,r} \mu \left(\psi_1^{SV,r} k_2^{SV,r} + \psi_2^{SV,r} k_1^{SV,r}\right) \nonumber \\
&=a^{L,t} \mu^{+} \left(\psi_1^{L,t} k_2^{L,t} + \psi_2^{L,t} k_1^{L,t}\right) + a^{SV,t} \mu^{+} \left(\psi_1^{SV,t} k_2^{SV,t} + \psi_2^{SV,t} k_1^{SV,t}\right). \label{eq:jumpforce2}
\end{align}
Thus, equations \eqref{amplitudes1}, \eqref{amplitudes2}, \eqref{eq:jumpforce1}, \eqref{eq:jumpforce2} form an algebraic system for the unknown amplitudes $a^{L,r},$ $a^{SV,r},$ $a^{L,t},$ $a^{SV,t}$ from which we can fully calculate the solution.

\subsubsection{Incident SH-transverse wave}\label{sec:CCSH}
In this case, the solution on the ``minus'' side of the interface is
\begin{equation}\label{eq:solSHminus}
u^{-}(x_1,x_2,t)= a^{SH,i} \psi^{SH,i} e^{i\left(\left\langle x,k^{SH,i}\right\rangle-\omega t\right)} + a^{SH,r} \psi^{SH,r} e^{i\left(\left\langle x,k^{SH,r}\right\rangle-\omega t\right)}, 
\end{equation}
and on the ``plus'' side
\begin{equation}\label{eq:solSHplus}
u^{+}(x_1,x_2,t)= a^{SH,t} \psi^{SH,t} e^{i\left(\left\langle x,k^{SH,t}\right\rangle-\omega t\right)},
\end{equation}
where the vectors $\psi^{SH,i},\psi^{SH,r}, \psi^{SH,t}$ are all equal to $(0,0,1)^T$, according to the third equation of \eqref{eq:psiLSVSH}. Following the same reasoning as in section \ref{sec:CCLSV}, the continuity of displacement condition now only involves the $u_3$ component, which is the only non-zero one, and reads (evaluating again at $x_1=0$)

\begin{equation}\label{eq:expfactorSH}
a^{SH,i} e^{i\left(x_2 k_2^{SH,i}\right)} + a^{SH,r} e^{i\left(x_2 k_2^{SH,r}\right)} = a^{SH,t} e^{i\left(x_2 k_2^{SH,t}\right)},
\end{equation}
which, since the exponentials build a family of linearly independent functions and if we exclude the case where all amplitudes $a^{SH,i}, a^{SH,r} , a^{SH,t}$ are identically equal to zero, becomes \textbf{Snell's law} for out-of-plane motions
\begin{equation}\label{eq:SnellCauchySH}
\boxed{
k_2^{SH,i} = k_2^{SH,r} = k_2^{SH,t}
}\,.
\end{equation}
Using that, we see that the exponentials in \eqref{eq:expfactorSH} cancel out leaving only
\begin{equation}\label{eq:displSH}
a^{SH,i} + a^{SH,r} = a^{SH,t}.
\end{equation}

As for the jump in force in the case of SH waves, the total force on both sides is $F^{-} = f^{SH,i} + f^{SH,r}$, $F^{+} = f^{SH,t}$, with the $f$'s being evaluated at $x_1=0$. Now the forces are vectors of the form
\[
f = (0,0,f_3)^T,
\]
where, once again, $f_i = \sigma_{ij}\nu_j$, where $\nu = (1,0,0)^T$ is the vector normal to the interface.\footnote{We now see that the component of the stress which has an influence in this boundary condition is $\sigma_{31}$.} Condition \eqref{jumpforcevec} can now be written as 
\begin{equation}
\sigma_{31}^{SH,i} + \sigma_{31}^{SH,r} = \sigma_{31}^{SH,t}
\end{equation}
The stresses are calculated again by \eqref{eq:Cauchystress} and using the solutions \eqref{eq:solSHplus} and \eqref{eq:solSHminus} for the displacement and using \eqref{eq:SnellCauchySH} gives
\begin{equation}\label{eq:stressesSH}
\mu \left(a^{SH,i} k_1^{SH,i}+a^{SH,r} k_1^{SH,r} \right) = \mu^{+} a^{SH,t} k_1^{SH,t}.  
\end{equation}

Equations \eqref{eq:displSH} and \eqref{eq:stressesSH} build a system for the unknown amplitudes $a^{SH,r},a^{SH,t}$, which we can solve and fully determine the solution to the reflection-transmission problem.

\subsubsection{A condition for the onset of Stoneley waves at a Cauchy/Cauchy interface}
In this subsection we show how we can find explicit conditions for the onset of Stoneley waves at Cauchy/Cauchy interfaces. This can be done by simply requesting that the quantities under the square roots in equation \eqref{eq:k1Cauchy} become negative. 

Assume that the incident wave is longitudinal. This means that its speed is given by $c_L^-=\sqrt{(2\mu^-+\lambda^-)/\rho^-}$ and that the wave vector $k$ can now be written as $k=(k_1,k_2)=\kabs(\sin \theta_i, -\cos \theta_i)$, where $\kabs = \omega/c_L^-$ and $\theta_i$ is the angle of incidence. This incident longitudinal wave gives rise to a longitudinal and a transverse wave both on the ``$-$'' and on the ``$+$'' side. Setting the quantity under the square root in the first equation in \eqref{eq:k1Cauchy} to be negative and using the fact that $k_2 = -\kabs \cos \theta_i$, gives a condition for the appearance of Stoneley waves in the case of an incident longitudinal wave: 

\begin{align}
\frac{\omega^2}{(c_L^{+})^2}-k_2^+ <0  \Leftrightarrow \frac{\omega^2}{(c_L^{+})^2}-k_2^2 <0 &\Leftrightarrow \frac{\omega^2}{(c_L^{+})^2} <\kabs^2 \cos^2 \theta_i \nonumber \\ 
&\Leftrightarrow \frac{\omega^2}{(c_L^{+})^2} < \frac{\omega^2}{(c_L^-)^2} \cos^2 \theta_i  \nonumber \\
&\Leftrightarrow \cos^2 \theta_i > \left( \frac{c_L^-}{c_L^{+}}\right)^2 \nonumber\\
&\Leftrightarrow \cos^2 \theta_i > \frac{\rho^{+}(2 \mu^{-} + \lambda^{-})}{\rho^{-}(2\mu^{+} + \lambda^{+})}. \label{eq:StoneleyLong}
\end{align}
To establish the previous relation we also used the fact that $k_2^+ = k_2$, as established by \textbf{Snell's law} in \eqref{eq:SnellCauchy}.

Similar arguments can be carried out when considering all other possibilities for incident,  transmitted and reflected waves, as detailed in Tables \ref{table:StoneleyT} and \ref{table:StoneleyR}.

\begin{table}[H]
	\centering
	\begin{tabular}{c|c|c|c}
		Incident Wave & Transmitted L & Transmitted SV &Transmitted SH  \\
		\hline 
		L & $\cos^2 \theta_i > \frac{\rho^+(2\mu^- + \lambda^-)}{\rho^-(2 \mu^+ + \lambda^+)}$ & $\cos^2 \theta_i >  \frac{\rho^+(2\mu^- + \lambda^-)}{\rho^- \mu^+}$ & $-$ \\
		\hline
		SV & $\cos^2 \theta_i > \frac{\rho^+\mu^-}{\rho^-(2 \mu^+ + \lambda^+)}$ & $\cos^2 \theta_i > \frac{\rho^+\mu^-}{\rho^- \mu^+}$ & $-$   \\
		\hline
		SH &  $-$& $-$ &$\cos^2 \theta_i > \frac{\rho^+\mu^-}{\rho^- \mu^+}$  
	\end{tabular}
	\caption{\small Conditions for appearance of transmitted Stoneley waves for all types of waves at a Cauchy/Cauchy interface.}
	\label{table:StoneleyT}
\end{table}
\begin{table}[ht!]
	\centering
	\begin{tabular}{c|c|c}
		Incident Wave & Reflected L & Reflected SV \\
		\hline 
		L  &$-$ & $-$\\
		\hline
		SV   &$\cos^2 \theta_i > \frac{\mu^-}{2\mu^- + \lambda^-}$ & $-$ \\
		\hline
		SH &    $-$& $-$ 
	\end{tabular}
	\caption{\small Conditions for appearance of reflected Stoneley waves for all types of waves at a Cauchy/Cauchy interface.}
	\label{table:StoneleyR}
\end{table}
\normalsize
The conditions in Tables \ref{table:StoneleyT} and \ref{table:StoneleyR} establish that the square of the cosine of the angle of incidence must be greater than a given quantity (this happens for smaller angles) for Stoneley waves to appear. This means that it is most likely to observe Stoneley waves when the angle of incidence is smaller than the normal incidence angle, i.e. for incident waves which are inclined with respect to the surface upon which they hit. Moreover, if there exists a strong contrast in stiffness between the two sides and the ``$-$'' side is by far stiffer than the ``$+$'' side, then Stoneley waves could be observed for angles closer to normal incidence. On the other hand, if the ``$-$'' side is only slightly stiffer than the ``$+$'' side, then Stoneley waves will be observed only for smaller angles (far from normal incidence). We refer to Appendix \ref{appendixCauchy3} for an explicit calculation of all these conditions.

In order to fix ideas, assuming that both materials on the left and on the right have the same density, i.e. $\rho^{-} = \rho^{+}$ and examining Table \ref{table:StoneleyT}, we can deduce that Stoneley waves appear only when the expressions on the right of each inequality are less than one, i.e:\footnote{In order to respect positive definiteness of the strain energy density, we must require that $2 \mu^+ + \lambda^+>0$ and $\mu^+>0$.}
\begin{itemize}
\item For an incident L wave, the transmitted L mode becomes Stoneley only if the Lam\'e parameters of the material are chosen in such a way that $2\mu^{+} + \lambda^{+} > 2\mu^{-} + \lambda^{-}$.	
\item For an incident L wave, the transmitted SV mode becomes Stoneley only if the Lam\'e parameters of the material are chosen in such a way that $\mu^{+}  > 2\mu^{-} + \lambda^{-}$.	
\item For an incident SV wave, the transmitted L mode becomes Stoneley only if the Lam\'e parameters of the material are chosen in such a way that $2\mu^{+} + \lambda^{+} > \mu^{-} $.	
\item For an incident SV wave, the transmitted SV mode becomes Stoneley only if the Lam\'e parameters of the material are chosen in such a way that $\mu^{+}  > \mu^{-} $.	
\item For an incident SH wave, the transmitted SH mode (which is the only transmitted mode) becomes Stoneley only if the Lam\'e parameters of the material are chosen in such a way that $\mu^{+}  > \mu^{-} $.
\item Reflected waves can become Stoneley waves only when the incident wave is SV. The only mode which can be converted into a Stoneley wave is the L mode when $\mu^{-}>-\lambda^{-}$.	
\end{itemize}
All these cases are shown graphically in Figures \ref{fig:Stoneley1}, \ref{fig:Stoneley2}, \ref{fig:Stoneley3}. Figures \ref{fig:StoneleyTransmittedL1} and \ref{fig:StoneleyTransmittedL2} demonstrate the case of an incident L wave. From Table \ref{table:StoneleyT} we deduce that either only the L transmitted mode will be Stoneley or both L and SV (when $\rho^- = \rho^+$). The same is true for the case of an incident SV wave as shown in Figures \ref{fig:StoneleyTransmittedSV1} and \ref{fig:StoneleyTransmittedSV2}. For incident SH waves, we only have one transmitted mode which can also be Stoneley as shown in Figure \ref{fig:StoneleyTransmittedSH}. Figure \ref{fig:StoneleyReflected} shows the only possible manifestation of a reflected Stoneley wave, which is the L mode for the case of an incident SV wave. 

\vspace{-0.5cm}
\begin{figure}[H]
	\begin{subfigure}{.5\textwidth}
		\centering
		\includegraphics[scale=0.28]{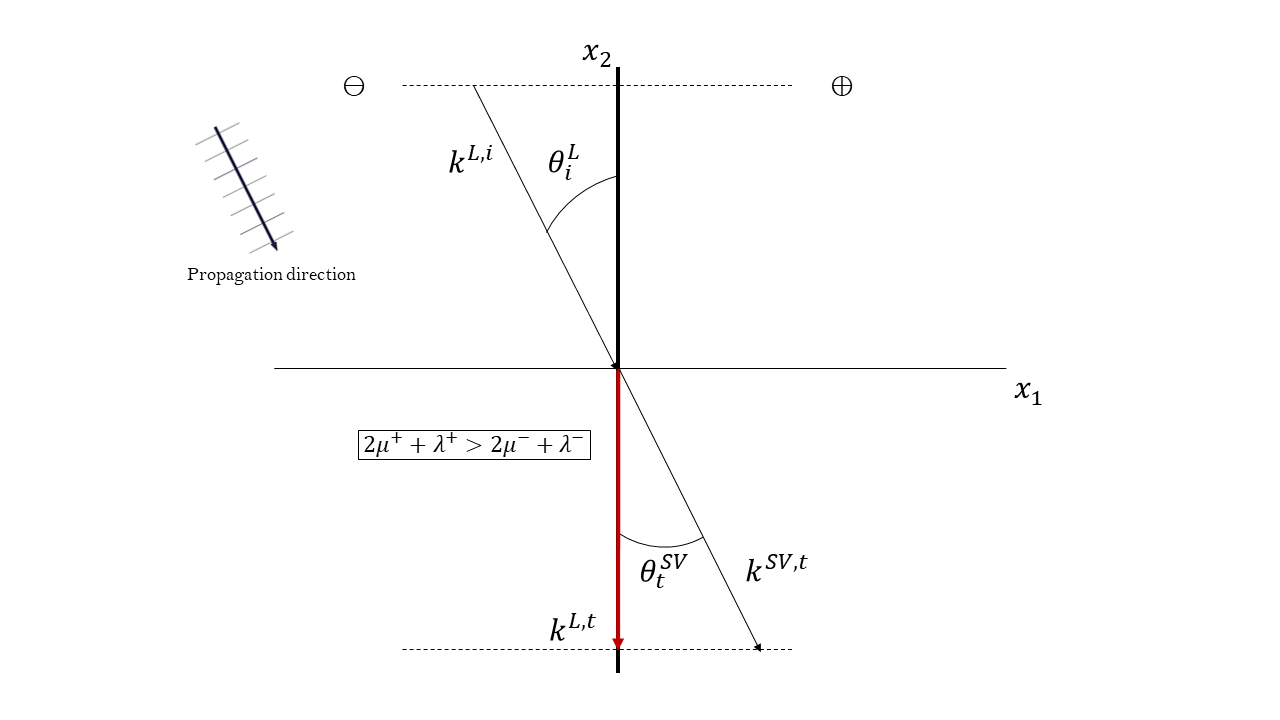}
		\caption{}
		\label{fig:StoneleyTransmittedL1}
	\end{subfigure}%
	\begin{subfigure}{.5\textwidth}
		\centering
		\includegraphics[scale=0.28]{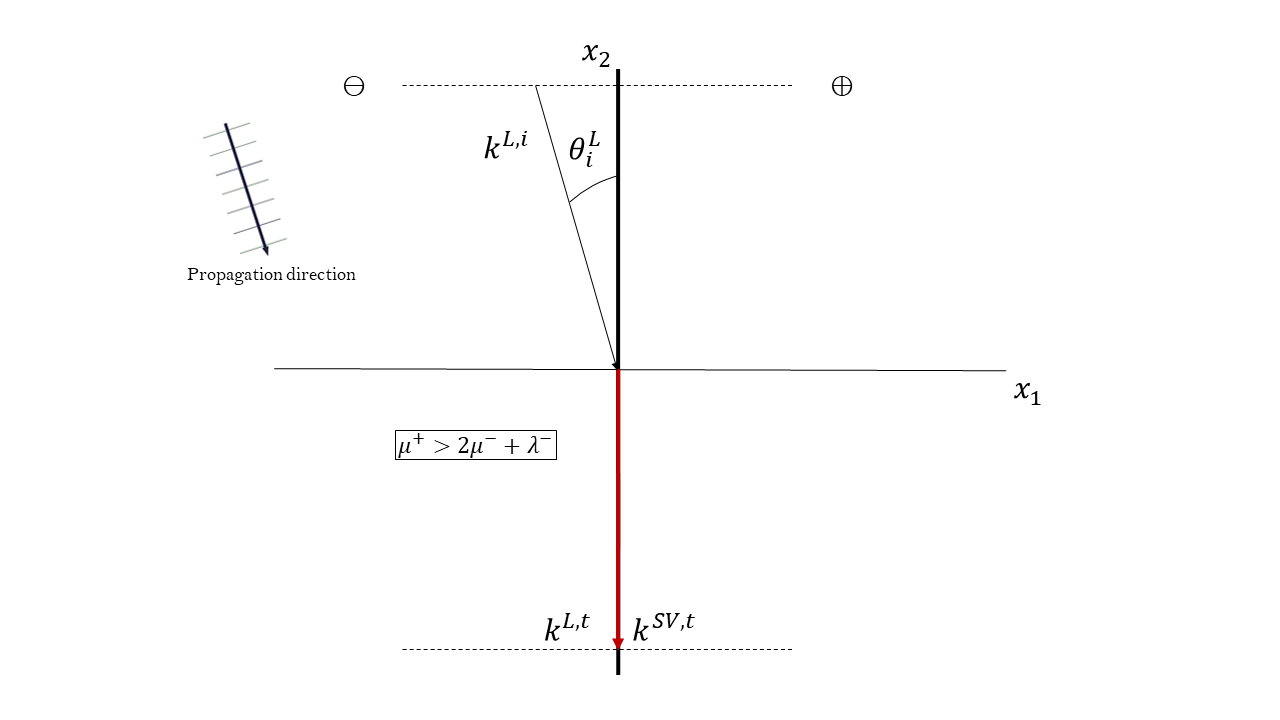}
		\caption{}
		\label{fig:StoneleyTransmittedL2}
	\end{subfigure}
	\caption{\small Possible manifestations of Stoneley waves at a Cauchy/Cauchy interface for L and SV transmitted modes in the case of incident L wave. The vectors in black represent propagative modes, while the vectors in red represent modes which propagate along the surface only (Stoneley modes).}
	\label{fig:Stoneley1}
\end{figure}
\begin{figure}[H]
	\begin{subfigure}{.5\textwidth}
		\centering
		\includegraphics[scale=0.28]{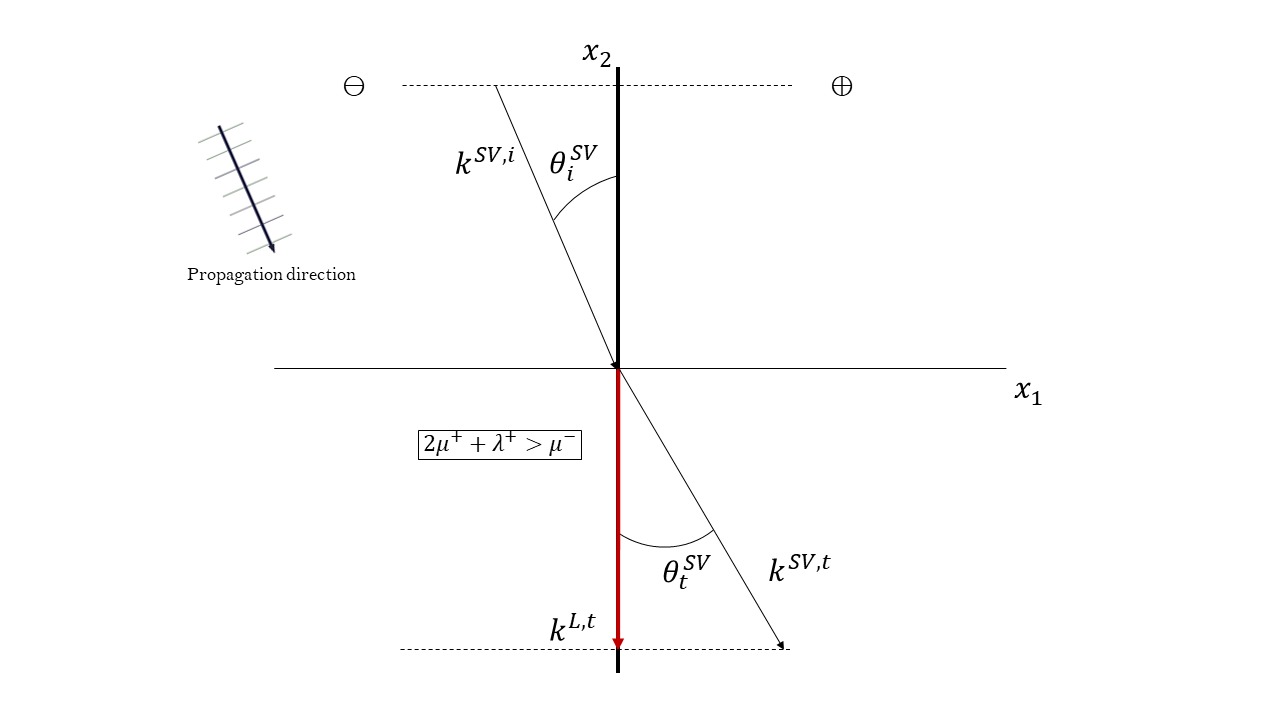}
		\caption{}
		\label{fig:StoneleyTransmittedSV1}
	\end{subfigure}%
	\begin{subfigure}{.5\textwidth}
		\centering
		\includegraphics[scale=0.28]{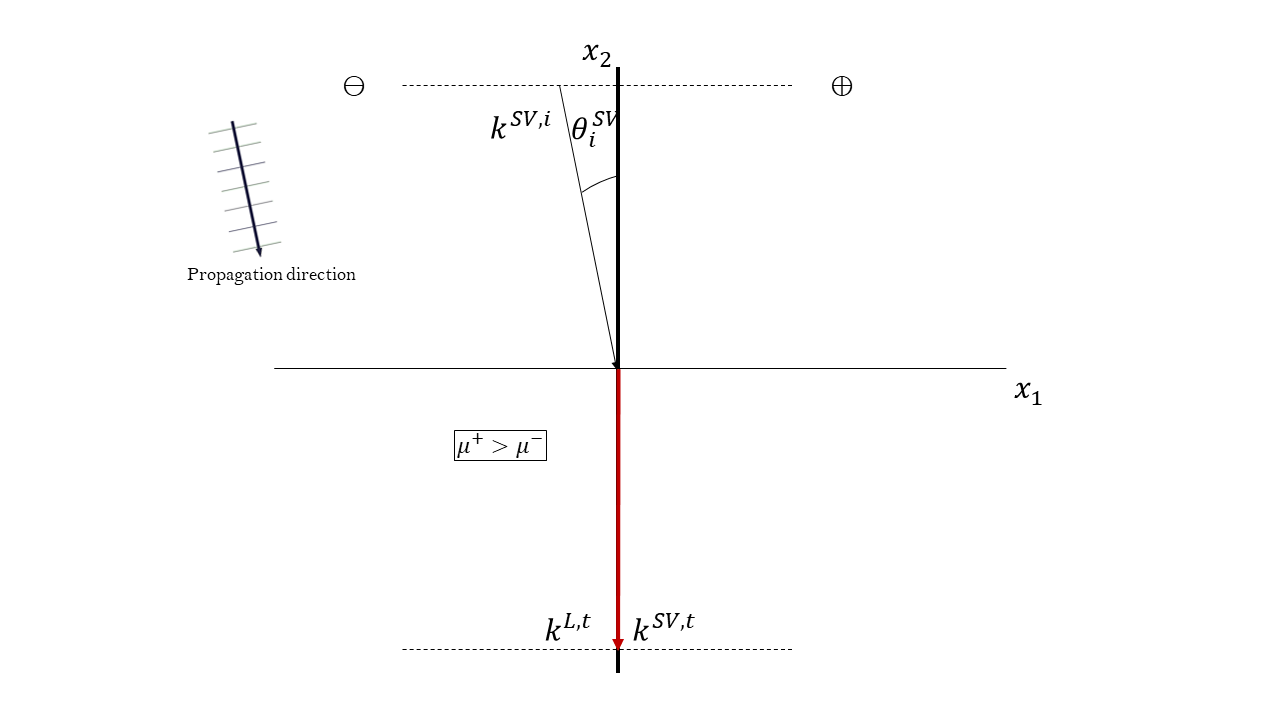}
		\caption{}
		\label{fig:StoneleyTransmittedSV2}
	\end{subfigure}
	\caption{\small Possible manifestations of Stoneley waves at a Cauchy/Cauchy interface for L and SV transmitted modes in the case of incident SV wave. The vectors in black represent propagative modes, while the vectors in red represent modes which propagate along the surface only (Stoneley modes).}
	\label{fig:Stoneley2}
\end{figure}
\begin{figure}[H]
	\begin{subfigure}{.5\textwidth}
		\centering
		\includegraphics[scale=0.28]{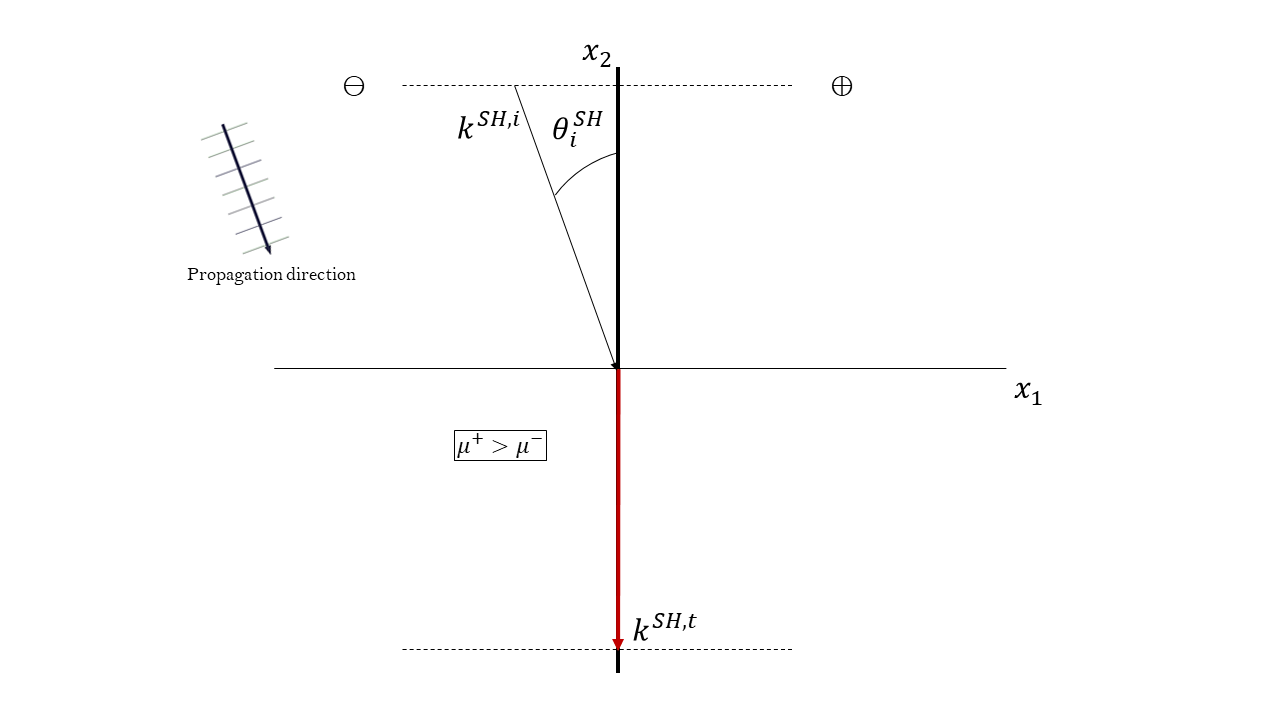}
		\caption{}
		\label{fig:StoneleyTransmittedSH}
	\end{subfigure}%
	\begin{subfigure}{.5\textwidth}
		\centering
		\includegraphics[scale=0.28]{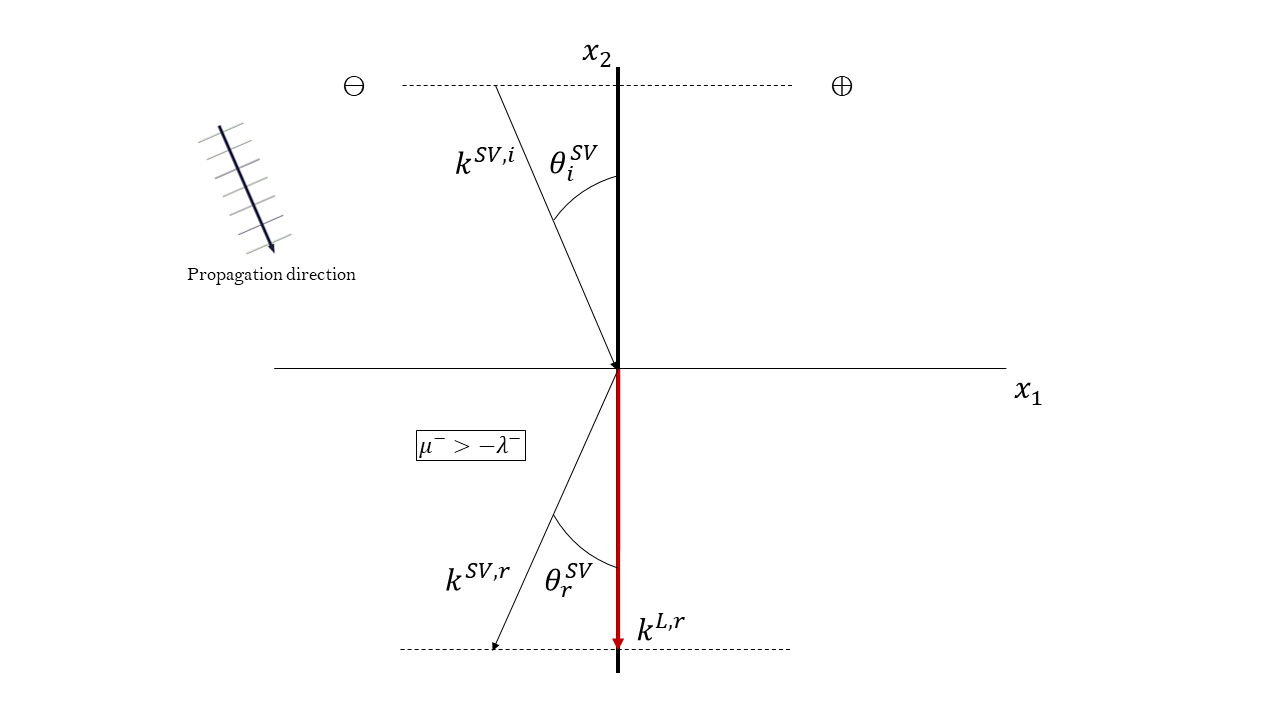}
		\caption{}
		\label{fig:StoneleyReflected}
	\end{subfigure}
	\caption{\small Possible manifestations of Stoneley waves at a Cauchy/Cauchy interface for the only SH transmitted mode in the case of an incident SH wave (a) and for the reflected L mode in the case of an incident SV wave (b). The vectors in black represent propagative modes, while the vectors in red represent modes which propagate along the surface only (Stoneley modes).}
	\label{fig:Stoneley3}
\end{figure}

\subsubsection{Determination of the reflection and transmission coefficients}
We denote by $H_1$ the first (normal) component of the flux vector and we introduce the following quantities
\begin{equation}\label{JCauchy}
J^i=\frac{1}{T}\int_0^{T}H^i_1(x,t) dt,\quad J^r=\frac{1}{T}\int_0^{T}H^r_1(x,t) dt,\quad J^t=\frac{1}{T}\int_0^{T}H^t_1 (x,t)dt,
\end{equation}
where $H^i_1=H^{L,i}_1$ (or, $H^i_1=H^{SV,i}_1$ if we consider an incident SV wave), $H^r_1 = H^{L,r}_1 + H^{SV,r}_1$ and \\$H^t_1 = H^{L,t}_1 + H^{SV,t}_1$, $T$ being the period of the wave.\footnote{Or, in the case of an incident SH wave $H^i_1=H^{SH,i}_1, H^r_1 = H^{SH,r}_1$ and $H^t_1 = H^{SH,t}_1$.} Then, the \textbf{reflection} and \textbf{transmission coefficients} are defined as 
\begin{equation}\label{refltranscoeff}
\mathcal{R} = \frac{J^r}{J^t}, \quad \mathcal{T}=\frac{J^t}{J^i}.
\end{equation}
These coefficients tell us what part of the average normal flux of the incident wave is reflected and what part is transmitted; also, since the system is conservative, we must have $\mathcal{R}+\mathcal{T}=1$. 

The integrals involved in these expressions are the average normal fluxes of the respective waves (incident, reflected or transmitted). We use Lemma \ref{Lemma1} (provided with proof in Appendix \ref{appendixCauchy}) in the computations of these coefficients.

In order to have a physical meaning, the final solution for the displacement $u$ must be real, so that we consider only the real parts of the displacements and stresses for the computation of the flux. For the first component of the flux vector for a longitudinal or SV wave, we have, according to equation \eqref{eq:Cauchyflux} and using the plane-wave ansatz
\begin{align}
\frac{1}{T}\int_0^{T} H_1 dt  &= \frac{1}{T}\int_0^{T} \Re\left(-u_{1,t}\right)\Re\left(\sigma_{11}\right) + \Re\left(-u_{2,t}\right)\Re\left(\sigma_{12}\right) dt \nonumber\\
&= \frac{1}{T} \int_0^{T} \Re\left(i\omega a \psi_1 e^{i\left(\langle x,k\rangle -\omega t\right)}\right)\Re\left( \left[(2 \mu + \lambda) \psi_1 k_1 + \lambda \psi_2 k_2\right]i a e^{i\left(\langle x,k\rangle -\omega t\right)}\right) \nonumber\\
&\hspace{0.965cm} +\Re\left( i\omega a \psi_2 e^{i\left(\langle x,k\rangle -\omega t\right)}\right)\Re\left(\mu(\psi_1 k_2 + \psi_2 k_2)i a e^{i\left(\langle x,k\rangle -\omega t\right)}\right)dt \nonumber\\
&\hspace{-0.39cm}\underset{\mathrm{ Lemma \ref{Lemma1}}}{=}\frac{1}{2} \Re\left(\left[(2\mu + \lambda) |\psi_1|^2 k_1 + \lambda \psi_1^{*}\psi_2 k_2 + \mu \left(\psi_1\psi_2^{*}k_2 + |\psi_2|^2k_1\right)\right]|a|^2\omega\right) \label{eq:fluxLSV}.
\end{align}
As for the case of an SH wave, we have
\begin{align}
\frac{1}{T}\int_0^{T} H_1 dt  &= \frac{1}{T}\int_0^{T} \Re\left(-u_{3,t}\right)\Re\left(\sigma_{13}\right)dt=\frac{1}{T}\int_0^{T} \Re\left(i\omega ae^{i\left(\langle x,k\rangle -\omega t\right)} \right)\Re\left(ik_1\mu a e^{i\left(\langle x,k\rangle -\omega t\right)}\right)dt \nonumber\\
&\underset{\mathrm{ Lemma \ref{Lemma1}}}{=} \frac{1}{2}\Re\left(\mu k_1 |a|^2 \omega\right) \label{eq:fluxSH}.
\end{align}
Such expressions for the fluxes, together with the linear decompositions given for $H_1^r$ and $H_1^t$, allow us to explicitly compute the reflection and transmission coefficients. 
\subsection{The particular case of propagative waves}
We have seen that, when considering two Cauchy media with an interface, two cases are possible, namely waves which propagate in the two considered half-planes and Stoneley waves, which only propagate along the interface but decay away from it. Stoneley waves do not propagate in the considered media and are related to imaginary values of the first component of the wave number. When considering fully propagative waves ($k_1$ and $k_2$ both real) the results provided before can be interpreted on a more immediate physical basis, which we detail in the present section.
The previous ansatz and calculations were performed without any hypothesis on the nature of the components of $k$: they were assumed to be either real or imaginary. 
However, when we consider a fully propagating wave we will demonstrate that we can recover some classical formulas and results which are usually found in the literature by considering the vector of direction of propagation of the wave, instead of the wave-vector $k$. 
For a fully propagative wave, the plane-wave ansazt can be written as
\begin{equation}\label{eq:planeansatzCauchy2}
u=\widehat{\psi}\,e^{i(|k| \left\langle x ,\xi\right\rangle - \omega t)} = \widehat{\psi}\,e^{i(|k|( x_1\xi_1 +  x_2\xi_2) - \omega t)},
\end{equation}
where $|k|$ is now the wave-number, which is defined as the modulus of the wave-vector $k$ and $\xi=(\xi_1,\xi_2)^T:=\frac{k}{|k|}$ is the so-called vector of propagation. This real vector $\xi$ has unit length ($\xi_1^2+\xi_2^2=1$).

By inserting $k_1 = \kabs \xi_1$ and $k_2 = \kabs \xi_2$ in \eqref{eq:k1Cauchysquared} we find
\begin{equation}\label{eq:k1Cauchy2squared}
\kabs^2 = \frac{\omega^2}{c_L^2}\text{ or }\kabs^2 = \frac{\omega^2}{c_S^2}, 
\end{equation}
or,
\begin{equation} \label{eq:k1Cauchy2}
\kabs = \pm \frac{\omega}{c_L} \text{ or } \kabs =\pm \frac{\omega}{c_S},
\end{equation}
where, again, the signs in \eqref{eq:k1Cauchy2} must be chosen depending on what kind of wave we consider (positive for incident and transmitted waves, negative for reflected waves).  Expressions \eqref{eq:k1Cauchy2} give the well-known linear dependence between the frequency $\omega$ and the wave-number $\kabs$ through the speeds $c_L$ and $c_S$ for longitudinal and shear waves respectively. Such behavior is known as a ``non-dispersive'' behavior, which means that in a Cauchy medium longitudinal and shear waves propagate at a constant speed ($c_L$ for longitudinal and $c_S$ for shear waves).

Choosing the first solution in \eqref{eq:k1Cauchy2}, so that $\omega = \kabs c_L$ and inserting $k=\kabs \xi$ into \eqref{eq:nullspacelong} we can can find the nullspace in the case of a propagative longitudinal wave
\begin{equation}\label{eq:nullspaceCauchylongReal}
\widehat{\psi} = \left(\begin{array}{c}
1 \\
\frac{c_L \kabs \xi_2}{\sqrt{\kabs^2 c_L^2 - c_L^2 \kabs \xi_2^2}}
\end{array}\right)=
\left(
\begin{array}{c}
1\\
\frac{\xi_2}{\xi_1}
\end{array}\right)
\end{equation}

Equivalently, choosing the second solution in \eqref{eq:k1Cauchy2squared} so that $\omega = \kabs c_S$ and again inserting $k=\kabs \xi$ into \eqref{eq:nullspaceshear} we find for the second component of the eigenvector
\begin{equation}
\frac{k_2^2c_S^2-\omega^2}{k_2c_S\sqrt{\omega^2-k_2^2c_S^2}}=\frac{\kabs^2 \xi_2^2 c_S^2-\kabs^2 c_S^2}{\kabs \xi_2 c_S \sqrt{\kabs^2 c_S^2 -\kabs^2 \xi_2^2 c_S^2}}=-\frac{\xi_1^2}{\xi_2 \xi_1},
\end{equation}
so the eigenvector for a propagative shear wave is\footnote{We neglected the sign of the $\kabs$ in the above calculations. Fixing the direction of propagation will automatically impose the sign of both $\kabs$ and $\xi_1$, which we then plug into equations \eqref{eq:nullspaceCauchylongReal} or \eqref{eq:nullspaceCauchyShearReal}.}

\begin{equation}\label{eq:nullspaceCauchyShearReal}
\widehat{\psi}=\left(
\begin{array}{c}
1\\
-\frac{\xi_1}{\xi_2}
\end{array}\right).
\end{equation}

The forms \eqref{eq:nullspaceCauchylongReal} and \eqref{eq:nullspaceCauchyShearReal} for the eigenvectors of L and SV propagative waves are suggestive because they allow to immediately  visualize the vector of propagation $\xi$ and, thus, the eigenvectors $\psi$ themselves in terms of the angles formed by the considered propagative wave and the interface (see Figure \ref{fig:SnellCauchy} and Table \ref{table:Table1}).

\begin{figure}[H]
	\centering
	\includegraphics[scale=0.35]{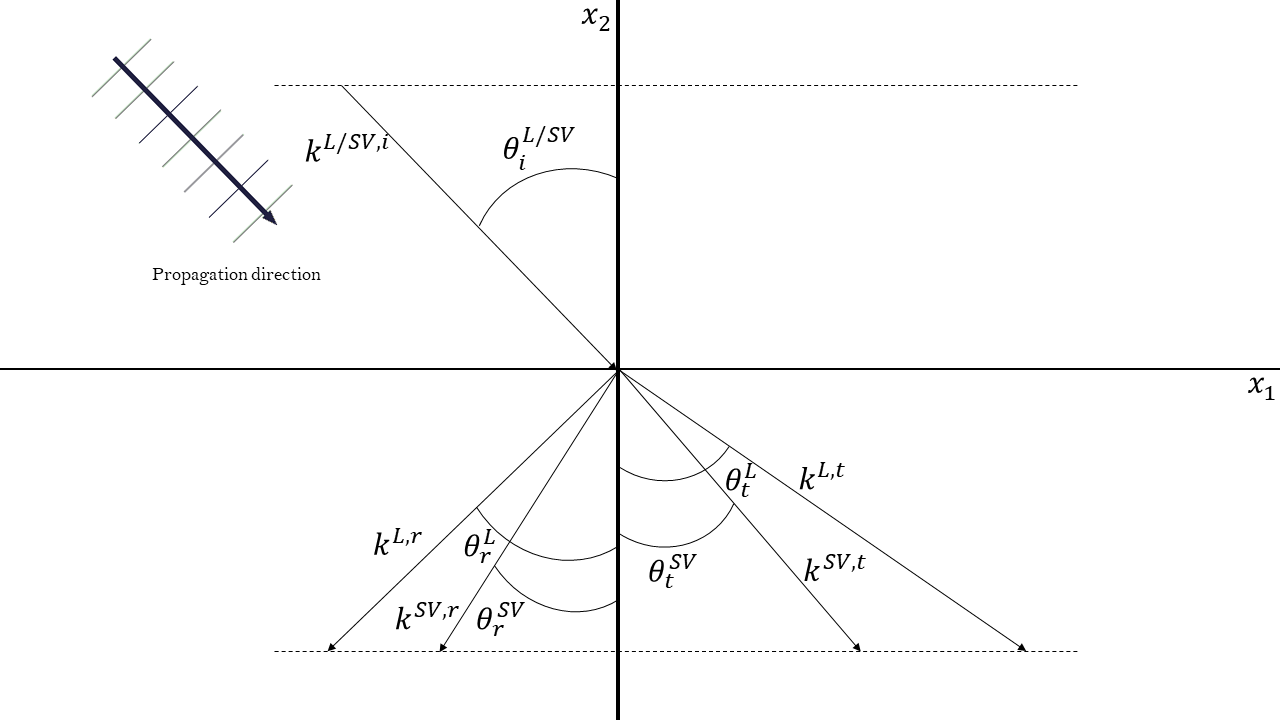}
	\caption{\small Illustration of reflection and transmission patterns for propagative waves and Snell's law at a Cauchy/Cauchy interface. The second components of all wave-vectors are equal to each other according to \eqref{eq:SnellCauchy}, forcing the reflected and transmitted wave-vectors to form angles with the interface as shown here according to \eqref{eq:SnellCauchyReal}.}
	\label{fig:SnellCauchy}
\end{figure}

In the propagative case, the solutions \eqref{eq:solplus} and \eqref{eq:solminus} particularize into
\begin{align}
u^{-}(x_1,x_2,t) &= a^{L,i} \psi^{L,i} e^{i\left(\left\langle x, \kabs \xi^{L,i}\right\rangle - \omega t\right)} + a^{L,r} \psi^{L,r} e^{i\left(\left\langle x, \kabs \xi^{L,r}\right\rangle - \omega t\right)} + a^{SV,r} \psi^{SV,r} e^{i\left(\left\langle  x,\kabs \xi^{SV,r}\right\rangle - \omega t\right)},\label{eq:solplusprop}\\ 
u^{+}(x_1,x_2,t) &= a^{L,t} \psi^{L,t} e^{i\left(\left\langle x, \kabs \xi^{L,t} \right\rangle - \omega t\right)} + a^{SV,t} \psi^{SV,t}e^{i\left(\left\langle   x,\kabs \xi^{SV,t}\right\rangle - \omega t\right)}. \label{eq:solminusprop} 
\end{align}
\begin{table}[H]
	\centering
	\begin{tabular}{c|c|c|c}
		Wave & $\xi$ & $\psi$ & Speed of propagation \\
		\hline
		$L,i$ & $\left(\sin \theta^{L}_i, -\cos\theta^{L}_i,0\right)^T$ & $\left(\sin \theta^{L}_i, -\cos\theta^{L}_i,0\right)^T$ &  $c_L$\\
		\hline 
		$SV,i$ & $\left(\sin \theta^{SV}_i, -\cos\theta^{SV}_i,0\right)^T$ & $\left(\cos \theta^{SV}_i, \sin\theta^{SV}_i,0\right)^T$ &  $c_S$ \\
		\hline
		$L,r$ & $\left(-\sin \theta^{L}_r, -\cos \theta^{L}_r,0\right)^T$ & $\left(-\sin \theta^{L}_r, -\cos \theta^{L}_r,0\right)^T$ & $c_L$\\
		\hline
		$SV,r$ & $\left(-\sin \theta^{SV}_r, -\cos \theta^{SV}_r,0\right)^T$ & $\left(\cos \theta^{SV}_r, -\sin \theta^{SV}_r,0\right)^T $ & $c_S$\\
		\hline
		$L,t$ & $\left(\sin \theta^{L}_t, -\cos \theta^{L}_t,0\right)^T $ & $\left(\sin \theta^{L}_t, -\cos \theta^{L}_t,0\right)^T $ & $c_L^+$\\
		\hline
		$SV,t$ & $\left(\sin \theta^{SV}_t, -\cos \theta^{SV}_t,0\right)^T $ & $\left(\cos \theta^{SV}_t, \sin \theta^{SV}_t,0\right)^T $ & $c_S^+$\\
	\end{tabular}
	\caption{\small Summary of the vectors of direction of propagation and of vibration for all different waves produced at a Cauchy/Cauchy interface.}
	\label{table:Table1}
\end{table}

Using in \eqref{eq:solplusprop}, \eqref{eq:solminusprop} and the forms given in Table \ref{table:Table1} for the propagation vectors $\xi$ and the eigenvectors $\psi$ as well as expressions \eqref{eq:k1Cauchy2} and \eqref{eq:nullspaceCauchylongReal}, we can remark that the only unknowns are the amplitudes $a$ and the angles $\theta$. The angles $\theta$ of the different waves can be computed in terms of the angle of the incident wave by using boundary conditions. Indeed, condition \eqref{eq:SnellCauchy} can be rewritten in the propagative case as
\begin{equation}\label{eq:SnellCauchyProp}
|k^{L,i}| \xi_2^{L,i} = |k^{L,r}| \xi_2^{L,r} = |k^{SV,r}| \xi_2^{SV,r} = |k^{L,t}| \xi_2^{L,t} =|k^{SV,t}| \xi_2^{SV,t},
\end{equation}
which, using equations \eqref{eq:k1Cauchy2} and \eqref{eq:nullspaceCauchylongReal}, as well as the $\xi_2$ given in Table \ref{table:Table1} and simplifying the frequency, gives the well-known \textbf{Snell's law} (Fig. \ref{fig:SnellCauchy})\footnote{Equation \eqref{eq:SnellCauchyReal} clarifies why the angles of a different reflected and transmitted waves are chosen to be as in Fig. \ref{fig:SnellCauchy}, instead of choosing the opposite ones.} 

\begin{equation}\label{eq:SnellCauchyReal}
\boxed{\frac{\cos \theta^{L}_i}{c_i} = \frac{\cos \theta^{L}_r}{c_L} = \frac{\cos \theta^{SV}_r}{c_S} = \frac{\cos \theta^{L}_t}{c_L^{+}} = \frac{\cos \theta^{SV}_t}{c_S^{+}}}\, ,
\end{equation}
where we have to choose the speed of the incident wave $c_i=c_L$ if it is longitudinal or $c_i=c_S$ if it is shear.

As already remarked, once the angles of the different propagative waves are computed, the only unknowns in the solutions \eqref{eqsolplusprop} and \eqref{eq:solminusprop} are the scalar amplitudes $a$, which can be computed as done in subsections \ref{sec:CCLSV} and \ref{sec:CCSH}, by imposing boundary conditions. The treatise made in this section does not add new features to the previous considerations made in section \ref{sec:ReflTransCC}, but allows us to visualize the traveling waves according to the classical \textbf{Snell's law} and to recover classical results concerning the dispersion curves. Clearly, such reasoning cannot be repeated for Stoneley waves, for which the more general digression made in section \ref{sec:ReflTransCC} must be addressed. 

\section{Basics on dispersion curves analysis for bulk wave propagation in relaxed micromorphic media}\label{sec:dispersion}
Since it is useful for the comparison with the literature, we recall here the classical analysis of dispersion curves for the considered relaxed micromorphic medium (see \cite{madeo2016reflection,madeo2016complete,dagostino2017panorama}).
 
To that end, we make the hypothesis of propagative waves (see eq. \eqref{relaxedplanewave1}) and we show the plots of the frequency $\omega$ against the wave-number $\kabs$ which are known as dispersion curves.

We will show that some frequency ranges, known as band-gaps, exist, such that for a given frequency in this range, no real value of the wave-number $\kabs$ can be found. This basically means that the hypothesis of propagative wave is not satisfied in this range of frequencies and the solution must be more generally written as 
\begin{equation}e^{i\left( x_1 k_1 + x_2 k_2 - \omega t\right)},\end{equation} where $k_1$ and $k_2$ are both imaginary, giving rise to evanescent waves (i.e. waves decaying exponentially in both the $x_1$ and the $x_2$ direction).

We explicitly remark here, that this treatise on dispersion curves in the relaxed micromorphic model has already been performed in \cite{madeo2016reflection,madeo2016complete,dagostino2017panorama}, but we recall it here for the $2$D case in order to have a direct idea on the band-gap region of the considered medium. The reader who is uniquely interested in the reflective properties of interfaces and not in the bulk properties of the relaxed micromorphic model can entirely skip this section.	

From now on, we set the following parameter abbreviations for characteristic speeds and frequencies: 
\small
\begin{align}
&c_m = \sqrt{\frac{\mue \Lc}{\eta}},\quad c_s = \sqrt{\frac{\mue + \muc}{\rho}},\quad c_p=\sqrt{\frac{\lame +2\mue}{\rho}},\quad c_l=\sqrt{\frac{\lame + \mue - \muc}{\rho}},\nonumber \\
&\omega_s = \sqrt{\frac{2(\mue + \mumic)}{\eta}}, \quad \omega_r =\sqrt{\frac{2\muc}{\eta}},\quad \omega_t = \sqrt{\frac{\mumic}{\eta}},\quad \omega_l = \sqrt{\frac{\lammic + 2 \mumic}{\eta}},\label{eq:charseppesfreqRMM}\\
&\omega_p = \sqrt{\frac{(3\lame + 2\mue) + (3\lammic + 2\mumic)}{\eta}}. \nonumber 
\end{align}
\normalsize

As done for the Cauchy case, we suppose that the involved kinematic fields only depend on $x_1$ and $x_2$ (no dependence on the out-of-plane variable $x_3$), i.e.
\begin{equation}\label{eq:uP2D}
u = (u_1(x_1,x_2,t),u_2(x_1,x_2,t),u_3(x_1,x_2,t))^T,\quad P=(P_1^T(x_1,x_2,t),P_2^T(x_1,x_2,t),P_3^T(x_1,x_2,t))^T,
\end{equation}
where we recall that, according to our notation, $P^T_i, \text{ } i=1,2,3$ are the rows of the micro-distortion tensor $P$.

We plug $u$ and $P$  from \eqref{eq:uP2D}  into \eqref{eq:relaxedeqns}. The resulting system of equations is presented in Appendix \ref{appendixRelaxed1} in component-wise notation. We proceed to make the following change of variables which are motivated by the Cartan-Lie decomposition of the tensor $P$:
\begin{align}
&P^S = \frac{1}{3}(P_{11} + P_{22} + P_{33}), \quad P^{D}_1 = P_{11} - P^S, \quad P^D_2 = P_{22} - P^S, \quad P_{(1\gamma)} = \frac{1}{2}(P_{1\gamma} + P_{\gamma 1}), \nonumber \\
&P_{[1\gamma]} = \frac{1}{2}(P_{1\gamma} - P_{\gamma 1}), \quad P_{(23)} = \frac{1}{2}(P_{23} + P_{32}), \quad P_{[23]} = \frac{1}{2}(P_{12} - P_{21}), \label{eq:CartanLie}
\end{align}
with $\gamma = 2,3$.
We can then collect the variables which are coupled (see the equations presented in Appendix \ref{appendixRelaxed2}) as
\begin{align}
v^1 &= \left(u_1, u_2, P_1^D, P_2^D, P^S, P_{(12)}, P_{[12]}\right)^T, \label{eq:v1}\\
v^2 &= \left(u_3, P_{(13)}, P_{[13]}, P_{(23)}, P_{[23]}\right)^T. \label{eq:v2}
\end{align}
We make the following plane-wave ansatz: 
\begin{align}
v^1 &= \widehat{\phi}\,e^{i(\left\langle x,k \right\rangle - \omega t)}=\widehat{\phi}\,e^{i(x_1 k_1 + x_2 k_2 - \omega t)},\label{eq:planewavemicromorphic1}\\
v^2 &= \widehat{\chi}\,e^{i(\left\langle x,\widetilde{k} \right\rangle - \omega t)}=\widehat{\chi}\,e^{i(x_1\widetilde{k}_1  +x_2 \widetilde{k}_2  - \omega t)},\label{eq:planewavemicromorphic2}
\end{align}
and end up with two mutually uncoupled systems of the form (see Appendix \ref{appendixRelaxed3} for the explicit form of the matrices $A_1$ and $A_2$)
\begin{equation} \label{eq:systemmicro}
A_1 \cdot \widehat{\phi} = 0,\quad
A_2 \cdot \widehat{\chi} = 0, 
\end{equation}
where $A_1 \in \C^{7\times 7}$, $A_2\in \C^{5\times 5}$, $\widehat{\phi}	\in \C^7$ and $\widehat{\chi}\in \C^5$.  Closer examination of the first system reveals that the components of the kinematic fields involved in these equations are the first and second only, while in the second system only components involving the out-of-plane direction $x_3$ are always present in every equation. This means that, in analogy to the case of Cauchy media, we have the same kind of uncoupling between movement in the $(x_1x_2)-$ plane (in-plane) and in the $(x_2x_3)-$ plane (out-of-plane). There is, however, no immediate distinction of longitudinal and shear waves. 

\subsection{In-plane variables}\label{sec:planewaverelaxedA1}
We assume a propagating wave in which case the plane-wave ansatz is
\begin{equation} \label{relaxedplanewave1}
v^1 =\widehat{\phi}\, e^{i\left(\kabs\left(x_1\xi_1  + x_2\xi_2 \right)-\omega t\right)}, 
\end{equation}
where $\xi=(\xi_1,\xi_2)^T$ is a real unit vector and $\widehat{\phi}$ is the vector defined of amplitudes. 

The polynomial $\det A_1$ is of degree $14$ in $\omega$ and of degree $10$ in $\kabs$. Solving the equation $\det A_1=0$ with respect to $\omega$ gives fourteen solutions of the form
\begin{equation}\label{eq:omRMM1}
\omega(\kabs)\, =\, \pm \omega_1(\kabs),\,\,\pm \omega_2(\kabs), \,\,\pm \omega_3(\kabs), \,\,\pm \omega_4(\kabs), \,\, \pm \omega_5(\kabs),\,\,\pm \omega_6(\kabs),\,\, \pm \omega_7(\kabs);
\end{equation}
while solving $\det A_1 = 0$ with respect to to $\kabs$, gives ten solutions of the form 
\begin{equation}\label{eq:kRMM1}
\kabs(\omega)\,=\,\pm |k^{\op}|(\omega),\,\,\pm |k^{\twop}|k(\omega),\,\,\pm |k^{\thp}|(\omega),\,\,\pm |k^{\fp}|(\omega),\,\,\pm |k^{\fip}|(\omega);
\end{equation}
in both cases, we keep only the positive values because the wave is traveling in the $x_1>0$ direction. 

Plotting the functions $\omega_i(\kabs)$, $i=1,\ldots,7$ in the $(\omega,\kabs)-$ plane gives us the dispersion curves for plane waves propagating in a relaxed micromorphic medium in two space dimensions (see Figure \ref{fig:dispA1}). The results concerning the dispersive behavior of relaxed micromorphic media are presented in \cite{madeo2016reflection,madeo2016complete} and are recalled in Appendix \ref{appendixRelaxed41} for the sake of completeness. 
\begin{figure}[H]
	\centering
	\includegraphics[scale=0.65]{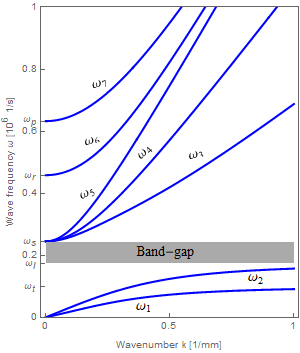} 
	\caption{\small The in-plane dispersion curves of a plane wave propagating in an isotropic relaxed micromorphic continuum in two space dimensions. The gray region denotes the band-gap, in which propagation cannot occur.}
	\label{fig:dispA1}
\end{figure}

\subsection{Out-of-plane-variables}
Analogously to the case of in-plane variables, we assume a propagating wave in which case the plane-wave ansatz is
\begin{equation} \label{relaxedplanewave2}
v^2 =\widehat{\chi}\, e^{i\left(|\widetilde{k}|\left(x_1\xi_1  + x_2\xi_2 \right)-\widetilde{\omega} t\right)}, 
\end{equation}
where $\xi=(\xi_1,\xi_2)^T$ is a real unit vector and $\widehat{\chi}$ is the vector of amplitudes. 

The polynomial $\det A_2$ is of degree $10$ in $\widetilde{\omega}$ and of degree $8$ in $|\widetilde{k}|$. Solving the equation $\det A_2=0$ with respect to $\widetilde{\omega}$ gives ten solutions of the form
\begin{equation}\label{eq:omRMM2}
\widetilde{\omega}(|\widetilde{k}|)\, =\, \pm \widetilde{\omega}_1(|\widetilde{k}|),\,\,\pm \widetilde{\omega}_2(|\widetilde{k}|),\,\,\pm \widetilde{\omega}_3(|\widetilde{k}|),\,\,\pm \widetilde{\omega}_4(|\widetilde{k}|),\,\, \pm \widetilde{\omega}_5(k);
\end{equation}
while solving $\det A_2 = 0$ with respect to to $\widetilde{k}$, gives eight solutions 
\begin{equation}\label{eq:kRMM2}
|\widetilde{k}|(\widetilde{\omega})\,=\,\pm|\widetilde{k}^{\op}|(\widetilde{\omega}),\,\,\pm|\widetilde{k}^{\twop}|(\widetilde{\omega}),\,\,\pm|\widetilde{k}^{\thp}|(\widetilde{\omega}),\,\,\pm|\widetilde{k}^{\fp}|(\widetilde{\omega}).
\end{equation}

Once again, we only consider the positive values of the $\widetilde{\omega}$'s and $|\widetilde{k}|$'s  since the waves are traveling in the $x_1>0$ direction. The dispersion curves for the out-of-plane variables are presented in Figure \ref{fig:dispA2}.

\begin{figure}[h!]
	\centering
	\includegraphics[scale=0.65]{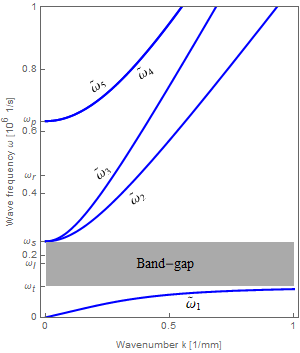} 
	\caption{\small The out-of-plane dispersion curves of a wave propagating in a relaxed micromorphic continuum in two space dimensions. The gray region depicts the band-gap for out-of-plane quantities.}
	\label{fig:dispA2}
\end{figure}

Extra details are given in \cite{madeo2016reflection} and recalled in Appendix \ref{appendixRelaxed42}.

\section{Reflective properties of a Cauchy/relaxed micromorphic interface}\label{sec:reflCauchyRMM}

In this section, closely following what was done for the Cauchy case in section \ref{sec:Cauchywaveprop}, we look for non-trivial solutions of the equations \eqref{eq:systemmicro} imposing $\det A_1= 0$ and $\det A_2 = 0$. Once again, we fix the second component $k_2$ (resp. $\widetilde{k}_2$) of the wave-vector and solve these equations with respect to $k_1$ (resp. $\widetilde{k}_1$).\footnote{We will show later on that also in the case of an interface between a Cauchy and a relaxed micromorphic medium, the component $k_2$ of the wave-vector can be considered to be known when imposing the jump conditions holiding at the interface} The expressions for the solutions of these equations are quite complex and we do not present them explicitly here. As discussed before, we find five and four solutions for the two systems respectively\footnote{Since, as already stated, in this case we cannot a priori distinguish which modes are longitudinal and which are shear, we slightly shift our notation and use numbers in parentheses instead of describing the nature of the mode.}
\begin{equation}\label{eq:pmk}
\pm k_1^{\op}(k_2,\omega),\,\,\pm k_1^{\twop}(k_2,\omega),\,\,\pm k_1^{\thp}(k_2,\omega),\,\,\pm k_1^{\fp}(k_2,\omega),\,\,\pm k_1^{\fip}(k_2,\omega),
\end{equation}
for the in-plane problem and
\begin{equation}\label{eq:pmktil}
\pm \widetilde{k}_1^{\op}(\widetilde{k}_2,\omega),\,\,\pm \widetilde{k}_1^{\twop}(\widetilde{k}_2,\omega),\,\,\pm \widetilde{k}_1^{\thp}(\widetilde{k}_2,\omega),\,\,\pm \widetilde{k}_1^{\fp}(\widetilde{k}_2,\omega), 
\end{equation}
for the out-of-plane problem.

Such solutions for $k_1$ (resp. $\widetilde{k}_1$) depend on the second component $k_2$ (resp. $\widetilde{k}_2$)  of the wave-vector and on the frequency, but of course also on the values of the material parameters of the relaxed micromorphic model. We plug these solutions of the characteristic polynomials into the matrix $A_1$ (resp. $A_2$) and calculate for each different $k$ (resp. $\widetilde{k}$) the five (resp. four) nullspaces of the matrix. We find 
\begin{equation}\label{eq:phiRMM}
\widehat{\phi}^{\op},\widehat{\phi}^{\twop},\widehat{\phi}^{\thp},\widehat{\phi}^{\fp},\widehat{\phi}^{\fip}, 
\end{equation}
\begin{equation} \label{eq:chiRMM}
\widehat{\chi}^{\op},\widehat{\chi}^{\twop},\widehat{\chi}^{\thp},\widehat{\chi}^{\fp}, 
\end{equation}
as solutions to the equations $A_1 \cdot \widehat{\phi} = 0$ and $A_2 \cdot \widehat{\chi} =0$, respectively. We normalize these vectors, thus introducing the normal vectors
\begin{equation}\label{eq:phichinormal}
\phi^{(i)} = \frac{1}{|\widehat{\phi}^{(i)}|}\widehat{\phi}^{(i)}, \quad \chi^{(j)} = \frac{1}{|\widehat{\chi}^{(j)}|}\widehat{\chi}^{(j)},
\end{equation}
$i=1,\ldots 5$, $j=1,\ldots 4$. Finally, we can write the solution to equations \eqref{eq:relaxedeqns} as
\small
\begin{equation}\label{eq:solmicro1} 
v^1 = \alpha_1 \phi^{\op} e^{i\left(\left\langle x,k^{\op} \right\rangle - \omega t\right)} + \alpha_2 \phi^{\twop} e^{i\left(\left\langle x,k^{\twop} \right\rangle - \omega t\right)}+\alpha_3 \phi^{\thp} e^{i\left(\left\langle x,k^{\thp} \right\rangle - \omega t\right)}+\alpha_4 \phi^{\fp} e^{i\left(\left\langle x,k^{\fp} \right\rangle - \omega t\right)}+\alpha_5 \phi^{\fip} e^{i\left(\left\langle x,k^{\fip} \right\rangle - \omega t\right)},
\end{equation}

\begin{equation}\label{eq:solmicro2}
v^2 = \beta_1 \chi^{\op} e^{i\left(\left\langle x,\widetilde{k}^{\op} \right\rangle - \omega t\right)} + \beta_2 \chi^{\twop} e^{i\left(\left\langle x, \widetilde{k}^{\twop} \right\rangle - \omega t\right)}+\beta_3 \chi^{\thp} e^{i\left(\left\langle x,\widetilde{k}^{\thp} \right\rangle - \omega t\right)}+\beta_4 \chi^{\fp} e^{i\left(\left\langle x,\widetilde{k}^{\fp} \right\rangle - \omega t\right)}, 
\end{equation}\normalsize
where $\alpha_i, \beta_j \in \C$ for $i=1,\ldots 5$, $j=1,\ldots, 4$ are the unknown amplitudes of the different modes of propagation.\footnote{We recall again that, once the eigenvalue problem is solved (i.e. once the wave-vector $k$ (resp. $\widetilde{k}$ is known) and the eigenvectors $\phi$ (resp. $\chi$) are computed, the only unknowns of the problem remain the five (resp. four) amplitudes $\alpha$ (resp $\beta$), which can be computed by imposing boundary conditions.}

We explicitly remark that expressions \eqref{eq:pmk} and \eqref{eq:pmktil} for the first component $k_1$ and $\widetilde{k}_1$ of the wave-vectors, can give rise, similarly to the Cauchy case, to different scenarios when varying the value of the frequency $\omega$ and the material parameters. As a matter of fact, we briefly remarked before that $k_2$ can be considered to be known when imposing jump conditions. Indeed, following analogous steps to those performed  to obtain equation \eqref{eq:SnellCauchy} for the interface between two Cauchy media, we can impose the continuity of displacements between a Cauchy and a relaxed micromorphic medium. Considering the first component of the vector equation for the continuity of displacement, in which the plane-wave ansatz has been used, one can find, when imposing a longitudinal incident wave on the Cauchy side\footnote{When imposing a longitudinal incident wave on the Cauchy side, $k_2^{L,i}$ is considered to be known. The same reasoning holds when imposing an incident SV wave; in this case, $k_2^{L,i}$ must be replaced by $k_2^{SV,i}$ in eq. \eqref{eq:SnellmicromorphicLSV}.}
\begin{empheq}[box=\fbox]{gather}
k_2^{L,i} = k_2^{L,r} = k_2^{SV,r} = k_2^{\op,t} = k_2^{\twop,t} = k_2^{\thp,t} = k_2^{\fp,t} = k_2^{\fip,t}. \label{eq:SnellmicromorphicLSV}
\\\notag\text{\textbf{Generalized in-plane Snell's Law}}
\end{empheq}
On the other hand, when imposing an out-of-plane shear incident wave, the continuity of displacement at the interface gives
\begin{empheq}[box=\fbox]{gather}
k_2^{SH,i} =  k_2^{SH,r} = \widetilde{k}_2^{\op,t} = \widetilde{k}_2^{\twop,t} = \widetilde{k}_2^{\thp,t} = \widetilde{k}_2^{\fp,t}. \label{eq:SnellmicromorphicSH}
\\\notag\text{\textbf{Generalized out-of-plane Snell's Law}}
\end{empheq}
Equations \eqref{eq:SnellmicromorphicLSV} and \eqref{eq:SnellmicromorphicSH} tell us that, when fixing the incident wave in the Cauchy medium to be longitudinal ($k_2^{L,i}$ known), in-plane shear ($k_2^{SV,i}$ known) or out-of-plane shear ($k_2^{SH,i}$ known), the second components of all the reflected and transmitted wave-vectors are known. They are the \textbf{generalized Snell's law} for the case of a Cauchy/relaxed micromorphic interface. As before, this traces two possible scenarios, given that the value of $k_2$ for the incident wave is always supposed to be real and positive (propagative wave)
\begin{enumerate}
\item both $k_1$ and $k_2$ (resp. $\widetilde{k}_1, \widetilde{k}_2$) are real (when computing $k_1$ or $\widetilde{k}_1$ via \eqref{eq:pmk} or \eqref{eq:pmktil} respectively) so that one has propagative waves.
\item $k_2$ (resp. $\widetilde{k}_2$) is real and $k_1$ (resp. $\widetilde{k}_1$), when computed via \eqref{eq:pmk} (resp. \eqref{eq:pmktil}) is imaginary, so that one has Stoneley waves propagating only along the interface and decaying away from it.
\end{enumerate}
\begin{figure}[H]
	\centering
	\includegraphics[scale=0.7]{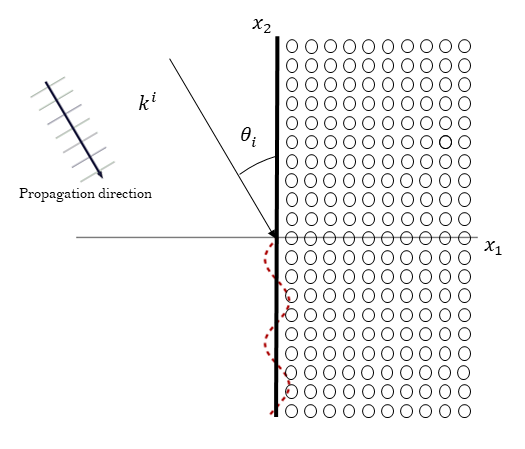}
	\caption{\small Simplified representation of the onset of a interface wave (in red) propagating along the interface between a homogeneous solid and a metamaterial. Depending on the relative stiffnesses of the two media, each of the existing low and high-frequency modes can either become Stoneley or remain propagative.}
	\label{fig:StoneleyMetamaterial}
\end{figure}
Depending on the values of the frequency, of the material parameters and of the angle of incidence, each of the five in-plane waves, or of the four out-of-plane waves, can be either propagative or Stoneley. Thus, Stoneley waves can appear at the considered homogeneous solid/metamaterial interface (see Fig. \ref{fig:StoneleyMetamaterial} for a simplified illustration), both for low and for high-frequency modes.

\subsection{Determination of the reflection and transmission coefficients in the case of a relaxed micromorphic medium}
As for the flux, the normal outward pointing vector to the surface (the $x_2$ axis) is $\nu = (-1,0,0)$. This means that in the expression \eqref{eq:relaxedflux} for the flux, we need only take into account the first component. According to our definition \eqref{eq:relaxedflux}, we have
\begin{equation}\label{eq:relaxedflux1}
H_1=-u_{i,t}\widetilde{\sigma}_{i1}-m_{ih}P_{ij,t}\eps_{jh1}, \quad i,j,h\in\{1,2,3\}.
\end{equation}


This equation for the flux must now be written with respect to the new variables $v^1$ and $v^2$. It is a tedious but easy calculation to see that the following holds
\begin{equation}\label{eq:relaxedfluxmatrixinplane}
\widetilde{H}:=H_1=v^1_{,t}\cdot (H^{11}\cdot v^1_{,1} +H^{12}\cdot  v^1_{,2}+H^{13}\cdot v^1), 
\end{equation}
for the in-plane problem and 
\begin{equation}\label{eq:relaxedfluxmatrixoutofplane}
\widetilde{H}:=H_1=v^2_{,t}\cdot (H^{21}\cdot v^2_{,1} +H^{22}\cdot  v^2_{,2}+H^{23}\cdot v^2), 
\end{equation}
for the out-of plane problem, where $H^{11}, H^{12}, H^{13},$ and $H^{21}, H^{22}, H^{23}$ are matrices of suitable dimensions (found in Appendix \ref{appendixRelaxed6}).

Having calculated the ``transmitted'' flux, we can now look at the reflection and transmission coefficients for the case of a Cauchy/relaxed micromorphic interface. 

To that end, we again define
\begin{equation}
J^i = \frac{1}{T}\int_0^{T} H^i(x,t) dt, \quad J^r = \frac{1}{T}\int_0^{T} H^r(x,t) dt, \quad J^t = \frac{1}{T}\int_0^{T} H^t(x,t) dt,
\end{equation}
where $H^i=H^{L,i}$, $H^r=H^{L,r}+H^{SV,r}$ and $H^t = \widetilde{H}$. Then the reflection and transmission coefficients are 
\begin{align}\label{eq:reflcoeff2}
\mathcal{R}=\frac{J^r}{J^i}, \quad \mathcal{T}=\frac{J^t}{J^i}.
\end{align}

In order to  easily compute these coefficients, we again employ Lemma \ref{Lemma1}. 
Finally, once again we have that $\mathcal{R} + \mathcal{T}=1$. 

In the case of a Cauchy/relaxed micromorphic interface, the dependency of the fluxes on the frequency $\omega$ is maintained. This is due to the fact that the amplitudes needed to calculate the flux depend on $\omega$ (dispersive response), something which is not the case in the Cauchy/Cauchy interface, as was evident in the previous section.

\section{Results}\label{sec:Results}
In this section we present our results concerning the reflective properties of an interface between a Cauchy medium and a relaxed micromorphic medium. We will show that, at low frequencies, the considered interface can be regarded as an interface between a Cauchy medium and a second Cauchy medium, equivalent to the relaxed micromorphic oneand with macroscopic stiffnesses $\lambda_{\text{macro}}$ and $\mu_{\text{macro}}$, when suitable boundary conditions are imposed.

Moreover, we will be able to show that critical angles for the incident wave can be identified in the low-frequency regime, beyond which we can observe the onset of Stoneley waves. These angles are computed from the relations established in Table \ref{table:StoneleyT}. 

In order to present explicit numerical results for the reflective properties of the interface between a Cauchy and a relaxed micromorphic medium, we chose the values for the parameters of the relaxed micromorphic medium as shown in Table \ref{table:parameters}. We explicitly remark that other values of such parameters could be chosen, which would be more or less close to real metamaterials parameters (\cite{madeo2016first,madeo2017relaxed,dagostino2018effective}). Nevertheless, the basic results which we want to show in the present paper are not qualitatively affected by this choice since they only depend on the relative stiffness of the two media which are considered on the two sides and not on the absolute values of such stiffnesses.
\begin{table}[ht]
	\centering
	\begin{tabular}{c c c c c c c c}
		$\rho$ [$\mathrm{kg/m^3}$]&   $\eta$ [$\mathrm{kg/m}$] & $\muc$ [$\mathrm{Pa}$]& $\mue$ [$\mathrm{Pa}$] & $\mumic$ [$\mathrm{Pa}$] & $\lammic$ [$\mathrm{Pa}$] & $\lame$ [$\mathrm{Pa}$] & $L_c$ [$\mathrm{m}$]\\
		\hline
		$2000$&  $10^{-2}$ & $2\times 10^9$ & $2\times 10^8$ & $10^8$ & $10^8$ & $4\times 10^8$  & $10^{-2}$
	\end{tabular}
	\caption{\small Numerical values of the constitutive parameters chosen for the relaxed micromorphic medium.}
	\label{table:parameters}
\end{table}

We can now use the following homogenization formulas, presented in \cite{barbagallo2017transparent,dagostino2018effective}, to compute the equivalent macroscopic coefficients of the Cauchy medium which is approximating the relaxed micromorphic medium at low frequencies
\begin{equation}\label{eq:macroparameters}
\mu_{\text{macro}}=\frac{\mue \,\mumic}{\mue +\mumic}, \qquad
2\mu_{\text{macro}} + 3\lambda_{\text{macro}} = \frac{(2\mue + 3\lame)(2\mumic + 3\lammic)}{2(\mue+\mumic) + 3(\lame + \lammic)}. 
\end{equation}
Note that the Cosserat couple modulus $\muc$ does not appear in the homogenization formulas  \eqref{eq:macroparameters}.

Using formulas \eqref{eq:macroparameters}, we compute the stiffnesses $\lambda_{\text{macro}}$ and $\mu_{\text{macro}}$ of the Cauchy medium which is equivalent to the relaxed micromorphic medium of Table \ref{table:parameters} in the low-frequency regime, as in the following Table:
\begin{table}[H]
	\centering
	\begin{tabular}{c c c}
		$\rho$ [$\mathrm{kg/m^3}$]& $\lambda_{\text{macro}}$ [$\mathrm{Pa}$] & $\mu_{\text{macro}}$ [$\mathrm{Pa}$] \\
		\hline
		$2000$&	$8.25397 \times 10^7$& $6.66667 \times 10^7$ 
	\end{tabular}
	\caption{\small Macro parameters of the equivalent Cauchy medium corresponding to the relaxed medium of Table \ref{table:parameters} at low frequencies.}
	\label{table:macroparameters}
\end{table}

At this point, we will consider the two cases in which the Cauchy medium on the ``$-$'' side (where the incident wave is traveling) is stiffer or softer than the equivalent Cauchy medium on the ``$+$'' side. We will show how, as expected, this difference in stiffness affects the onset of Stoneley waves at low frequencies and, as a consequence, the transmission patterns across the considered interface. 

We will also show that the relaxed micromorphic model is able to predict the appearance of Stoneley waves at higher frequencies, which are substantially microstructure-related.

We will finally show that the relaxed micromorphic model also allows for the description of wide frequency bounds, for which extraordinary reflection is observed. Such frequency bounds go beyond the band-gap region and are related to the presence of the interface, as well as the relative mechanical properties of the considered media. In some cases, high-frequency critical angles discriminating between total transmission and total reflection can also be identified. 

 \subsection{Cauchy medium which is ``stiffer'' than the relaxed micromorphic one}\label{sec:ResultsStiffer}
In this section we present the reflective properties of a Cauchy/relaxed micromorphic interface for which we consider that the Cauchy medium on the left side is ``stiffer'' than the corresponding macroscopic parameters of the relaxed micromorphic medium in Table \ref{table:macroparameters} on the right side. To that end, we chose the material parameters of the left Cauchy medium to be those presented in Table \ref{table:Cauchyparameters} and we explicitly remark that these values are greater than those of Table \ref{table:macroparameters}, which are relative to the equivalent Cauchy medium corresponding to the considered relaxed micromorphic one.	
\begin{table}[ht]
	\centering
	\begin{tabular}{c c c}
		$\rho$ [$\mathrm{kg/m^3}$] &$\lambda$ [$\mathrm{Pa}$] & $\mu$ [$\mathrm{Pa}$] \\
		\hline
		$2000$ & $4\times 10^8$ & $2\times 10^8$
	\end{tabular}
	\caption{\small Lam\'e parameters and mass density of the Cauchy medium on the left side of the considered Cauchy/relaxed micromorphic interface.}
	\label{table:Cauchyparameters}
\end{table}

For the chosen values of the constitutive parameters, the critical angles of the incident wave giving rise to Stoneley waves can be calculated using Tables \ref{table:StoneleyT} and \ref{table:StoneleyR}. As already mentioned, this approximation for the Cauchy/relaxed micromorphic interface is valid in the low-frequency regime (see \cite{neff2017real}), where the relaxed micromorphic medium is well approximated by its Cauchy counterpart with macroscopic stiffnesses $\lambda_{\text{macro}}$ and $\mu_{\text{macro}}$. We explicitly identify in the following figures \ref{fig:Tfree} and \ref{fig:Tfixed} the ``low-frequency regime'', where the aforementioned approximation is valid, as well as the critical angles governing the onset of Stoneley waves.

The explicit computed values of these critical angles are given in Table \ref{table:criticalangles1}.
\begin{table}[H]
	\centering
	\begin{tabular}{c c c c c}
		Incident wave & $\theta_{\text{crit}}^{L,r}$  & $\theta_{\text{crit}}^{L,t}$& $\theta_{\text{crit}}^{SV,t}$& $\theta_{\text{crit}}^{SH,t}$ \\
		\hline
		L& $-$ & $-$ & $-$ & $-$ \\
		\hline 
		SV & $\frac{33 \pi}{100}$ & $\frac{17\pi}{200}$ &$-$  & $-$ \\
		\hline 
		SH	& $-$ & $-$ & $-$ & $-$ \\
	\end{tabular}
	\caption{\small Critical angles governing the onset of Stoneley waves at the Cauchy/equivalent Cauchy interface between the two Cauchy media given in Tables \ref{table:Cauchyparameters} and \ref{table:macroparameters}, respectively. These values are computed according to the formulas given in Tables \ref{table:StoneleyT} and \ref{table:StoneleyR}. The superscripts $r$ and $t$ stand for ``reflected'' and ``transmitted''.}
	\label{table:criticalangles1}
\end{table}

\begin{figure}[H]
	\begin{subfigure}{.5\textwidth}
		\centering
		\includegraphics[scale=0.45]{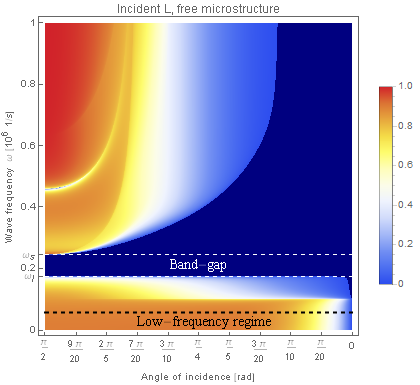}
		\caption{}
		\label{fig:FreeL}
	\end{subfigure}%
	\begin{subfigure}{.5\textwidth}
		\centering
		\includegraphics[scale=0.45]{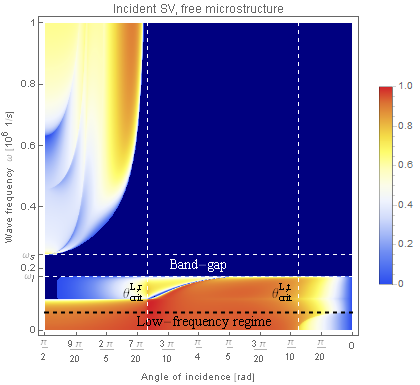}
		\caption{}
		\label{fig:FreeSV}
	\end{subfigure}\\
	\centering
	\begin{subfigure}{.5\textwidth}
		\centering
		\includegraphics[scale=0.45	]{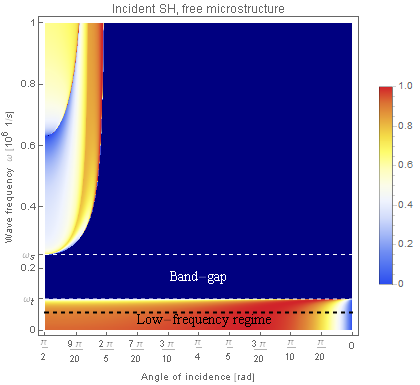}
		\caption{}
		\label{fig:FreeSH}
	\end{subfigure}
	\caption{\small Transmission coefficients as a function of the angle of incidence $\theta_i$ and of the wave-frequency $\omega$ for L (a), SV (b) and SH (c) incident waves for the case of macro-clamp with free microstructure. The origin coincides with normal incidence ($\theta_i=\pi/2$), while the angle of incidence decreases towards the right until it reaches the value $\theta_i=0$, which corresponds to the limit case where the incidence is parallel to the interface. The band-gap region is highlighted by two dashed horizontal lines, where, as expected, we observe no transmission. The low-frequency regime is highlighted by the bottom horizontal dashed line, while the critical angles for the onset of Stoneley waves are denoted by vertical dashed lines. The dark blue zone shows that no transmission takes place, while the gradual change from dark blue to red shows the increase of transmission, red being total transmission.}
	\label{fig:Tfree}
\end{figure}

Figure \ref{fig:Tfree} shows the transmission coefficient for the considered Cauchy/relaxed micromorphic interface, as a function of the angle of incidence and of the frequency, when the microstructure is free to move at the interface ($P$ is left arbitrary at the interface). The coloring of this plot is such that the dark blue regions mean zero transmission, while the gradual change towards red is the increase in transmission (red is total transmission). Before commenting on the details of the behavior of the transmission coefficient, we recall that the case of free microstructure boundary condition is the only one which allows us to precisely obtain a Cauchy/equivalent Cauchy interface in the low-frequency regime, something which is not possible when imposing the fixed microstructure boundary condition ($P_{ij}=0, \text{ } i=2,3, \text{ } j=1,2,3$) at the interface. As a matter of fact, it is firmly established that a relaxed micromorphic continuum is equivalent to a Cauchy continuum with stiffnesses $\lambda_{\text{macro}}$ and $\mu_{\text{macro}}$ when considering the low-frequency regime (see \cite{neff2017real}), but this is proven only for the bulk medium. When considering an interface between a Cauchy and a relaxed micromorphic medium, the latter will behave exactly as an equivalent Cauchy medium at low frequencies only if the micro-distortion tensor $P$ is left free at the interface. Indeed, this tensor will arrange its values at the interface in order to let the low-frequency reflective properties of the Cauchy/relaxed micromorphic interface be equivalent to those of a Cauchy/equivalent Cauchy interface. On the other hand, if we impose the fixed microstructure boundary conditions, the tangential components of the tensor $P$ are forced to vanish at the interface, so that the effect of the microstructure is artificially introduced in the response of the material even for those low frequencies for which the bulk material would tend to behave as an equivalent Cauchy medium. 

\vspace{-0.5cm}
Having drawn such preliminary conclusions, we can now comment Figures \ref{fig:Tfree} and \ref{fig:Tfixed} in detail. For the set of numerical values of the parameters given in Table \ref{table:Cauchyparameters} and \ref{table:macroparameters}, we established that Stoneley waves can appear in the low-frequency regime only when imposing the incident wave to be SV. In particular, the onset of Stoneley waves in the low-frequency regime can be observed in this case only for longitudinal reflected and transmitted waves when the angles of incidence are beyond $\theta_{\text{crit}}^{L,r}$ and $\theta_{\text{crit}}^{L,t}$, respectively. This fact can be retrieved in Figure \ref{fig:FreeSV}, in which an increase of the transmission coefficient can be observed in the low-frequency regime corresponding to $\theta_{\text{crit}}^{L,r}$ (Stoneley reflected waves are created, producing a decrease of the reflected normal flux and, due to energy conservation, a consequent increase of the transmitted normal flux). On the other hand, we can notice in the same figure a decrease of the transmitted energy in the low-frequency regime beyond the critical angle $\theta_{\text{crit}}^{L,t}$. This is sensible, given that beyond the value of $\theta_{\text{crit}}^{L,t}$, transmitted Stoneley waves are created, which do not contribute to propagative transmitted waves in the relaxed micromorphic continuum. We can also explicitly remark that such a decrease of transmitted energy beyond $\theta_{\text{crit}}^{L,t}$ in the low-frequency regime is much more pronounced than in the corresponding Figures \ref{fig:FreeL} and \ref{fig:FreeSH}. This means that the creation of transmitted Stoneley waves contributes to a decrease of the transmitted energy in the low-frequency regime, but a decreasing trend for the transmission coefficient is observed also for the other cases, when considering angles which are far from normal incidence. This goes along the common feeling, according to which the more inclined the incident wave is with respect to the interface, the less transmission one can expect.
The same behavior, even if qualitatively and quantitatively different, can be found in Figure \ref{fig:FixedSV}, in which an increase of transmission can be observed after $\theta_{\text{crit}}^{L,r}$ and a decrease after $\theta_{\text{crit}}^{L,t}$, also for the case of fixed microstructure boundary conditions.
\vspace{-0.3cm}
\begin{figure}[H]
	\begin{subfigure}{.5\textwidth}
		\centering
		\includegraphics[scale=0.45]{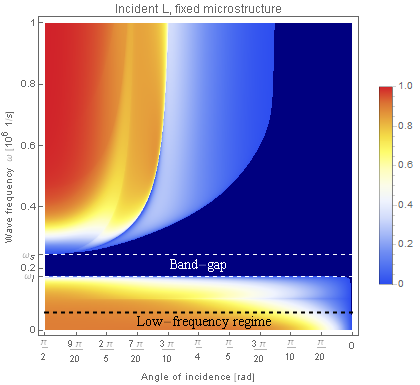}
		\caption{}
		\label{fig:FixedL}
	\end{subfigure}%
	\begin{subfigure}{.5\textwidth}
		\centering
		\includegraphics[scale=0.45]{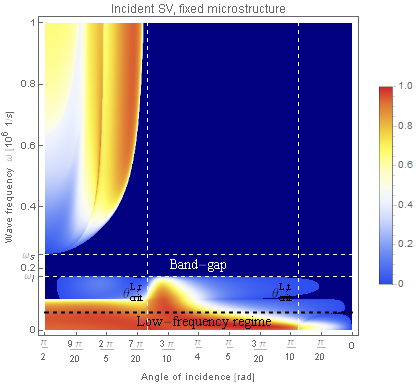}
		\caption{}
		\label{fig:FixedSV}
	\end{subfigure}\\
	\centering
	\begin{subfigure}{.5\textwidth}
		\centering
		\includegraphics[scale=0.45]{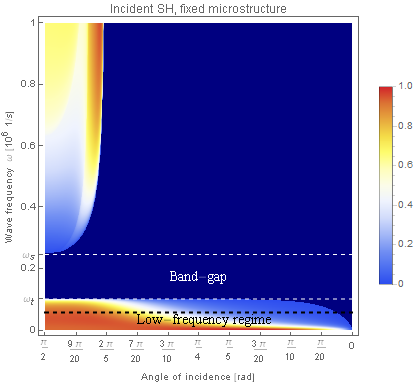}
		\caption{}
		\label{fig:FixedSH}
	\end{subfigure}
	\caption{\small Transmission coefficients as a function of the angle of incidence $\theta_i$ and of the wave-frequency $\omega$ for L (a), SV (b) and SH (c) incident waves for the case of macro-clamp with fixed microstructure. The origin coincides with normal incidence ($\theta_i=\pi/2$), while the angle of incidence decreases towards the right until it reaches the value $\theta_i=0$, which corresponds to the limit case where the incidence is parallel to the interface. The band-gap region is highlighted by two dashed horizontal lines, where, as expected, we observe no transmission. The low-frequency regime is highlighted by the bottom horizontal dashed line, while the critical angles for the onset of Stoneley waves are denoted by vertical dashed lines. The dark blue zone shows that no transmission takes place, while the gradual change from dark blue to red shows the increase of transmission, red being total transmission.}
	\label{fig:Tfixed}
\end{figure}

Direct comparison of Figures \ref{fig:Tfree} and \ref{fig:Tfixed} allows us to identify the effect that the chosen type of boundary conditions has on the transmission properties of the interface. We already remarked that, at low frequencies, common trends can be identified which are related to critical angles determining the onset of Stoneley waves at the Cauchy/equivalent Cauchy interface. Nevertheless, some differences can also be remarked which are entirely related to the choice of boundary conditions. 

Surprisingly, the effect of boundary conditions intervenes already for low frequencies, meaning that the fact of imposing the value of $P$ at the interface introduces a tangible effect of the interface microstructured properties on the overall behavior of the considered system. In particular, we can notice that the fact of forcing $P=0$ at the interface globally reduces the low-frequency transmission for angles which are much closer to normal incidence, than for the case of free microstructure. This means that the fact of considering a microstructure which is not free to vibrate at the interface, allows for microstructure-related reflections, even if the frequency is relatively low. 
Such additional reduction of transmission takes place for incident waves which are very inclined with respect to the surface ($\theta_i\leq \pi/4$).

Up to now, we only discussed the transmittive properties of the considered Cauchy/relaxed micromorphic interface on the low-frequency regime. Some of the features that we discussed on Stoneley waves 
can be retrieved by observing Figures \ref{fig:k1LF}, \ref{fig:k1SVF} and \ref{fig:k1SHF} in which the plots of the imaginary part of the first component of the wave vector $k_1$ are given for each mode of the relaxed micromorphic medium, for L, SV and SH incident waves respectively. 

\begin{figure}[H]
	\begin{subfigure}{.5\textwidth}
		\centering
		\includegraphics[scale=0.35]{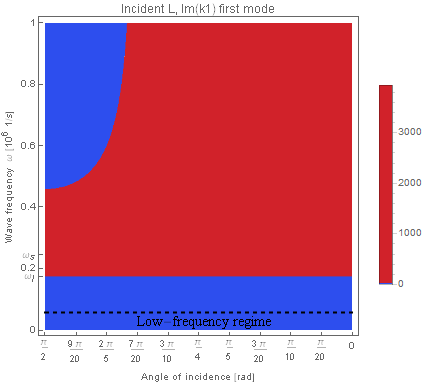}
		\caption{}
		\label{fig:k1L2F}
	\end{subfigure}%
	\begin{subfigure}{.5\textwidth}
		\centering
		\includegraphics[scale=0.35]{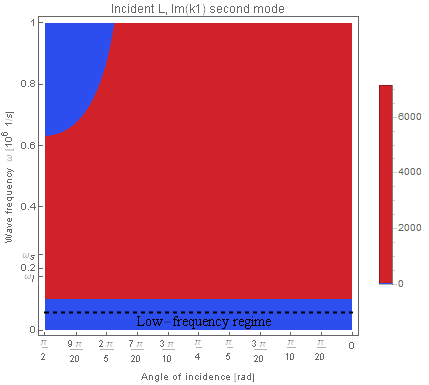}
		\caption{}
		\label{fig:k1L4F}
	\end{subfigure}\\
	\begin{subfigure}{.5\textwidth}
		\centering
		\includegraphics[scale=0.35]{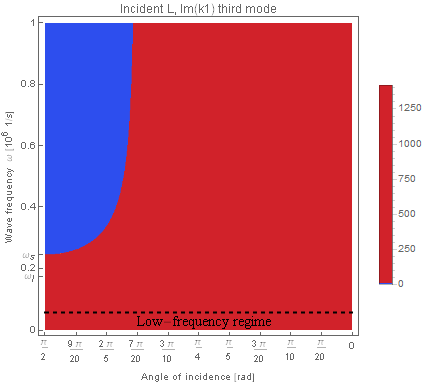}
		\caption{}
		\label{fig:k1L1F}
	\end{subfigure}%
	\begin{subfigure}{.5\textwidth}
		\centering
		\includegraphics[scale=0.35]{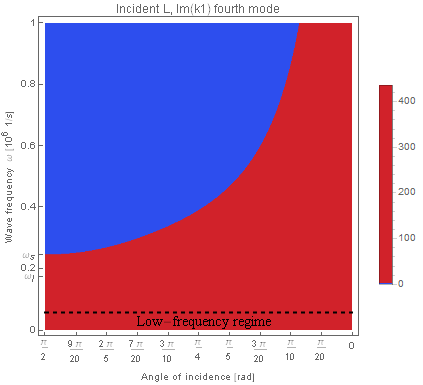}
		\caption{}
		\label{fig:k1L3F}
	\end{subfigure}
	\centering
	\begin{subfigure}{.5\textwidth}
		\centering
		\includegraphics[scale=0.35]{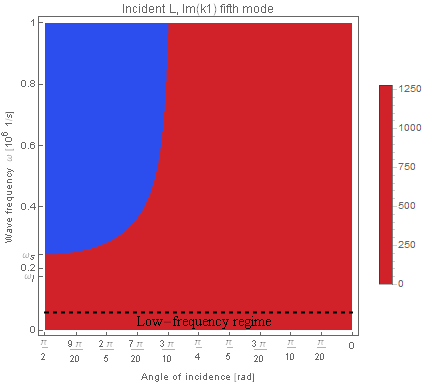}
		\caption{}
		\label{fig:k1L5F}
	\end{subfigure}
	\caption{\small Values of $\Im(k_1)$ as a function of the angle of incidence $\theta_i$ and of the wave-frequency $\omega$ for the five modes of the relaxed micromorphic medium and for the case of an incident L wave. The origin coincides with normal incidence ($\theta_i=\pi/2$), while the angle of incidence decreases towards the right until it reaches the value $\theta_i=0$, which corresponds to the limit case where the incidence is parallel to the interface. The first two modes (a) and (b) correspond to the L and SV modes for the equivalent Cauchy continuum at low frequencies. The red color in these plots means that the mode is Stoneley and does not propagate, while blue means that the mode is propagative.}
	\label{fig:k1LF}
\end{figure}

The blue region denotes $\Im(k_1)=0$ (which implies that $k_1$ is real), while $\Im(k_1)$ is not vanishing in the red regions. In other words, we can say that for each mode, the red color means that there are Stoneley waves associated to that mode. The first two modes in Figures \ref{fig:k1LF} and \ref{fig:k1SVF} correspond to L and SV Cauchy-like modes, while the first mode in Figure \ref{fig:k1SHF} is the SH Cauchy-like mode in the low-frequency regime. Since we are considering a relaxed micromorphic medium, three additional modes with respect to the Cauchy case are present both for the in-plane (Figures \ref{fig:k1LF} and \ref{fig:k1SVF}) and for the out-of-plane problem (Fig. \ref{fig:k1SHF}). For the Cauchy-like modes we can observe that at low frequencies they are always propagative, except in the case of an incident SV wave, for which Stoneley longitudinal waves appear beyond $\theta_{\text{crit}}^{L,t}$ (also Stoneley reflected waves can be observed in this case, but we do not present the plots of $\Im(k_1)$ for reflected waves to avoid overburdening).

We can note by inspecting Figures \ref{fig:k1LF}, \ref{fig:k1SVF} and (Fig. \ref{fig:k1SHF} that the presence of Stoneley waves at high frequencies is much more widespread than at low frequencies for all $5$ (resp. $4$) modes.

\begin{figure}[H]
	\begin{subfigure}{.5\textwidth}
		\centering
		\includegraphics[scale=0.35]{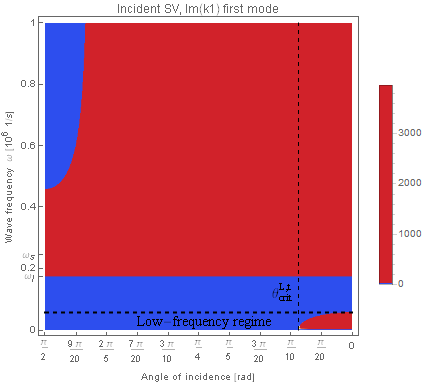}
		\caption{}
		\label{fig:k1S2F}
	\end{subfigure}%
	\begin{subfigure}{.5\textwidth}
		\centering
		\includegraphics[scale=0.35]{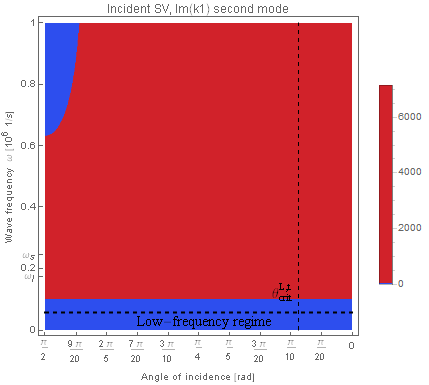}
		\caption{}
		\label{fig:k1S4F}
	\end{subfigure}\\
	\begin{subfigure}{.5\textwidth}
		\centering
		\includegraphics[scale=0.35]{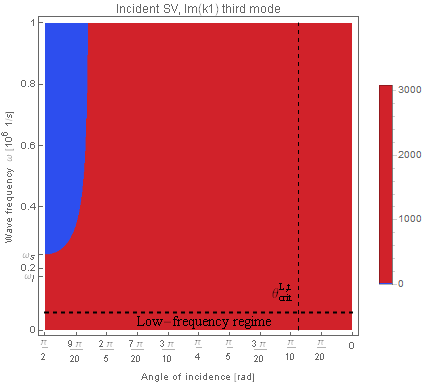}
		\caption{}
		\label{fig:k1S1F}
	\end{subfigure}%
	\begin{subfigure}{.5\textwidth}
		\centering
		\includegraphics[scale=0.35]{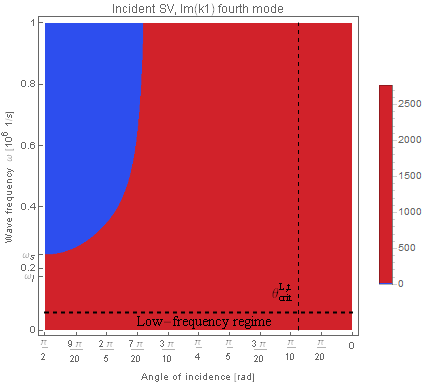}
		\caption{}
		\label{fig:k1S3F}
	\end{subfigure}
	\centering
	\begin{subfigure}{.5\textwidth}
		\centering
		\includegraphics[scale=0.4]{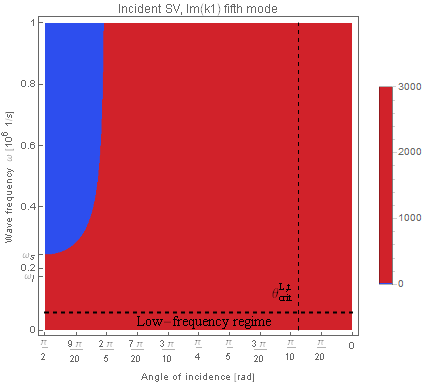}
		\caption{}
		\label{fig:k1S5F}
	\end{subfigure}
	\caption{\small Values of $\Im(k_1)$ as a function of the angle of incidence $\theta_i$ and of the wave-frequency $\omega$ for the five modes of the relaxed micromorphic medium and for the case of an incident SV wave. The origin coincides with normal incidence ($\theta_i=\pi/2$), while the angle of incidence decreases towards the right until it reaches the value $\theta_i=0$, which corresponds to the limit case where the incidence is parallel to the interface. The first two modes (a) and (b) correspond to the L and SV modes for the equivalent Cauchy continuum at low frequencies. The red color in these plots means that the mode is Stoneley and does not propagate, while blue means that the mode is propagative.}
	\label{fig:k1SVF}
\end{figure}

\begin{figure}[H]
	\centering
	\begin{subfigure}{.5\textwidth}
		\centering
		\includegraphics[scale=0.35]{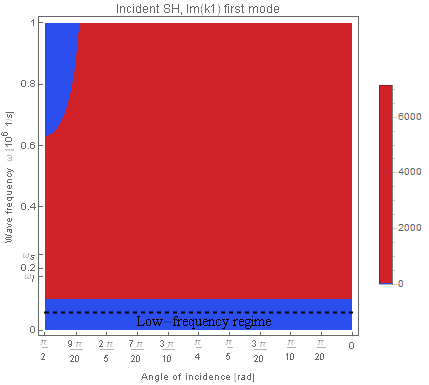}
		\caption{}
		\label{fig:k1SH3F}
	\end{subfigure}%
	\begin{subfigure}{.5\textwidth}
		\centering
		\includegraphics[scale=0.35]{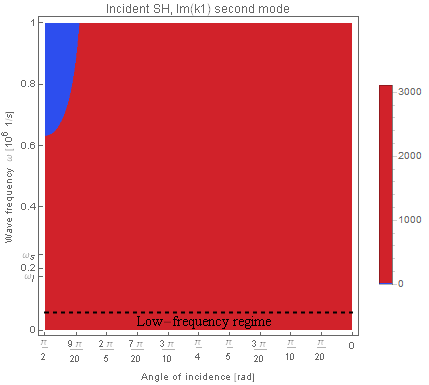}
		\caption{}
		\label{fig:k1SH1F}
	\end{subfigure}\\
	\begin{subfigure}{.5\textwidth}
		\centering
		\includegraphics[scale=0.35]{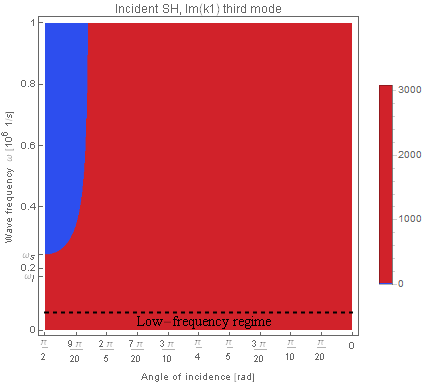}
		\caption{}
		\label{fig:k1SH2F}
	\end{subfigure}%
	\begin{subfigure}{.5\textwidth}
		\centering
		\includegraphics[scale=0.35]{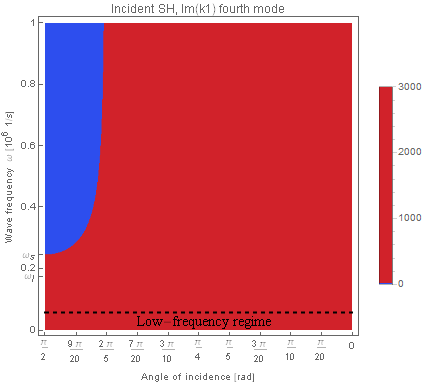}
		\caption{}
		\label{fig:k1SH4F}
	\end{subfigure}
	\caption{\small Values of $\Im(k_1)$ as a function of the angle of incidence $\theta_i$ and of the wave-frequency $\omega$ for the four modes of the relaxed micromorphic medium and for the case of an incident SH wave. The origin coincides with normal incidence ($\theta_i=\pi/2$), while the angle of incidence decreases towards the right until it reaches the value $\theta_i=0$, which corresponds to the limit case where the incidence is parallel to the interface. The first mode (a) corresponds to the SH mode for the equivalent Cauchy continuum at low frequencies. The red color in these plots means that the mode is Stoneley and does not propagate, while blue means that the mode is propagative.}
	\label{fig:k1SHF}
\end{figure}

We can observe by direct observation of Figures \ref{fig:k1LF}, \ref{fig:k1SVF} and \ref{fig:k1SHF} that high-frequency critical angles exist for each mode corresponding to which a transition from Stoneley to propagative waves takes place. The value of such critical angles depends on the frequency for the medium-frequency regime and become constant for higher frequencies.
The influence of the existence of such high-frequency critical angles can be directly observed on the patterns of the transmission coefficient in Figures \ref{fig:Tfree} and \ref{fig:Tfixed}, in which high frequency transmission is observed for angles closer to normal incidence and no transmission is reported for smaller angles due to the simultaneous presence of Stoneley waves for all modes. We can call such zones in which transmission is equal to one ``extraordinary transmission regions'' (see e.g. \cite{misseroni2016cymatics}). Such extraordinary transmission can be used as a basis for the conception of innovative systems such as selective cloaking and non-destructive evaluation.

We can finally remark that the influence of the choice of boundary conditions on the high-frequency behavior of the transmission coefficient is still present, but do not determine drastic changes on the transmission patterns (see Figures \ref{fig:Tfree} and \ref{fig:Tfixed}). 

\subsection{Cauchy medium which is ``softer'' than the relaxed micromorphic one}
In this section we present the reflective properties of a Cauchy/relaxed micromorphic interface for which we consider that the Cauchy medium on the left is ``softer'' than the relaxed micromorphic medium on the right in the same sense as in the previous section. To that end, we choose the material parameters of the left Cauchy medium to be those presented in the following Table and  we explicitly remark that these values are smaller than those of Table \ref{table:macroparameters}.

\begin{table}[ht]
	\centering
	\begin{tabular}{c c c}
		$\rho$ [$\mathrm{kg/m^3}$] & $\lambda$ [$\mathrm{Pa}$] & $\mu$ [$\mathrm{Pa}$] \\
		\hline
	$2000$ & $2\times 10^7$ & $0.7\times 10^7$
	\end{tabular}
	\caption{\small Lam\'e parameters of the ``softer'' Cauchy medium on the left side of the considered Cauchy/relaxed micromorphic interface.}
	\label{table:CauchyparametersSofter}
\end{table}
With these new parameters we can compute again, following Tables \ref{table:StoneleyT} and \ref{table:StoneleyR}, the critical angles for the appearance of Stoneley waves at low frequencies. We present these values in Table \ref{table:criticalangles2}:

\begin{table}[H]
	\centering
	\begin{tabular}{c c c c c}
		Incident wave & $\theta_{\text{crit}}^{L,r}$  & $\theta_{\text{crit}}^{L,t}$& $\theta_{\text{crit}}^{SV,t}$& $\theta_{\text{crit}}^{SH,t}$ \\
		\hline
		L& $-$ & $\frac{37 \pi}{100}$ & $\frac{49 \pi}{200}$ & $-$ \\
		\hline 
		SV & $\frac{7 \pi}{20}$ & $\frac{11\pi}{25}$ &$\frac{39 \pi}{100}$  & $-$ \\
		\hline 
		SH	& $-$ & $-$ & $-$ & $\frac{39 \pi}{100}$ \\
	\end{tabular}
	\caption{\small Critical angles governing the onset of Stoneley waves at a Cauchy/equivalent Cauchy interface between the two Cauchy media given in Tables \ref{table:CauchyparametersSofter} and \ref{table:macroparameters}, respectively. These values are computed according to the formulas given in Tables \ref{table:StoneleyT} and \ref{table:StoneleyR}. The superscripts $r$ and $t$ stand for ``reflected'' and ``transmitted'', respectively.} 
	\label{table:criticalangles2}
\end{table}

Figures \ref{fig:Tfreeg} and \ref{fig:Tfixedg} show the transmission coefficient for the softer Cauchy/relaxed micromorphic interface as a function of the angle of incidence and of frequency for both boundary conditions. The coloring of this plot is again such that the dark blue regions mean zero transmission, while the gradual change towards red is the increase in transmission (red being total transmission).

Table \ref{table:criticalangles2} shows that when the incident wave travels in a soft medium and hits the interface separating this medium from a stiffer one, many critical angles exist which determine the onset of Stoneley waves for all types of incident wave at low frequencies. Since many more Stoneley waves are created with respect to the previous case of section \ref{sec:ResultsStiffer}, we would expect less transmission in the low-frequency regime than before. This is indeed the case if we inspect Figures \ref{fig:Tfreeg} and \ref{fig:Tfixedg}: the presence of low-frequency Stoneley waves induces a wides zero-transmission zone in the low-frequency regime. We can also detect a certain role of boundary conditions in widening these zero-transmission zones when considering the fixed microstructure boundary condition (see Figure \ref{fig:Tfixedg}). 

Figures \ref{fig:k1LgF}, \ref{fig:k1SVgF} and \ref{fig:k1SHgF} once again show the imaginary part of the first component of the wave-vector $k_1$ for each mode of the relaxed micromorphic medium on the right. We see that Stoneley waves are observed almost everywhere both at low and high frequencies, with the exception of angles which are very close to normal incidence. Once again, the blue region denotes $\Im(k_1)=0$ (which implies that $k_1$ is real), while $\Im(k_1)$ is not vanishing in the red regions, which means that for each mode, the red color denotes that there are Stoneley waves associated to that mode. The first two modes in Figures \ref{fig:k1LgF} and \ref{fig:k1SVgF} correspond to L and SV Cauchy-like modes, while the first mode in Figure \ref{fig:k1SHgF} is the SH Cauchy-like mode in the low-frequency regime. Since we are considering a relaxed micromorphic medium, three additional modes with respect to the Cauchy case are present both for the in-plane (Figures \ref{fig:k1LgF} and \ref{fig:k1SVgF}) and for the out-of-plane problem. For this choice of parameters which make the left-side medium ``softer'' than the corresponding Cauchy medium on the right, we see that the Cauchy-like modes for all incident waves become Stoneley after a critical angle (clearly denoted on the plots with a vertical dashed line), something which is in accordance with Table \ref{table:criticalangles2}. Also Stoneley reflected waves can be observed in this case, but we do not present the plots of $\Im(k_1)$ for reflected waves to avoid overburdening.

We can conclude that it is possible to create an almost perfect total screen, which completely reflects incident waves for almost all angles of incidence and wave frequencies. This extraordinary possibility can be obtained by simply tailoring the properties of the left Cauchy medium which has to be chosen to be suitably softer than the right equivalent Cauchy medium. Regions of extraordinary transmission for very wade ranges of incident angles and wave frequencies can be engineered, opening the door to exciting applications.

\begin{figure}[H]
	\begin{subfigure}{.5\textwidth}
		\centering
		\includegraphics[scale=0.45]{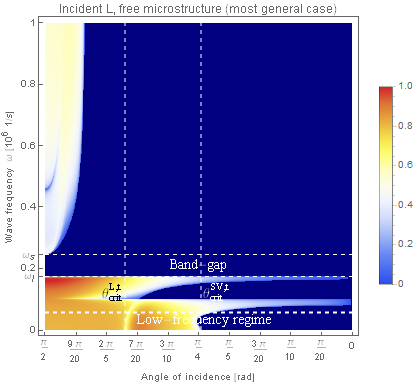}
		\caption{}
		\label{fig:FreeLg}
	\end{subfigure}%
	\begin{subfigure}{.5\textwidth}
		\centering
		\includegraphics[scale=0.45]{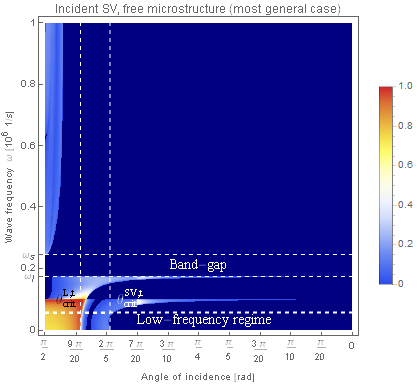}
		\caption{}
		\label{fig:FreeSVg}
	\end{subfigure}\\
	\centering
	\begin{subfigure}{.5\textwidth}
		\centering
		\includegraphics[scale=0.45]{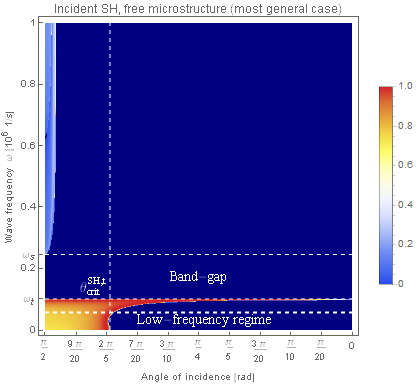}
		\caption{}
		\label{fig:FreeSHg}
	\end{subfigure}
	\caption{\small Transmission coefficients as a function of the angle of incidence $\theta_i$ and of the wave-frequency $\omega$ for L (a), SV (b) and SH (c) incident waves for the case of macro-clamp with free microstructure  and for a ``softer'' Cauchy medium on the left. The origin coincides with normal incidence ($\theta_i=\pi/2$), while the angle of incidence decreases towards the right until it reaches the value $\theta_i=0$, which corresponds to the limit case where the incidence is parallel to the interface. The band-gap region is highlighted by two dashed horizontal lines, where, as expected, we observe no transmission. The low-frequency regime is highlighted by the bottom horizontal dashed line, while the critical angles for the onset of Stoneley waves are denoted by vertical dashed lines. The dark blue zone shows that no transmission takes place, while the gradual change from dark blue to red shows the increase of transmission, red being total transmission.}
	\label{fig:Tfreeg}
\end{figure}

\begin{figure}[H]
	\begin{subfigure}{.5\textwidth}
		\centering
		\includegraphics[scale=0.45]{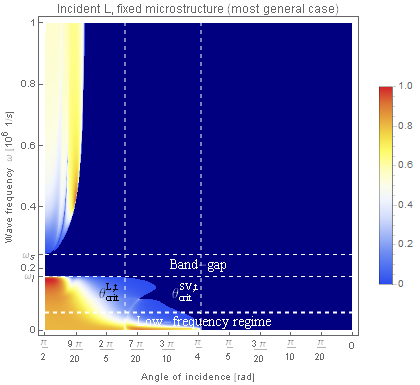}
		\caption{}
		\label{fig:FixedLg}
	\end{subfigure}%
	\begin{subfigure}{.5\textwidth}
		\centering
		\includegraphics[scale=0.45]{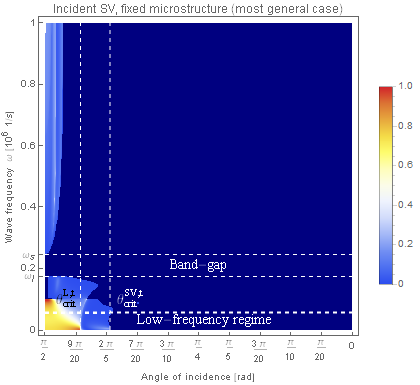}
		\caption{}
		\label{fig:FixedSVg}
	\end{subfigure}\\
	\centering
	\begin{subfigure}{.5\textwidth}
		\centering
		\includegraphics[scale=0.45]{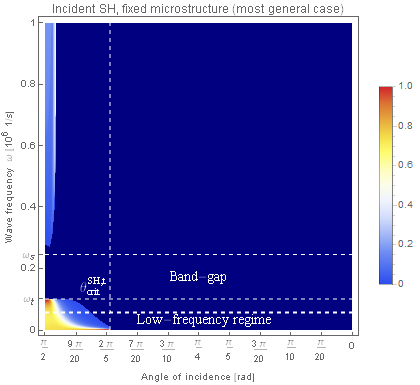}
		\caption{}
		\label{fig:FixedSHg}
	\end{subfigure}
	\caption{\small Transmission coefficients as a function of the angle of incidence $\theta_i$ and of the wave-frequency $\omega$  for L (a), SV (b) and SH (c) incident waves for the case of macro-clamp with fixed microstructure and for a ``softer'' Cauchy medium on the left. The origin coincides with normal incidence ($\theta_i=\pi/2$), while the angle of incidence decreases towards the right until it reaches the value $\theta_i=0$, which corresponds to the limit case where the incidence is parallel to the interface. The band-gap region is highlighted by two dashed horizontal lines, where, as expected, we observe no transmission. The low-frequency regime is highlighted by the bottom horizontal dashed line, while the critical angles for the onset of Stoneley waves are denoted by vertical dashed lines. The dark blue zone shows that no transmission takes place, while the gradual change from dark blue to red shows the increase of transmission, red being total transmission.}
	\label{fig:Tfixedg}
\end{figure}


\begin{figure}[H]
	\begin{subfigure}{.5\textwidth}
		\centering
		\includegraphics[scale=0.35]{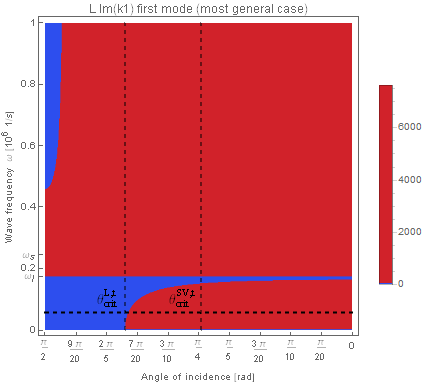}
		\caption{}
		\label{fig:k1L2gF}
	\end{subfigure}%
	\begin{subfigure}{.5\textwidth}
		\centering
		\includegraphics[scale=0.35]{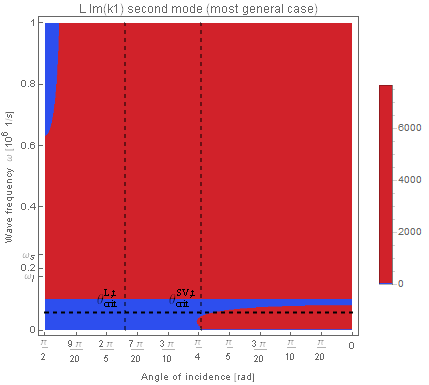}
		\caption{}
		\label{fig:k1L4gF}
	\end{subfigure}\\
	\begin{subfigure}{.5\textwidth}
		\centering
		\includegraphics[scale=0.35]{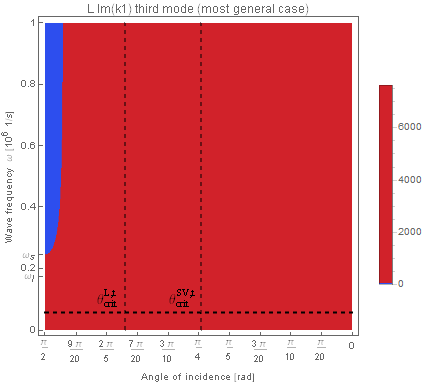}
		\caption{}
		\label{fig:k1L1g1F}
	\end{subfigure}%
	\begin{subfigure}{.5\textwidth}
		\centering
		\includegraphics[scale=0.35]{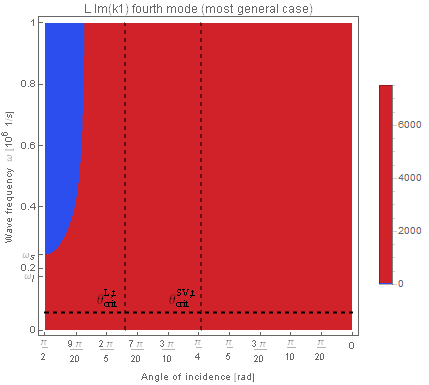}
		\caption{}
		\label{fig:k1L3gF}
	\end{subfigure}
	\centering
	\begin{subfigure}{.5\textwidth}
		\centering
		\includegraphics[scale=0.35]{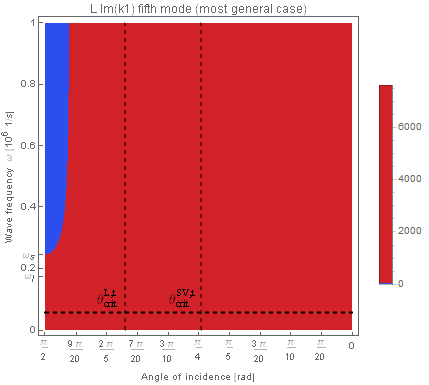}
		\caption{}
		\label{fig:k1L5gF}
	\end{subfigure}
	\caption{\small Values of $\Im(k_1)$ as a function of the angle of incidence $\theta_i$ and of the wave-frequency $\omega$ for the five modes of the relaxed micromorphic medium for the case of an incident L wave and a ``softer'' Cauchy medium on the left. The origin coincides with normal incidence ($\theta_i=\pi/2$), while the angle of incidence decreases towards the right until it reaches the value $\theta_i=0$, which corresponds to the limit case where the incidence isparallel to the interface. The first two modes (a) and (b) correspond to the L and SV modes for the equivalent Cauchy continuum at low frequencies. The red color in these plots means that the mode is Stoneley and does not propagate, while blue means that the mode is propagative.}
	\label{fig:k1LgF}
\end{figure}

\begin{figure}[H]
	\begin{subfigure}{.5\textwidth}
		\centering
		\includegraphics[scale=0.35]{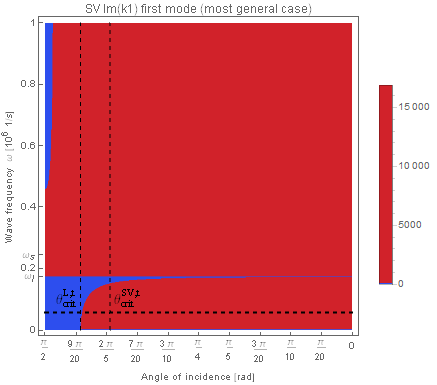}
		\caption{}
		\label{fig:k1SV2gF}
	\end{subfigure}%
	\begin{subfigure}{.5\textwidth}
		\centering
		\includegraphics[scale=0.35]{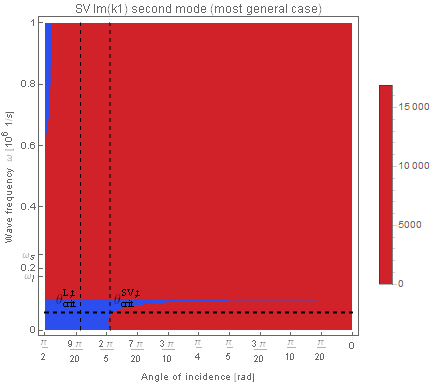}
		\caption{}
		\label{fig:k1SV4gF}
	\end{subfigure}\\
	\begin{subfigure}{.5\textwidth}
		\centering
		\includegraphics[scale=0.35]{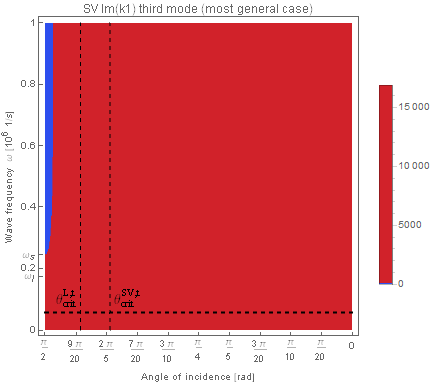}
		\caption{}
		\label{fig:k1SV1g1F}
	\end{subfigure}%
	\begin{subfigure}{.5\textwidth}
		\centering
		\includegraphics[scale=0.35]{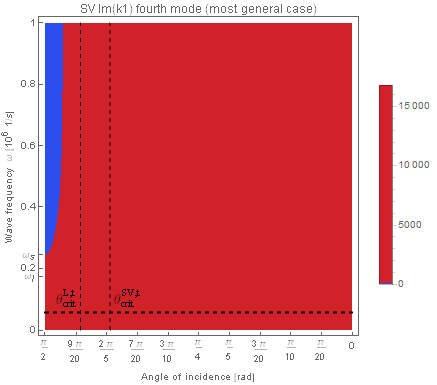}
		\caption{}
		\label{fig:k1SV3gF}
	\end{subfigure}
	\centering
	\begin{subfigure}{.5\textwidth}
		\centering
		\includegraphics[scale=0.35]{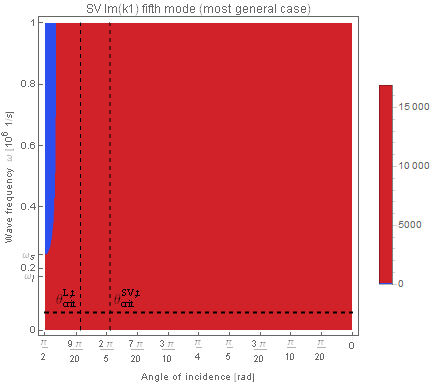}
		\caption{}
		\label{fig:k1SV5gF}
	\end{subfigure}
	\caption{\small Values of $\Im(k_1)$ as a function of the angle of incidence $\theta_i$ and of the wave-frequency $\omega$ for the five modes of the relaxed micromorphic medium for the case of an incident SV wave and a ``softer'' Cauchy medium on the left. The origin coincides with normal incidence ($\theta_i=\pi/2$), while the angle of incidence decreases towards the right until it reaches the value $\theta_i=0$, which corresponds to the limit case where the incidence is parallel to the interface. The first two modes (a) and (b) correspond to the L and SV modes for the equivalent Cauchy continuum at low frequencies. The red color in these plots means that the mode is Stoneley and does not propagate, while blue means that the mode is propagative.}
	\label{fig:k1SVgF}
\end{figure}

\begin{figure}[H]
	\centering
	\begin{subfigure}{.5\textwidth}
		\centering
		\includegraphics[scale=0.35]{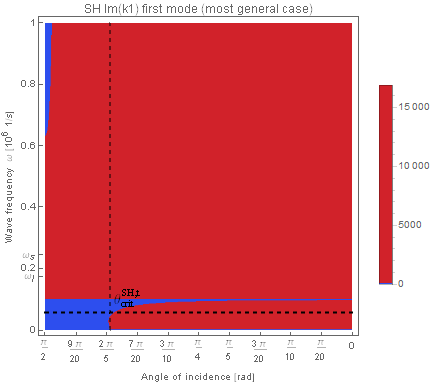}
		\caption{}
		\label{fig:k1SH3gF}
	\end{subfigure}%
	\begin{subfigure}{.5\textwidth}
		\centering
		\includegraphics[scale=0.35]{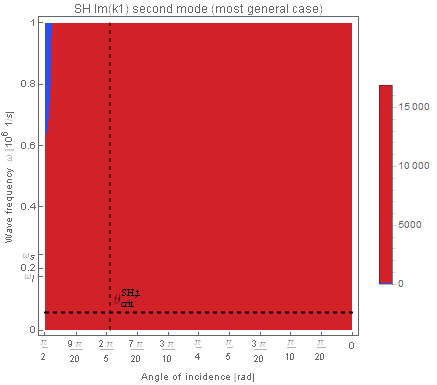}
		\caption{}
		\label{fig:k1SH1gF}
	\end{subfigure}\\
	\begin{subfigure}{.5\textwidth}
		\centering
		\includegraphics[scale=0.35]{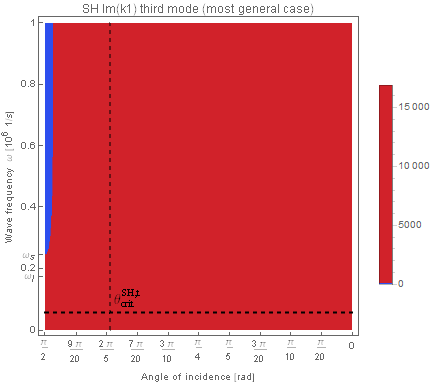}
		\caption{}
		\label{fig:k1SH2gF}
	\end{subfigure}%
	\begin{subfigure}{.5\textwidth}
		\centering
		\includegraphics[scale=0.35]{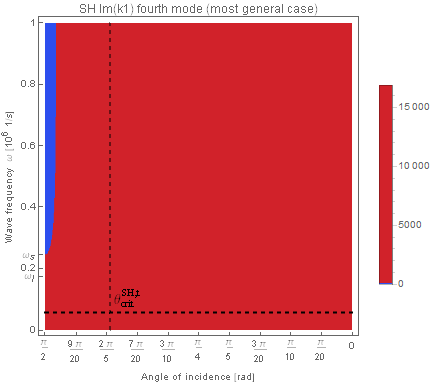}
		\caption{}
		\label{fig:k1SH4gF}
	\end{subfigure}
		\caption{\small Values of $\Im(k_1)$ as a function of the angle of incidence $\theta_i$ and of the wave-frequency $\omega$ for the four modes of the relaxed micromorphic medium for the case of an incident SH wave and a ``softer'' Cauchy medium on the left. The origin coincides with normal incidence ($\theta_i=\pi/2$), while the angle of incidence decreases towards the right until it reaches the value $\theta_i=0$, which corresponds to the limit case where the incidence is parallel to the interface. The first mode (a) corresponds to the SH mode for the equivalent Cauchy continuum at low frequencies. The red color in these plots means that the mode is Stoneley and does not propagate, while blue means that the mode is propagative.}
	\label{fig:k1SHgF}
\end{figure}

\section{Conclusions}
In this paper we present the detailed study of the reflective and refractive properties of a two-dimensional interface separating a classical Cauchy medium from a relaxed micromorphic medium. Both media are assumed to be semi-infinite.

We show in great detail that critical angles of incidence exist, beyond which classical Stoneley waves appear at low frequencies. It is shown that these critical angles directly depend on the relative mechanical properties of the two media. Moreover, we unveil the existence of critical angles which give rise to Stoneley waves at higher frequencies. These Stoneley waves are clearly related to the presence of an underlying microstructure in the metamaterial. 

We show that, due to the onset of low and high-frequency Stoneley waves, wide frequency bounds where total reflection and/or total transmission occur can be engineered. This total reflection/transmission phenomenon is appealing for applications, in which total screens for elastic waves, such as cloaks or wave-filters, are desirable. It is clear that the ability of widening the frequency and incident angle intervals for which total reflection/transmission occur, would be of paramount importance for conceiving new devices which are more and more performant for wavefront manipulation.

We also clearly show that the simple fact of suitably tailoring the relative stiffnesses of the two media allows for the possibility of conceiving almost perfect total screens which do not transmit elastic waves for any kind of incident wave (longitudinal, in-plane and out-of-plane shear) and for almost all (low and high) frequencies and angles of incidence. Acting on such relative stiffnesses allows to achieve also the opposite situation, where total transmission occurs for large frequency bounds before a microstructure-related critical angle. This could be exploited for the conception of selective cloaks which make objects transparent to waves dependently on the angle of incidence. 

Further work will be devoted to extending the present results to the case of anisotropic media and to considering refection/transmission properties at interfaces between finite media. 

\newpage
%

{\footnotesize{}\bibliographystyle{plain}
	\bibliography{library}

\begin{thebibliography}{10}

\bibitem{achenbach1973wave}
Jan~D. Achenbach.
\newblock {\em {Wave Propagation in Elastic Solids}}.
\newblock North-Holland Publishing Company, Amsterdam, The Netherlands, 1973.

\bibitem{auld1973acoustic2}
Bertram~A. Auld.
\newblock {\em {Acoustic Fields and Waves in Solids, Vol. II}}.
\newblock Wiley-Interscience Publication, 1973.

\bibitem{barbagallo2017transparent}
Gabriele Barbagallo, Angela Madeo, Marco~Valerio {d}'Agostino, Rafael Abreu,
  Ionel-Dumitrel Ghiba, and Patrizio Neff.
\newblock {Transparent anisotropy for the relaxed micromorphic model:
  macroscopic consistency conditions and long wave length asymptotics}.
\newblock {\em International Journal of Solids and Structures}, 120:7--30,
  2017.

\bibitem{bueckmann2015mechanical}
Tiemo B\"uckmann, Muamer Kadic, Robert Schittny, and Martin Wegener.
\newblock {Mechanical cloak design by direct lattice transformation}.
\newblock {\em Proceedings of the National Academy of Sciences},
  112(6):4930--4934, 2015.

\bibitem{chen2007acoustic}
Huanyang Chen and C.~T. Chan.
\newblock {Acoustic cloaking in three dimensions using acoustic metamaterials
  }.
\newblock {\em Applied Physics Letters}, 91, 2007.

\bibitem{craster2010high}
Richard~V. Craster, Julius Kaplunov, and Aleksey~V. Pichugin.
\newblock {High-frequency homogenization for periodic media}.
\newblock {\em Proceedings of the Royal Society A: Mathematical, Physical and
  Engineering Sciences}, 466(2120):2341--2362, 2010.

\bibitem{cummer2016controlling}
Steven~A. Cummer, Johan Christensen, and Andrea Al\'u.
\newblock {Controlling sound with acoustic metamaterials}.
\newblock {\em Nature Reviews Materials}, 113(3), 2014.

\bibitem{dagostino2018effective}
Marco~Valerio {d}'Agostino, Gabriele Barbagallo, Ionel-Dumitrel Ghiba, Bernhard
  Eidel, Patrizio Neff, and Angela Madeo.
\newblock {Effective description of anisotropic wave dispersion in mechanical
  band-gap metamaterials via the relaxed micromorphic model}.
\newblock {\em submitted, arXiv preprint}, 1709.07054, 2018.

\bibitem{dagostino2017panorama}
Marco~Valerio {d}'Agostino, Gabriele Barbagallo, Ionel-Dumitrel Ghiba, Angela
  Madeo, and Patrizio Neff.
\newblock {A panorama of dispersion curves for the weighted isotropic relaxed
  micromorphic model}.
\newblock {\em Zeitschrift f{\"{u}}r Angewandte Mathematik und Mechanik}, pages
  1--46, 2017.

\bibitem{ghiba2014relaxed}
Ionel-Dumitrel Ghiba, Patrizio Neff, Angela Madeo, Luca Placidi, and Giuseppe
  Rosi.
\newblock {The relaxed linear micromorphic continuum: Existence, uniqueness and
  continuous dependence in dynamics}.
\newblock {\em Mathematics and Mechanics of Solids}, 20(10):1171--1197, 2015.

\bibitem{graff1975wave}
Karl~F. Graff.
\newblock {\em {Wave Motion in Elastic Solids}}.
\newblock Dover Publications, Inc., New York, 1975.

\bibitem{guenneau2007acoustic}
S\'ebastien Guenneau, Alexander Movchan, Gunnar P\'etursson, and S.~Anantha
  Ramakrishna.
\newblock {Acoustic metamaterials for sound focusing and confinement}.
\newblock {\em New Journal of Physics}, 9(399), 2007.

\bibitem{kaina2014slow}
Nad\'ege Kaina, Alexandre Causier, Yoan Bourlier, Mathias Fink, Thomas
  Berthelot, and Geoffroy Lerosey.
\newblock {Slow waves in locally resonant metamaterials line defect
  waveguides}.
\newblock {\em Scientific Reports}, 7(1), 2014.

\bibitem{li2014acoustic}
Yong Li and Badreddine~M. Assouar.
\newblock {Acoustic metasurface-based perfect absorber with deep subwavelength
  thickness }.
\newblock {\em Applied Physics Letters}, 108(6), 2016.

\bibitem{liang2018wavefront}
Ben Lieang, Jian-chun Chang, and Cheng-Wei Qiu.
\newblock {Wavefront manipulation by acoustic metasurfaces: from physics and
  applications}.
\newblock {\em Nanophotonics}, 7(6):1191--1205, 2018.

\bibitem{lions1978asymptotic}
J.-L. Lions and G.~Papanicolaou.
\newblock {\em {Asymptotic Analysis for Periodic Structures}}.
\newblock AMS Chelsea Publishing, 1978.

\bibitem{liu2000locally}
Zhengyou Liu, Xixiang Zhang, Yiwei Mao, Yirong Zhu, Zhiyu Yang, Che~Ting Chan,
  and Ping Sheng.
\newblock {Locally resonant sonic materials}.
\newblock {\em Science}, 289(5485):1734--1736, 2000.

\bibitem{madeo2017relaxed}
Angela Madeo, Gabriele Barbagallo, Manuel Collet, Marco~Valerio {d}'Agostino,
  Marco Miniaci, and Patrizio Neff.
\newblock {Relaxed micromorphic modeling of the interface between a homogeneous
  solid and a band-gap metamaterial: New perspectives towards metastructural
  design}.
\newblock {\em Mathematics and Mechanics of Solids}, pages 1--22, 2017.

\bibitem{madeo2016first}
Angela Madeo, Gabriele Barbagallo, Marco~Valerio {d}'Agostino, Luca Placidi,
  and Patrizio Neff.
\newblock {First evidence of non-locality in real band-gap metamaterials:
  determining parameters in the relaxed micromorphic model}.
\newblock {\em Proceedings of the Royal Society A: Mathematical, Physical and
  Engineering Sciences}, 472(2190):20160169, 2016.

\bibitem{madeo2017modeling}
Angela Madeo, Manuel Collet, Marco Miniaci, K{\'{e}}vin Billon, Morvan Ouisse,
  and Patrizio Neff.
\newblock {Modeling phononic crystals via the weighted relaxed micromorphic
  model with free and gradient micro-inertia}.
\newblock {\em Journal of Elasticity}, 130(1):59--83, 2017.

\bibitem{madeo2017role}
Angela Madeo, Patrizio Neff, Elias~C. Aifantis, Gabriele Barbagallo, and
  Marco~Valerio {d}'Agostino.
\newblock {On the role of micro-inertia in enriched continuum mechanics}.
\newblock {\em Proceedings of the Royal Society A: Mathematical, Physical and
  Engineering Science}, 473(2198):20160722, 2017.

\bibitem{madeo2017review}
Angela Madeo, Patrizio Neff, Gabriele Barbagallo, Marco~Valerio {d}'Agostino,
  and Ionel-Dumitrel Ghiba.
\newblock {A review on wave propagation modeling in band-gap metamaterials via
  enriched continuum models}.
\newblock In Francesco {d}ell'Isola, Mircea Sofonea, and David~J. Steigmann,
  editors, {\em Mathematical Modelling in Solid Mechanics}, Advanced Structured
  Materials, pages 89--105. Springer, Singapore, 2017.

\bibitem{madeo2016complete}
Angela Madeo, Patrizio Neff, Marco~Valerio {d}'Agostino, and Gabriele
  Barbagallo.
\newblock {Complete band gaps including non-local effects occur only in the
  relaxed micromorphic model}.
\newblock {\em Comptes Rendus M{\'{e}}canique}, 344(11-12):784--796, 2016.

\bibitem{madeo2014band}
Angela Madeo, Patrizio Neff, Ionel-Dumitrel Ghiba, Luca Placidi, and Giuseppe
  Rosi.
\newblock {Band gaps in the relaxed linear micromorphic continuum}.
\newblock {\em Zeitschrift f{\"{u}}r Angewandte Mathematik und Mechanik},
  95(9):880--887, 2014.

\bibitem{madeo2015wave}
Angela Madeo, Patrizio Neff, Ionel-Dumitrel Ghiba, Luca Placidi, and Giuseppe
  Rosi.
\newblock {Wave propagation in relaxed micromorphic continua: modeling
  metamaterials with frequency band-gaps}.
\newblock {\em Continuum Mechanics and Thermodynamics}, 27(4-5):551--570, 2015.

\bibitem{madeo2016reflection}
Angela Madeo, Patrizio Neff, Ionel-Dumitrel Ghiba, and Giuseppe Rosi.
\newblock {Reflection and transmission of elastic waves in non-local band-gap
  metamaterials: a comprehensive study via the relaxed micromorphic model}.
\newblock {\em Journal of the Mechanics and Physics of Solids}, 95:441--479,
  2016.

\bibitem{misseroni2016cymatics}
Diego Misseroni, Daniel~J. Colquitt, Alexander~B. Movchan, Natasha~V. Movchan,
  and Ian~Samuel Jones.
\newblock {Cymatics for the cloaking of flexural vibrations in a structured
  plate}.
\newblock {\em Scientific Reports}, 6(23929), 2016.

\bibitem{neff2015relaxed}
Patrizio Neff, Ionel-Dumitrel Ghiba, Markus Lazar, and Angela Madeo.
\newblock {The relaxed linear micromorphic continuum: well-posedness of the
  static problem and relations to the gauge theory of dislocations}.
\newblock {\em The Quarterly Journal of Mechanics and Applied Mathematics},
  68(1):53--84, 2015.

\bibitem{neff2014unifying}
Patrizio Neff, Ionel-Dumitrel Ghiba, Angela Madeo, Luca Placidi, and Giuseppe
  Rosi.
\newblock {A unifying perspective: the relaxed linear micromorphic continuum}.
\newblock {\em Continuum Mechanics and Thermodynamics}, 26(5):639--681, 2014.

\bibitem{neff2017real}
Patrizio Neff, Angela Madeo, Gabriele Barbagallo, Marco~Valerio {d}'Agostino,
  Rafael Abreu, and Ionel-Dumitrel Ghiba.
\newblock {Real wave propagation in the isotropic-relaxed micromorphic model}.
\newblock {\em Proceedings of the Royal Society A: Mathematical, Physical and
  Engineering Sciences}, 473(2197):20160790, 2017.

\bibitem{rayleigh1885onwaves}
Lord John William~Strutt Rayleigh.
\newblock {On waves propagated along the plane surface of an elastic solid}.
\newblock {\em Proceedings of the London Mathematical Society}, 17:4--11, 1885.

\bibitem{stoneley1924elastic}
Robert Stoneley.
\newblock {Elastic waves at the surface of separation of two solids}.
\newblock {\em Proceedings of the Royal Society A: Mathematical, Physical and
  Engineering Sciences}, 106(738):416--428, 1924.

\bibitem{tallarico2017tilted}
Domenico Tallarico, Natalia~V. Movchan, Alexander~B. Movchan, and Daniel~J.
  Colquitt.
\newblock {Tilted resonators in a triangular elastic lattice: Chirality, Bloch
  waves and negative refraction}.
\newblock {\em Journal of the Mechanics and Physics of Solids}, 103:236--256,
  2017.

\bibitem{tallarico2017edge}
Domenico Tallarico, Alessio Trevisan, Natalia~V. Movchan, and Alexander~B.
  Movchan.
\newblock {Edge waves and localization in lattices containing tilted
  resonators}.
\newblock {\em Frontiers in Materials}, 4(June):1--13, 2017.

\bibitem{valentine2009optical}
Jason Valentine, Jensen Li, Thomas Zentgraf, Guy Bartal, and Xiang Zhang.
\newblock {An optical cloak made of dielectrics}.
\newblock {\em Nature Materials}, 8(7):568--571, 2009.

\bibitem{wang2014harnessing}
Pai Wang, Filippo Casadei, Sicong Shan, James~C. Weaver, and Katia Bertoldi.
\newblock {Harnessing buckling to design tunable locally resonant acoustic
  metamaterials}.
\newblock {\em Physical Review Letters}, 113, 2014.

\bibitem{xiao2011longitudinal}
Yong Xiao, Jihong Wen, and Xisen Wen.
\newblock {Longitudinal wave band gaps in metamaterial-based elastic rods
  containing multi-degree-of-freedom resonators}.
\newblock {\em New Journal of Physics}, 14, 2011.

\bibitem{xie2014wavefront}
Yangbo Xie, Wenqi Wang, Huanyang Chen, Adam Konneker, Bogdan-Ioan Popa, and
  Steven~A. Cummer.
\newblock {Wavefront modulation and subwavelength diffractive acoustics with an
  acoustic metasurface}.
\newblock {\em Nature Communications}, 5(1):1191--1205, 2014.

\bibitem{zhu2015study}
R.~Zhu, X.~N. Liu, and G.~L. Huang.
\newblock {Study of anomalous wave propagation and reflection in semi-infinite
  elastic metamaterials}.
\newblock {\em Wave Motion}, 55:73--83, 2015.

\end{thebibliography}
}

\newpage
\appendix
\section{Appendix for the classical Cauchy model}\label{appendixCauchy}

\subsection{Calculation of the determinant of $A$}\label{appendixCauchy1}
\small
We demonstrate the explicit calculation of the determinant of the matrix 
\begin{equation}\label{eq:ACauchymatrix}
A=\left(\begin{array}{cc}
\omega^2 - c_l^2 k_1^2 - c_s^2 k_2^2 & -c_V^2 k_1 k_2\\
-c_V^2 k_1 k_2 & \omega^2 - c_l^2 k_2^2 - c_s^2 k_1^2
\end{array}\right).
\end{equation}
We have
\begin{align}
\det A &= (\omega^2 - c_l^2 k_1^2 - c_s^2 k_2^2)(\omega^2 - c_l^2 k_2^2 - c_s^2 k_1^2) - c_V^4 k_1^2 k_2^2 \nonumber\\
&=\omega^4 - \omega^2 c_l^2 (k_1^2+k_2^2) - \omega^2 c_s^2 (k_1^2 + k_2^2) + c_l^2 c_s^2 (k_1^4 + k_2^4) + (c_l^4+c_s^4)k_1^2k_2^2 - c_V^4 k_1^2 k_2^2\\
&=\omega^4 - \omega^2 (k_1^2 +k_2^2)(c_l^2 +c_s^2)+ c_l^2 c_s^2 (k_1^4 + k_2^4) + (c_l^4+c_s^4) k_1^2 k_2^2 - c_V^2 k_1^2 k2^2 \nonumber\\
&= \omega^4- \omega^2(k_1^2 + k_2^2)\frac{2 \mu + \lambda + \mu}{\rho} + \frac{(2\mu + \lambda)\mu}{\rho^2}(k_1^4+k_2^4) + \frac{(2\mu+\lambda)^2+\mu^2}{\rho^2}k_1^2k_2^2 - \frac{(\mu+\lambda)^2}{\rho^2}k_1^2k_2^2 \nonumber \\
&=\frac{1}{\rho^2}\left[\rho^2\omega^4 - \rho \omega^2 \left((2\mu+\lambda)k_1^2 + (2\mu + \lambda) k_2^2 + \mu k_1^2 + \mu k_2^2\right) + \mu (2\mu + \lambda) k_1^4 + \mu (2\mu + \lambda)k_2^4 +2 \mu (2\mu +\lambda) k_1^2 k_2^2\right] \nonumber \\
&=\frac{1}{\rho^2}\left[\rho^2 \omega^4 -\rho \omega^2 (2\mu+\lambda)(k_1^2+k_2^2) -  \rho \omega^2 \mu (k_1^2+k_2^2) + (2\mu + \lambda)(k_1^2+k_2^2) \mu (k_1^2+k_2^2)    \right] \nonumber\\
&= \frac{1}{\rho^2}\left[(\mu(k_1^2+k_2^2)-\rho \omega^2)(2\mu + \lambda)(k_1^2+k_2^2) - \rho \omega^2 (2\mu + \lambda)(k_1^2+k_2^2)+\rho^2 \omega^4\right] \nonumber\\
&= \frac{1}{\rho^2}\left((2\mu+\lambda)(k_1^2+k_2^2)-\rho \omega^2\right)\left(\mu ( k_1^2 + k_2^2)-\rho \omega^2 \right). \label{eq:ADeterminant}
\end{align}

\subsection{Lemma \ref{Lemma1}}\label{appendixCauchy2}
\small
We have the following well-known result.
\begin{lemma}\label{Lemma1}
Let 
\[
u_1(x,t) = A(x)e^{i(\omega t-kx)}, \quad u_2(x,t) = B(x)e^{i(\omega t-kx)}
\]
be two functions with $A,B: \R^3 \to \C$. Then the following holds
\begin{equation}\label{eq:Lemma1}
\frac{1}{T} \int_0^{T} \Re\{u_1(x,t)\}  \Re\{u_2(x,t)\} dt= \frac{1}{2T} \Re (A B^{*}),
\end{equation}
where $T$ is the period of the functions $u_1, u_2$ and $B^{*}$ denotes the complex conjugate.
\end{lemma}
\begin{proof}
We have:
\begin{align}
\frac{1}{T}\int_0^{T} \Re \{u_1(x,t)\} \Re \{u_2(x,t)\}dt &= \frac{1}{T}\int_0^{T} \Re \left(A e^{i(\omega t-kx)}\right) \Re \left(B e^{i(\omega t -kx)}\right)dt \nonumber\\
&=\frac{1}{T}\int_0^{T} \frac{Ae^{i(\omega t-kx)} +A^{*}e^{-i(\omega t-kx)}}{2}\frac{Be^{i(\omega t-kx)} +B^{*}e^{-i(\omega t -kx)}}{2}dt \nonumber\\
&=\frac{1}{T}\int_0^{T} \frac{AB}{4}\underbrace{e^{2i \omega t}}_{\text{periodic}}e^{2ikx}+\frac{AB^{*}}{4}+\frac{A^{*}B}{4}+\frac{A^{*}B^{*}}{4}\underbrace{e^{2i \omega t}}_{\text{periodic}}e^{2ikx}dt \nonumber\\
&=\frac{1}{T}\int_0^{T} \frac{AB^{*}+A^{*}B}{4}dt =\frac{1}{2T}\int_0^{T} \Re\left(AB^{*}\right)dt  =\frac{1}{2}\Re\left(AB^{*}\right),\label{eq:ALemmaProof}
\end{align}
where we used the facts the periodic function $e^{2i\omega t}$ integrated over its period is zero and that for any complex number $z\in \C$: $\Re(z) = \frac{z+z^{*}}{2}$.
\end{proof}
\subsection{Conditions for the appearance of Stoneley waves}\label{appendixCauchy3}
We explicitly demonstrate all calculations carried out in order to produce Tables \ref{table:StoneleyT} and \ref{table:StoneleyR}.
\subsubsection*{Incident L, transmitted waves}
In this case, $k_2 = -\kabs \cos \theta = -\frac{\omega}{c_L^-} \cos \theta$. 
\begin{itemize}
\item L-mode: \\
\begin{equation}\label{eq:AStoneleyLLt}
\left(\frac{\omega}{c_L^+}\right)^2-k_2^2 < 0 \Rightarrow \left(\frac{\omega}{c_L^+}\right)^2 <k_2^2 \Rightarrow \left(\frac{\omega}{c_L^+}\right)^2 <\left(\frac{\omega}{c_L^-}\right)^2 \cos^2 \theta \Rightarrow \cos^2 \theta > \left(\frac{c_L^-}{c_L^+}\right)^2=\frac{\rho^+(2\mu^- + \lambda^-)}{\rho^-(2 \mu^+ + \lambda^+)}.
\end{equation}
\item SV-mode:
\begin{equation}\label{eq:AStoneleyLSVt}
\left(\frac{\omega}{c_S^+}\right)^2-k_2^2 < 0 \Rightarrow \left(\frac{\omega}{c_S^+}\right)^2 <k_2^2 \Rightarrow \left(\frac{\omega}{c_S^+}\right)^2 <\left(\frac{\omega}{c_L^-}\right)^2 \cos^2 \theta \Rightarrow \cos^2 \theta > \left(\frac{c_L^-}{c_S^+}\right)^2=\frac{\rho^+(2\mu^- + \lambda^-)}{\rho^- \mu^+ }.
\end{equation}
\end{itemize}
\subsubsection*{Incident SV, transmitted waves}
In this case, $k_2 = -\kabs \cos \theta = -\frac{\omega}{c_S^-} \cos \theta$. 
\begin{itemize}
	\item L-mode: \\
	\begin{equation}\label{eq:StoneleySVLt}
	\left(\frac{\omega}{c_L^+}\right)^2-k_2^2 < 0 \Rightarrow \left(\frac{\omega}{c_L^+}\right)^2 <k_2^2 \Rightarrow \left(\frac{\omega}{c_L^+}\right)^2 <\left(\frac{\omega}{c_S^-}\right)^2 \cos^2 \theta \Rightarrow \cos^2 \theta > \left(\frac{c_S^-}{c_L^+}\right)^2=\frac{\rho^+\mu^- }{\rho^-(2 \mu^+ + \lambda^+)}.
	\end{equation}
	\item SV-mode:
\begin{equation}\label{eq:AStoneleySVSVt}
\left(\frac{\omega}{c_S^+}\right)^2-k_2^2 < 0 \Rightarrow \left(\frac{\omega}{c_S^+}\right)^2 <k_2^2 \Rightarrow \left(\frac{\omega}{c_S^+}\right)^2 <\left(\frac{\omega}{c_S^-}\right)^2 \cos^2 \theta \Rightarrow \cos^2 \theta > \left(\frac{c_S^-}{c_S^+}\right)^2=\frac{\rho^+ \mu^- }{\rho^- \mu^+ }.
\end{equation}
\end{itemize}
\subsubsection*{Incident SH, transmitted waves}
In this case, $k_2 = -\kabs \cos \theta = -\frac{\omega}{c_S^-} \cos \theta$. 
\begin{itemize}
\item SH-mode:
\begin{equation}\label{eq:AStoneleySHt}
\left(\frac{\omega}{c_S^+}\right)^2-k_2^2 < 0 \Rightarrow \left(\frac{\omega}{c_S^+}\right)^2 <k_2^2 \Rightarrow \left(\frac{\omega}{c_S^+}\right)^2 <\left(\frac{\omega}{c_S^-}\right)^2 \cos^2 \theta \Rightarrow \cos^2 \theta > \left(\frac{c_S^-}{c_S^+}\right)^2=\frac{\rho^+ \mu^- }{\rho^- \mu^+ }.
\end{equation}
\end{itemize}

\subsubsection*{Incident Longitudinal, reflected waves}
In this case, $k_2 = -\kabs \cos \theta = -\frac{\omega}{c_L^-} \cos \theta$. 
\begin{itemize}
	\item L-mode: \\
	\begin{equation}\label{eq:StoneleyLLr}
	\left(\frac{\omega}{c_L^-}\right)^2-k_2^2 < 0 \Rightarrow \left(\frac{\omega}{c_L^-}\right)^2 <k_2^2 \Rightarrow \left(\frac{\omega}{c_L^-}\right)^2 <\left(\frac{\omega}{c_L^-}\right)^2 \cos^2 \theta \Rightarrow \cos^2 \theta > \left(\frac{c_L^-}{c_L^-}\right)^2=1,
	\end{equation}
which renders the L mode becoming Stoneley in the case of an incident Longitudinal wave, impossible.
	\item SV-mode:
\begin{equation}\label{eq:AStoneleyLSVr}
\left(\frac{\omega}{c_S^-}\right)^2-k_2^2 < 0 \Rightarrow \left(\frac{\omega}{c_S^-}\right)^2 <k_2^2 \Rightarrow \left(\frac{\omega}{c_S^-}\right)^2 <\left(\frac{\omega}{c_L^-}\right)^2 \cos^2 \theta \Rightarrow \cos^2 \theta > \left(\frac{c_L^-}{c_S^-}\right)^2>1,
\end{equation}
\end{itemize}
which renders the SV mode becoming Stoneley in the case of an incident Longitudinal wave, impossible, since $c_L>c_S$ by definition.
\subsubsection*{Incident SV, reflected waves}
In this case, $k_2 = -\kabs \cos \theta = -\frac{\omega}{c_S^-} \cos \theta$. 
\begin{itemize}
	\item L-mode: 
	\begin{equation}\label{eq:AStoneleySVLr}
	\left(\frac{\omega}{c_L^-}\right)^2-k_2^2 < 0 \Rightarrow \left(\frac{\omega}{c_L^-}\right)^2 <k_2^2 \Rightarrow \left(\frac{\omega}{c_L^-}\right)^2 <\left(\frac{\omega}{c_S^-}\right)^2 \cos^2 \theta \Rightarrow \cos^2 \theta > \left(\frac{c_S^-}{c_L^-}\right)^2=\frac{\mu^- }{(2 \mu^- + \lambda^-)}.
	\end{equation}
	\item SV-mode:
\begin{equation}\label{eq:AStoneleySVSVr}
\left(\frac{\omega}{c_S^-}\right)^2-k_2^2 < 0 \Rightarrow \left(\frac{\omega}{c_S^-}\right)^2 <k_2^2 \Rightarrow \left(\frac{\omega}{c_S^-}\right)^2 <\left(\frac{\omega}{c_S^-}\right)^2 \cos^2 \theta \Rightarrow \cos^2 \theta > \left(\frac{c_S^-}{c_S^-}\right)^2=1,
\end{equation}
which renders the SV mode becoming Stoneley in the case of an incident SV wave, impossible.\end{itemize}
\subsubsection*{Incident SH, reflected waves}
In this case, $k_2 = -\kabs \cos \theta = -\frac{\omega}{c_S^-} \cos \theta$. 

\begin{itemize}
	\item SH-mode:
\begin{equation}\label{eq:AStoneleySHr}
\left(\frac{\omega}{c_S^-}\right)^2-k_2^2 < 0 \Rightarrow \left(\frac{\omega}{c_S^-}\right)^2 <k_2^2 \Rightarrow \left(\frac{\omega}{c_S^-}\right)^2 <\left(\frac{\omega}{c_S^-}\right)^2 \cos^2 \theta \Rightarrow \cos^2 \theta > \left(\frac{c_S^-}{c_S^-}\right)^2=1,
\end{equation}
which shows that there can be no reflected Stoneley waves in the case of an incident SH wave.
\end{itemize}

\section{Appendix for the relaxed micromorphic model}\label{appendixRelaxed}
\subsection{Governing equations in component-wise notation}\label{appendixRelaxed1}
\small
\begin{align*}
u_{1,tt} &= \frac{\lame+2\mue}{\rho}u_{1,11} + \frac{\mue + \muc}{\rho} u_{1,22} +\frac{\mue - \muc + \lame}{\rho} u_{2,12} - \frac{2\mue}{\rho}P_{11,1} - \frac{\lame}{\rho}(P_{11,1} + P_{22,1} + P_{33,1}) \\ &- \frac{(\mue + \muc)}{\rho} P_{12,2} - \frac{(\mue - \muc)}{\rho}P_{21,2},\\
u_{2,tt} &= \frac{(\lame + 2\mue)}{\rho} u_{2,22} + \frac{\muc + \mue}{\rho} u_{2,11} + \frac{\mue - \muc + \lame}{\rho} u_{1,12} - \frac{2\mue}{\rho} P_{22,2} - \frac{\lame}{\rho}(P_{11,2} + P_{22,2} + P_{33,2}) \\
&- \frac{\mue + \muc}{\rho} P_{21,1} - \frac{\mue - \muc}{\rho} P_{12,1},\\
u_{3,tt} &= \frac{\mue + \muc}{\rho} (u_{3,11} + u_{3,22}) - \frac{\mue - \muc}{\rho} (P_{13,1}
 + P_{23,2}) - \frac{\mue + \muc}{\rho} (P_{31,1} + P_{32,2}),\\
P_{11,tt} &= \frac{2\mue + \lame}{\eta} u_{1,1} + \frac{\lame}{\eta} u_{2,2} - 2\frac{\mue + \mumic}{\eta}P_{11} - \frac{\lame + \lammic}{\eta}(P_{11} + P_{22} + P_{33}) + \frac{\mue \Lc}{\eta} (P_{11,22} - P_{12,12}),\\
P_{12,tt} &= \frac{\mue + \muc}{\eta} u_{1,2} + \frac{\mue - \muc}{\eta} u_{2,1} - \frac{\mue+ \muc + \mumic}{\eta} P_{12} - \frac{\mue - \muc + \mumic}{\eta} P_{21} 
+ \frac{\mue \Lc}{\eta}(P_{12,11} - P_{11,12}), \\
P_{13,tt} &= \frac{\mue - \muc}{\eta} u_{3,1} - \frac{\mue + \muc + \mumic}{\eta} P_{13} - \frac{\mue - \muc + \mumic}{\eta} P_{31} + \frac{\mue \Lc}{\eta} (P_{13,22} + P_{13,11}), \\
P_{21,tt} &= \frac{\mue - \muc}{\eta} u_{1,2} + \frac{\mue +\muc}{\eta} u_{2,1} - \frac{\mue - \muc + \mumic}{\eta} P_{12} - \frac{\mue + \muc + \mumic}{\eta} P_{21}+ \frac{\mue \Lc}{\eta} (P_{21,22} - P_{22,12}), \\
P_{22,tt} &= \frac{2\mue + \lame}{\eta} u_{2,2} + \frac{\lame}{\eta} u_{1,1} - 2 \frac{\mue + \mumic}{\eta} P_{22} - \frac{\lame + \lammic}{\eta} (P_{11} + P_{22} + P_{33})+ \frac{\mue \Lc}{\eta} (P_{22,11} - P_{21,12}), \\
P_{23,tt} &= \frac{\mue - \muc}{\eta} u_{3,2} - \frac{\mue + \muc + \mumic}{\eta} P_{23} - \frac{\mue - \muc + \mumic}{\eta} P_{32} + \frac{\mue \Lc}{\eta}(P_{23,22} + P_{23,11}), \\
P_{31,tt} &= \frac{\mue + \muc}{\eta} u_{3,1} - \frac{\mue - \muc + \mumic}{\eta} P_{13} - \frac{\muc + \mue + \mumic}{\eta} P_{31} + \frac{\mue \Lc}{\eta} (P_{31,22} - P_{32,12}), \\
P_{32,tt} &= \frac{\mue + \muc}{\eta}u_{3,2} - \frac{\mue - \muc + \mumic}{\eta} P_{23} -\frac{\mue + \muc + \mumic}{\eta} P_{32} + \frac{\mue \Lc}{\eta} (P_{32,11} - P_{31,12}), \\
P_{33,tt} &= \frac{\lame}{\eta} (u_{1,1} + u_{2,2}) - 2 \frac{\mue + \mumic}{\eta} P_{33} - \frac{\lame + \lammic}{\eta} (P_{11} + P_{22} + P_{33}) + \frac{\mue \Lc}{\eta}(P_{33,22} + P_{31 3,11}).
\end{align*}

\subsection{Governing equations with new variables}\label{appendixRelaxed2}
\small
Define the new variables\footnote{The definitions are motivated by the Cartan-Lie decomposition of the tensor $P$.}
\begin{align*}
&P^S = \frac{1}{3}(P_{11} + P_{22} + P_{33}), \quad P^{D}_1 = P_{11} - P^S, \quad P^D_2 = P_{22} - P^S, \quad P_{(1\gamma)} = \frac{1}{2}(P_{1\gamma} + P_{\gamma 1}), \\
&P_{[1\gamma]} = \frac{1}{2}(P_{1\gamma} - P_{\gamma 1}), \quad P_{(23)} = \frac{1}{2}(P_{23} + P_{32}), \quad P_{[23]} = \frac{1}{2}(P_{12} - P_{21}),
\end{align*}
with $\gamma=2,3$ and rewrite the equations with respect to these new variables: 
\begin{align*}
u_{1,tt} &= \frac{2\mue + \lame}{\rho}u_{1,11} + \frac{\mue + \muc}{\rho}u_{1,22} + \frac{\mue - \muc + \lame}{\rho}u_{2,12} - 2\frac{\mue}{\rho}P^{D}_{1,1} - \frac{3\lame + 2\mue}{\rho}P^{S}_{,1} -2\frac{\mue}{\rho}P_{(12),2} - 2\frac{\muc}{\rho}P_{[12],2},  \\
u_{2,tt} &= \frac{\mue - \muc + \lame}{\rho}u_{1,12} +  \frac{\mue + \muc}{\rho}u_{2,11} + \frac{2\mue + \lame}{\rho}u_{2,22}  - 2\frac{\mue}{\rho}P^{D}_{2,2} - \frac{3\lame + 2\mue}{\rho}P^{S}_{,2} -2\frac{\mue}{\rho}P_{(12),1} + 2\frac{\muc}{\rho}P_{[12],1},  \\
P^{D}_{1,tt} &= \frac{4}{3}\frac{\mue}{\eta}u_{1,1} - \frac{2}{3}\frac{\mue}{\eta}u_{2,2} - 2\frac{\mue + \mumic}{\eta} P^{D}_{1} +\frac{\mue \Lc}{3\eta}P^{D}_{1,11} + \frac{\mue \Lc}{\eta}P^{D}_{1,22} + \frac{\mue \Lc}{3\eta}P^{D}_{2,22} - \frac{2}{3}\frac{\mue \Lc}{\eta}P^{S}_{,11} + \frac{\mue \Lc}{3\eta}P^{S}_{,22}\\
& - \frac{\mue Lc}{3\eta}P_{(12),12} - \frac{\mue \Lc}{\eta}P_{[12],12},\\
P^{D}_{2,tt} &= -\frac{2}{3}\frac{\mue}{\eta}u_{1,1} + \frac{4}{3} \frac{\mue}{\eta}u_{2,2} +\frac{\mue \Lc}{3\eta}P^{D}_{1,11} - 2\frac{\mue + \mumic}{\eta} P^{D}_2 + \frac{\mue \Lc}{\eta} P^{D}_{2,11} + \frac{\mue \Lc}{3\eta}P^{D}_{2,22} + \frac{\mue \Lc}{3\eta}P^{S}_{,11} -\frac{2}{3}\frac{\mue \Lc}{\eta} P^{S}_{,22}\\
& - \frac{\mue \Lc}{3\eta}P_{(12),12} + \frac{\mue \Lc}{\eta} P_{[12],12}, \\
P^{S}_{,tt} &= \frac{2\mue + 3\lame}{3\eta}u_{1,1} + \frac{2\mue + 3\lame}{3\eta}u_{2,2} - \frac{\mue \Lc}{3\eta}P^{D}_{1,11} - \frac{\mue \Lc}{3\eta}P^{D}_{2,22} - \frac{(2\mue + 3\lame) + (2\mumic + 3\lammic)}{\eta} P^{S} +\frac{2}{3}\frac{\mue \Lc}{\eta}P^{S}_{,11} \\
&+ \frac{2}{3}\frac{\mue \Lc}{\eta}P^{S}_{,22}  - \frac{2}{3}\frac{\mue \Lc}{\eta}P_{(12),12}, \\
P_{(12),tt} &= \frac{\mue}{\eta}u_{1,2} + \frac{\mue}{\eta}u_{2,1} - \frac{1}{2}\frac{\muc \Lc}{\eta} P^{D}_{1,12}  - \frac{1}{2}\frac{\muc \Lc}{\eta} P^{D}_{2,12} - \frac{\muc \Lc}{\eta} P^{S}_{,12}	 -2\frac{\mue + \mumic}{\eta}P_{(12)} + \frac{1}{2}\frac{\muc \Lc}{\eta} P_{(12),11} + \frac{1}{2}\frac{\muc \Lc}{\eta} P_{(12),22}\\ 
& +\frac{1}{2}\frac{\muc \Lc}{\eta} P_{[12],11} - \frac{1}{2}\frac{\muc \Lc}{\eta}P_{[12],22}, \\
P_{[12],tt} &= \frac{\muc}{\eta}u_{1,2} - \frac{\muc}{\eta}u_{2,1}- \frac{1}{2}\frac{\mue \Lc}{\eta} P^{D}_{1,12} + \frac{1}{2}\frac{\mue \Lc}{\eta} P^{D}_{2,12}+ \frac{1}{2}\frac{\mue \Lc}{\eta}P_{(12),11}- \frac{1}{2}\frac{\mue \Lc}{\eta} P_{(12),22}-2\frac{\muc}{\eta}P_{[12]}  + \frac{1}{2}\frac{\mue \Lc}{\eta} P_{[12],11}\\
&  + \frac{1}{2}\frac{\mue \Lc}{\eta} P_{[12],22}, \\
u_{3,tt} &= \frac{\mue + \muc}{\rho}(u_{3,11} + u_{3,22}) - 2\frac{\mue}{\rho} P_{(13),1} + 2\frac{\muc}{\rho}P_{[13],1} - 2\frac{\mue}{\rho}P_{(23),2} + 2\frac{\muc}{\rho}P_{[23],2}, \\
P_{(13),tt} &= \frac{\mue}{\eta}u_{3,1} - 2\frac{\mue + \mumic}{\eta}P_{(13)} + \frac{1}{2}\frac{\mue \Lc}{\eta} P_{(13),11} + \frac{\mue \Lc}{\eta} P_{(13),22} + \frac{1}{2}\frac{\mue \Lc}{\eta}P_{[13],11} - \frac{1}{2}\frac{\mue \Lc}{\eta}P_{(23),12} + \frac{1}{2}\frac{\mue \Lc}{\eta} P_{[23],12},\\
P_{[13],tt} &= -\frac{\muc}{\eta}u_{3,1} + \frac{1}{2}\frac{\mue \Lc}{\eta}P_{(13),11} - 2\frac{\muc}{\eta}P_{[13]} + \frac{1}{2}\frac{\mue \Lc}{\eta}P_{[13],11}+ \frac{\mue \Lc}{\eta} P_{[13],22}  + \frac{1}{2}\frac{\mue \Lc}{\eta}P_{(23),12} - \frac{1}{2}\frac{\mue \Lc}{\eta}P_{[23],12},  \\
P_{(23),tt} &= \frac{\mue}{\eta}u_{3,2} - \frac{1}{2}\frac{\mue \Lc}{\eta}P_{(13),12}+ \frac{1}{2}\frac{\mue \Lc}{\eta}P_{[13],12}- 2\frac{\mue + \mumic}{\eta}P_{(23)} + \frac{\mue \Lc}{\eta}P_{(23),11} + \frac{1}{2}\frac{\mue \Lc}{\eta} P_{(23),22} +  \frac{1}{2}\frac{\mue \Lc}{\eta} P_{[23],22}, \\
P_{[23],tt} &= -\frac{\muc}{\eta}u_{3,2} + \frac{1}{2}\frac{\mue \Lc}{\eta}P_{(13),12} - \frac{1}{2}\frac{\mue \Lc}{\eta}P_{[13],12} + \frac{1}{2}\frac{\mue \Lc}{\eta} P_{(23),22} - 2\frac{\muc}{\eta}P_{[23]} + \frac{\mue \Lc}{\eta}P_{[23],11}  +  \frac{1}{2}\frac{\mue \Lc}{\eta} P_{[23],22}  . 
\end{align*}
Collect the new variables as 
\begin{align}
v^1 &= \left(u_1, u_2, P_1^D, P_2^D, P^S, P_{(12)}, P_{[12]}\right)^T, \\
v^2 &= \left(u_3, P_{(13)}, P_{[13]}, P_{(23)}, P_{[23]}\right)^T.
\end{align}

\subsection{The matrices $A_1$ and $A_2$}\label{appendixRelaxed3}
\small
The form of the matrices being too complicated to fit in one page, we present both $A_1$ and $A_2$ in a column-wise sense.

\begin{align}
\begin{array}{cc}
A_{11} = \left( \begin{array}{c}
\left( k_1^2  \frac{2\mue + \lame}{\rho} + k_2^2 \frac{\mue + \muc}{\rho}\right)-\omega^2 \\
k_1 k_2 \frac{\mue - \muc + \lame}{\rho} \\
-2ik_1 \frac{\mue}{\rho}\\
0\\
-ik_1 \frac{3\lame + 2\mue }{\rho}\\
-2ik_2 \frac{\mue}{\rho}\\
2ik_2 \frac{\mue}{\rho}
  \end{array} \right)^T,&  A_{12} = \left( \begin{array}{c}
  k_1 k_2 \frac{\mue - \muc + \lame}{\rho} \\
  \left(k_1^2 \frac{\mue + \muc}{\rho} +k_2^2 \frac{2\mue + \lame}{\rho}\right)-\omega^2\\
  0\\
 - 2i k_2 \frac{\mue}{\rho}\\
  -i k_2 \frac{3\lame + 2\mue}{\rho}\\
  -2i k_1 \frac{\mue}{\rho}\\
  2i k_1 \frac{\mue}{\rho}
  \end{array} \right)^T, 
  \end{array}
  \end{align}
  \begin{align}
  \begin{array}{cc}
A_{13} = \left( \begin{array}{c}
\frac{4}{3}i k_1 \frac{\mue}{\eta}\\
-\frac{2}{3}i k_2 \frac{\mue}{\eta}\\
\left(k_1^2 \frac{\mue Lc}{3\eta}+ k_2^2 \frac{\mue \Lc}{\eta}\right)+2\frac{\mue + \mumic}{\eta}-\omega^2\\
k_2^2\frac{\mue \Lc}{3\eta}\\
\left(k_2^2\frac{\mue\Lc}{3\eta}-k_1^2\frac{\mue\Lc}{\eta}\right)\\
-k_1k_2 \frac{\mue\Lc}{3\eta}\\
-k_1 k_2 \frac{\mue\Lc}{\eta}
  \end{array} \right)^T, & A_{14} = \left( \begin{array}{c}
  -\frac{2}{3}i k_1 \frac{\mue}{\eta}\\
  +\frac{4}{3}ik_2\frac{\mue}{\eta}\\
  k_1^2\frac{\mue\Lc}{3\eta}\\
  \left(k_1^2\frac{\mue\Lc}{\eta}+k_2^2\frac{\mue\Lc}{3\eta}\right)+2\frac{\mue+\mumic}{\eta}-\omega^2\\
  \left(k_1^2\frac{\mue\Lc}{3\eta}-\frac{2}{3} k_2^2 \frac{\mue\Lc}{\eta}\right)\\
  -k_1k_2\frac{\mue\Lc}{3\eta}\\
  k_1k_2\frac{\mue\Lc}{\eta}
  \end{array} \right)^T,\\
A_{15}=\left(\begin{array}{c}
i k_1 \frac{2\mue+3\lame}{3\eta}\\
i k_2\frac{2\mue+3\lame}{3\eta}\\
-k_1^2\frac{\mue\Lc}{3\eta}\\
-k_2^2 \frac{\mue\Lc}{3\eta}\\
\frac{2}{3}(k_1^2+k_2^2)\frac{\mue\Lc}{\eta}+\frac{(2\mue+3\lame)+(2\mumic+3\lammic)}{\eta}-\omega^2\\
-\frac{2}{3}k_1k_2\frac{\mue\Lc}{\eta}\\
0
\end{array}\right)^T, & A_{16}=\left(\begin{array}{c}
ik_2\frac{\mue}{\eta}\\
ik_1\frac{\mue}{\eta}\\
-\frac{1}{2} k_1 k_2 \frac{\mue\Lc}{\eta}\\
-\frac{1}{2}k_1 k_2 \frac{\mue\Lc}{\eta}\\
-k_1 k_2 \frac{\mue\Lc}{\eta}\\
\frac{1}{2}(k_1^2+k_2^2)\frac{\mue\Lc}{\eta}+2\frac{\mue+\mumic}{\eta}-\omega^2\\
\frac{1}{2}(k_1^2-k_2^2)\frac{\mue\Lc}{\eta}
\end{array}\right)^T,\\
\end{array}
\end{align}
\begin{align}
A_{17}&=\left(\begin{array}{c}
ik_2\frac{\muc}{\eta}\\
-ik_1\frac{\muc}{\eta}\\
-\frac{1}{2}k_1k_2\frac{\mue\Lc}{\eta}\\
\frac{1}{2} k_1 k_2 \frac{\mue\Lc}{\eta}\\
0\\
\frac{1}{2}(k_1^2-k_2^2)\frac{\mue\Lc}{\eta}\\
\frac{1}{2}(k_1^2+k_2^2)\frac{\mue\Lc}{\eta}+2\frac{\muc}{\eta}-\omega^2
\end{array}\right)^T.
\end{align}
Then, the matrix $A_1$ is 
\begin{equation}
A_1=(A_{11},A_{12},A_{13},A_{14},A_{15},A_{16},A_{17})^T.
\end{equation}

As for $A_2$ we have:
\begin{align}
\begin{array}{cc}
A_{21}=\left(\begin{array}{c}
(k_1^2+k_2^2)\frac{\mue+\muc}{\rho}-\omega^2\\
2ik_1\frac{\mue}{\rho}\\
-2ik_1\frac{\muc}{\rho}\\
2ik_2\frac{\mue}{\rho}\\
-2ik_2\frac{\muc}{\rho}
\end{array}\right)^T&
A_{22}=\left(\begin{array}{c}
-ik_1\frac{\mue}{\eta}\\
\left(k_1^2\frac{\mue\Lc}{2\eta}+k_2^2\frac{\mue\Lc}{\eta}\right)+2\frac{\mue+\mumic}{\eta}-\omega^2\\
k_1^2\frac{\mue\Lc}{2\eta}\\
-k_1k_2\frac{\mue\Lc}{2\eta}\\
k_1k_2\frac{\mue\Lc}{2\eta}
\end{array}\right)^T \\
A_{23}=\left(\begin{array}{c}
ik_1\frac{\muc}{\eta}\\
k_1^2\frac{\mue\Lc}{2\eta}\\
\left(k_1^2\frac{\mue\Lc}{2\eta}+k_2^2\frac{\mue\Lc}{\eta}\right)+2\frac{\muc}{\eta}-\omega^2\\
k_1k_2\frac{\mue\Lc}{2\eta}\\
-k_1k_2\frac{\mue\Lc}{2\eta}
\end{array}\right)^T&
A_{24}=\left(\begin{array}{c}
-ik_2\frac{\mue}{\eta}\\
-k_1k_2\frac{\mue\Lc}{2\eta}\\
k_1k_2\frac{\mue\Lc}{2\eta}\\
\left(k_1^2\frac{\mue\Lc}{\eta}+k_2^2\frac{\mue\Lc}{2\eta}\right)+2\frac{\mue+\mumic}{\eta}-\omega^2\\
k_2^2\frac{\mue\Lc}{2\eta}
\end{array}\right)^T
\end{array}
\end{align}
\begin{align}
A_{25}&=\left(\begin{array}{c}
ik_2\frac{\muc}{\eta}\\
k_1k_2\frac{\mue\Lc}{2\eta}\\
-k_1k_2\frac{\mue\Lc}{2\eta}\\
k_2^2\frac{\mue\Lc}{2\eta}\\
\left(k_1^2\frac{\mue\Lc}{\eta}+k_2^2\frac{\mue\Lc}{2\eta}\right)+2\frac{\muc}{\eta}-\omega^2
\end{array}\right)^T.
\end{align}
Then, the matrix $A_2$ is 
\begin{equation}
A_2 = (A_{21},A_{22},A_{23},A_{24},A_{25})^T.
\end{equation}
\subsection{Dispersion curves analysis of the relaxed micromorphic model}\label{appendixRelaxed4}
\small
\subsubsection{In-plane variables}\label{appendixRelaxed41}
First, observe that after replacing the wave form \eqref{relaxedplanewave1} in \eqref{eq:relaxedeqns}, $A_1$ can be written as (see Appendix \ref{appendixRelaxed5} for a demonstration of these matrices):
\begin{equation}\label{eq:A1decomp}
A_1 = \kabs^2A_1^R-\omega^2\idmat -i\kabs B_1^R-C_1^R.
\end{equation}

We now investigate the behaviour of the polynomial $\det A_1$ for the two limiting cases $\kabs \to 0$ and $\kabs \to \infty$, the final goal being to be able to determine the cut-off frequencies and asymptotes of the dispersion curves. 

Letting $\kabs \to 0$ and using \eqref{eq:A1decomp}, the equation $A_1 \cdot v^1=0$ becomes 
\begin{equation}\label{eq:A1k0}
(\omega^2 \idmat + C^R_1)\cdot v^1 = 0. 
\end{equation}
Using expression \eqref{CR1} for $C^R_1$ given in Appendix \ref{appendixRelaxed}, we get 
\begin{equation}
\left( \begin{array}{ccccccc}
\omega^2 & 0 & 0 & 0 & 0 & 0 & 0 \\
0 & \omega^2 &0 & 0 & 0 & 0 & 0 \\
0 & 0 & \omega^2 - \omega_s^2 & 0 & 0 & 0 & 0 \\
0 & 0 & 0 & \omega^2 - \omega_s^2 & 0 & 0 & 0 \\
0 & 0 & 0 & 0 & \omega^2 - \omega_p^2 & 0 & 0 \\
0 & 0 & 0 & 0 & 0 & \omega^2 - \omega_s^2 & 0 \\
0 & 0 & 0 & 0 & 0 & 0 & \omega^2 - \omega_r^2 
\end{array} \right)\cdot v^1 =0.
\end{equation}
This allows us to deduce that for small values of $\kabs$, the cut-off frequencies of the dispersion curves are $\omega = 0$, $\omega = \omega_s$, $\omega = \omega_r$, $\omega = \omega_p$ as shown in Figure \ref{fig:dispA1}. 

As for when $\kabs \to \infty$, we consider the case where the ratio $\kabs/\omega$ remains finite and so instead of studying the whole system we can simply regard the reduced system 
\begin{equation}
\left(\kabs^2 A^R_1 - \omega^2 \idmat\right)\cdot v^1 = 0.
\end{equation}

The determinant of $\kabs^2 A^R_1 - \omega^2 \idmat$ is
\begin{equation}\det \left(\kabs^2 A^R_1 - \omega^2 \idmat\right) = \omega^4 (\kabs^2 c_m^2 - \omega^2)^3(\kabs^2c_f^2-\omega^2)(\kabs^2 c_p^2 -\omega^2).\end{equation}
If we solve the  equation $\det \left(\kabs^2 A^R_1 - \omega^2 \idmat\right) = 0$, we find the solution $\omega = 0$, which has to be excluded since it violates the requirement $\kabs/\omega$ finite for $\kabs \to \infty$.  The other solutions are $\omega = c_m \kabs$, $\omega=c_f \kabs$, $\omega = c_p \kabs$; these are the asymptotes to the dispersion curves.

\begin{figure}[h!]
	\centering
	\begin{tabular}{ccc}
		\includegraphics[scale=0.65]{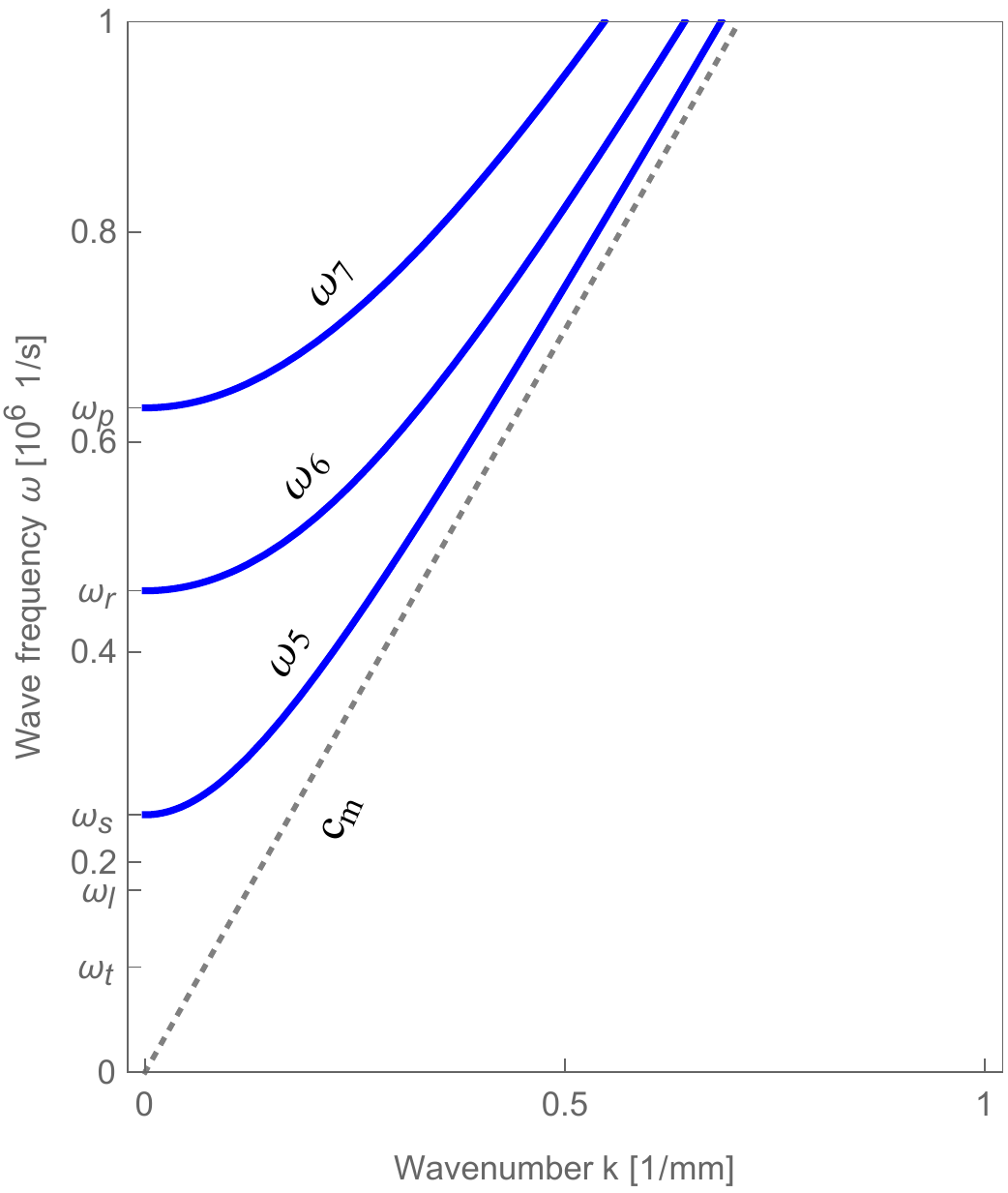} & 
		\includegraphics[scale=0.65]{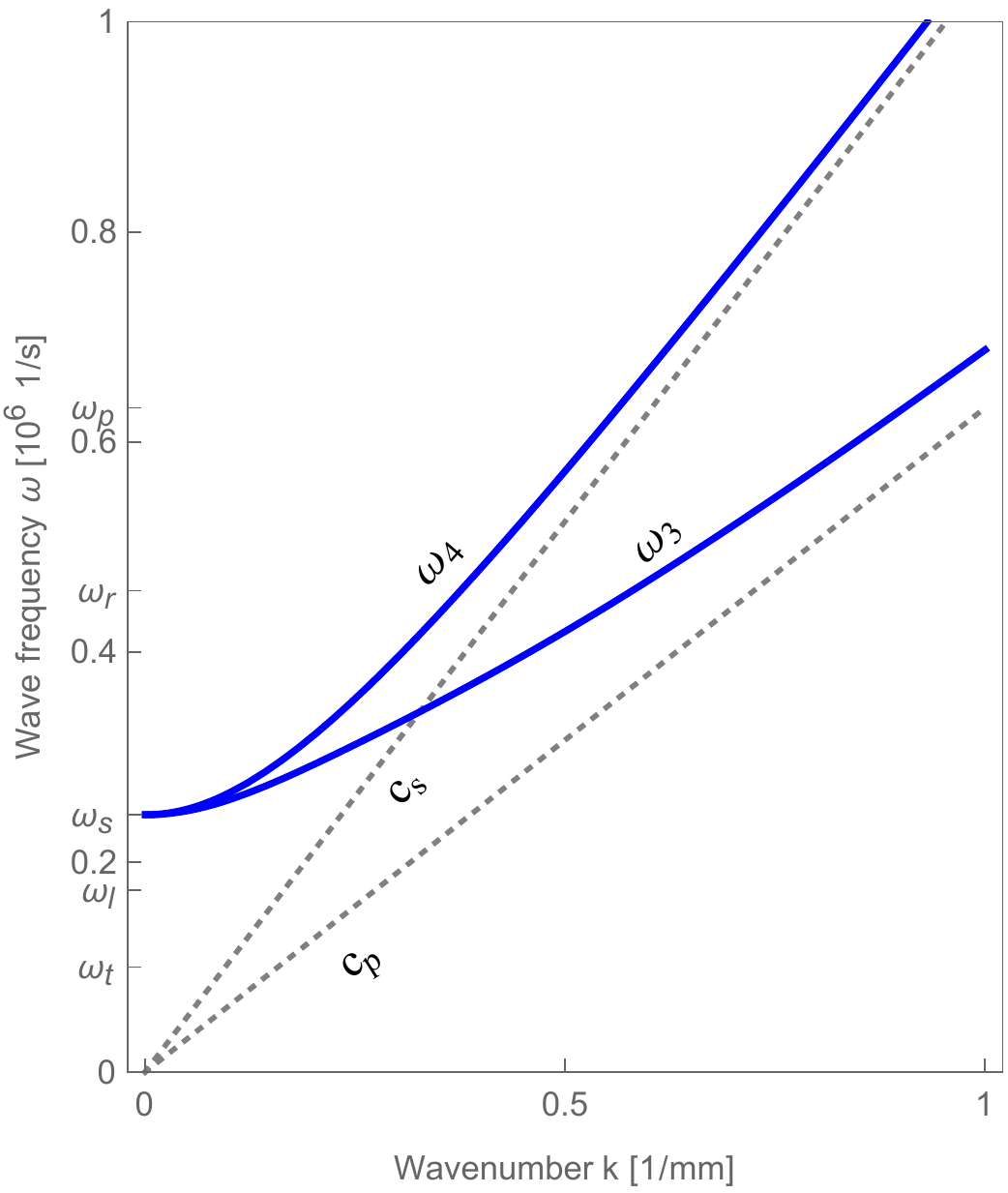}
		\tabularnewline
		(a) & (b)  
	\end{tabular}
	\caption{\small (a) Three distinct modes have distinct cut-off frequencies but a common asymptote with slope $c_m$, (b) Two distinct modes have the same cut-off frequencies but asymptotes with different slopes $c_f$ and $c_p$.}
	\label{fig:A1asymp}
\end{figure}
\begin{figure}[h!]
	\centering
	\includegraphics[scale=0.65]{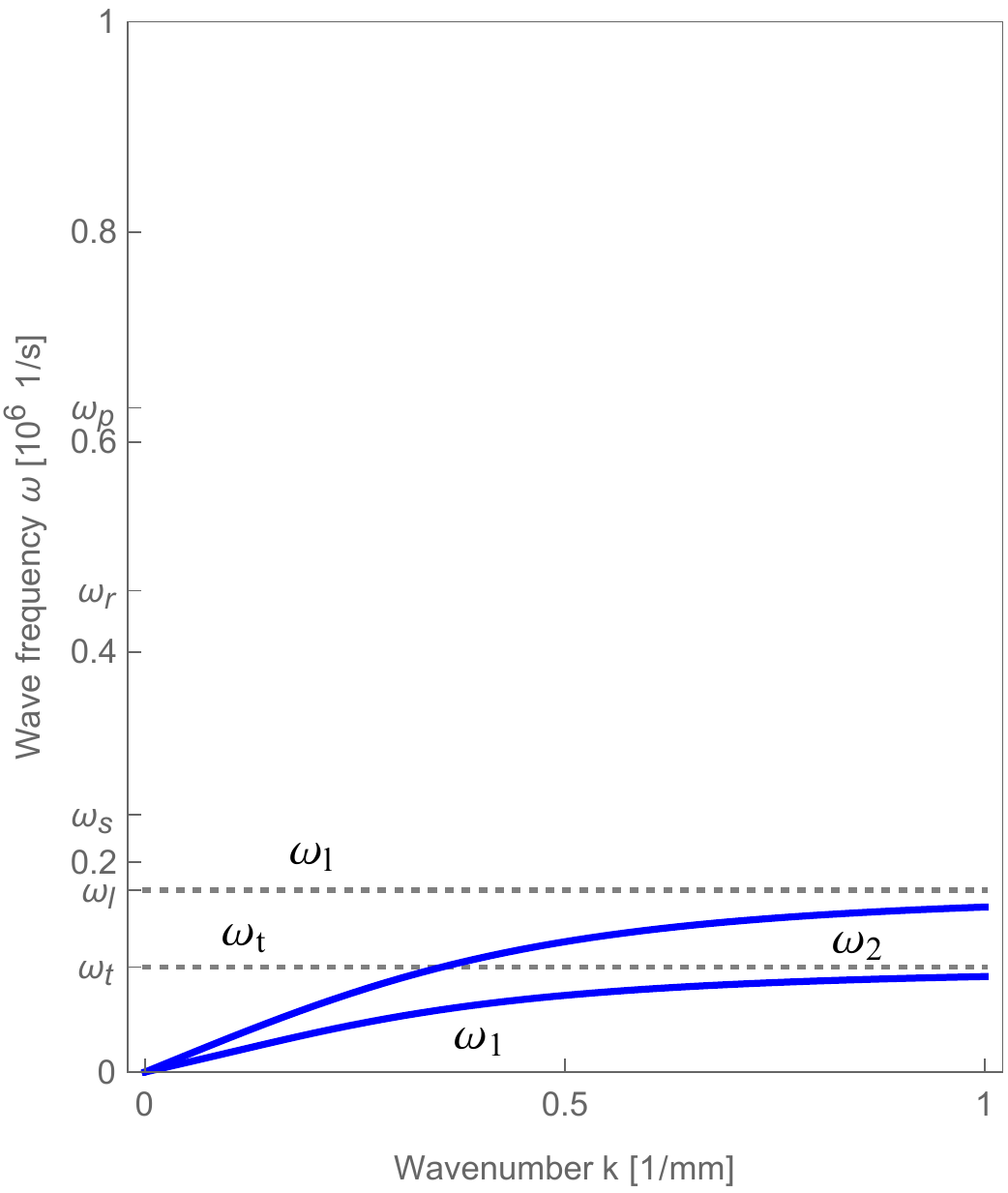} 
	\caption{\small Two modes with horizontal asymptotes $\omega_l$ and $\omega_t$.}
	\label{fig:A1horizasymp}
\end{figure}

Finally, the two horizontal asymptotes are $\omega = \omega_l$ and $\omega = \omega_t$ (see Section \ref{sec:horizasymptotes} for the analytical method of how to find them).

The first, second and fifth modes all have the same asymptote $\omega = c_m \kabs$ (Figure \ref{fig:A1asymp} (a)). Unfortunately, due to the very complicated expressions for $\omega_2(\kabs)$ and $\omega_3(\kabs)$, this cannot be analytically reaffirmed; numerically, however, we find:
\[
\lim_{\kabs\to \infty} |\omega_i(\kabs) - c_m \kabs|=0,
\]
for $i=1,2,5$.

Figure \ref{fig:dispA1} summarizes the main characteristics of the dispersion curves that we find for an isotropic relaxed micromorphic medium (cut-off frequencies, horizontal and oblique asymptotes). Figures \ref{fig:A1asymp} (a), (b) and \ref{fig:A1horizasymp} show in more detail the oblique and horizontal asymptotes for the different modes.

\subsubsection{Out-of-plane variables}\label{appendixRelaxed42}
Similarly to the in-plane case, we can again write 
\begin{equation}
A_2 = |\widetilde{k}|^2A_2^R-\omega^2\idmat-i|\widetilde{k}|B_2^R-C_2^R.
\end{equation}
We now investigate the behavior of the polynomial $\det A_2$ for the two limiting cases $|\widetilde{k}| \to 0$ and $|\widetilde{k}| \to \infty$. 

Letting $|\widetilde{k}| \to 0$, equation $A_2 \cdot v^2 = 0$ becomes 
\begin{equation}
\left( \begin{array}{ccccc}
\omega^2 & 0 & 0 & 0 & 0 \\
0 & \omega^2 - \omega_s^2 & 0 & 0 & 0 \\
0 & 0 & \omega^2 - \omega_r^2 & 0 & 0 \\
0 & 0 & 0 & \omega^2 - \omega_s^2 & 0 \\
0 & 0 & 0 & 0 & \omega^2 -\omega_r^2
\end{array}\right) \cdot v^2 = 0.
\end{equation}
This allows us to deduce that for small values of $|\widetilde{k}|$, the cut-off frequencies of the dispersion curves are $\omega = 0$, $\omega = \omega_s$, $\omega = \omega_r$ as shown in Figure \ref{fig:dispA2}.

We now let $|\widetilde{k}| \to \infty$ and we once again assume that the ratio $|\widetilde{k}|/\omega$ remains finite. Then, the system reduces to 
\begin{equation}
E_2 \cdot v^2 := (|\widetilde{k}|^2 A^R_2 - \omega^2 \idmat)\cdot v^2 = 0,
\end{equation}
where\footnote{We present the form such a matrix takes for the out-of-plane variables in the main text since it has a more transparent form, being of smaller dimension.} 
\begin{equation}
E_2=\left( \begin{array}{ccccc}
|\widetilde{k}|^2 c_f^2 - \omega^2 & 0 & 0 & 0 & 0\\
0 & |\widetilde{k}|^2 c_m^2\left(\frac{1}{2}\xi_1^2+\xi_2^2\right)-\omega^2 & \frac{1}{2}|\widetilde{k}|^2 c_m^2\xi_1^2  & -\frac{1}{2}|\widetilde{k}|^2 c_m^2 \xi_1\xi_2 &\frac{1}{2}|\widetilde{k}|^2 c_m^2 \xi_1\xi_2\\
0 &\frac{1}{2}|\widetilde{k}|^2 c_m^2\xi_1^2 & |\widetilde{k}|^2 c_m^2\left(\frac{1}{2}\xi_1^2 + \xi_2^2\right)-\omega^2 & \frac{1}{2}|\widetilde{k}|^2 c_m^2 \xi_1\xi_2 & -\frac{1}{2}|\widetilde{k}|^2 c_m^2 \xi_1\xi_2 \\
0 &-\frac{1}{2}|\widetilde{k}|^2 c_m^2 \xi_1\xi_2 & \frac{1}{2}|\widetilde{k}|^2 c_m^2 \xi_1\xi_2 &|\widetilde{k}|^2 c_m^2(\xi_1^2 + \frac{1}{2}\xi_2^2)-\omega^2 & \frac{1}{2}\xi_|\widetilde{k}|^2 c_m^22^2 \\
0 &\frac{1}{2}|\widetilde{k}|^2 c_m^2 \xi_1\xi_2 & -\frac{1}{2}|\widetilde{k}|^2 c_m^2 \xi_1\xi_2 & \frac{1}{2}|\widetilde{k}|^2 c_m^2\xi_2^2 &|\widetilde{k}|^2 c_m^2( \xi_1^2 + \frac{1}{2}\xi_2^2 )-\omega^2
\end{array} \right). .
\end{equation}
The determinant of $E_2$ is 
\begin{equation}
\det E_2 = -\omega^2 (|\widetilde{k}|^2c_m^2-\omega^2)^3 (|\widetilde{k}|^2c_f^2-\omega^2).
\end{equation}
The solutions of $\det E_2 =0$ obviously are $\omega = 0, \omega = c_m, \omega = c_f$. We have to exclude the first since it violates the condition that $k/\omega$ is finite for $|\widetilde{k}| \to \infty$ and so we conclude that the two non-horizontal asymptotes are the two remaining solutions. Finally, the only horizontal asymptote is $\omega = \omega_t$ (see Section \ref{sec:horizasymptotes}).
\begin{figure}[h!]
	\begin{centering}
		\begin{tabular}{ccc}
			\includegraphics[scale=0.65]{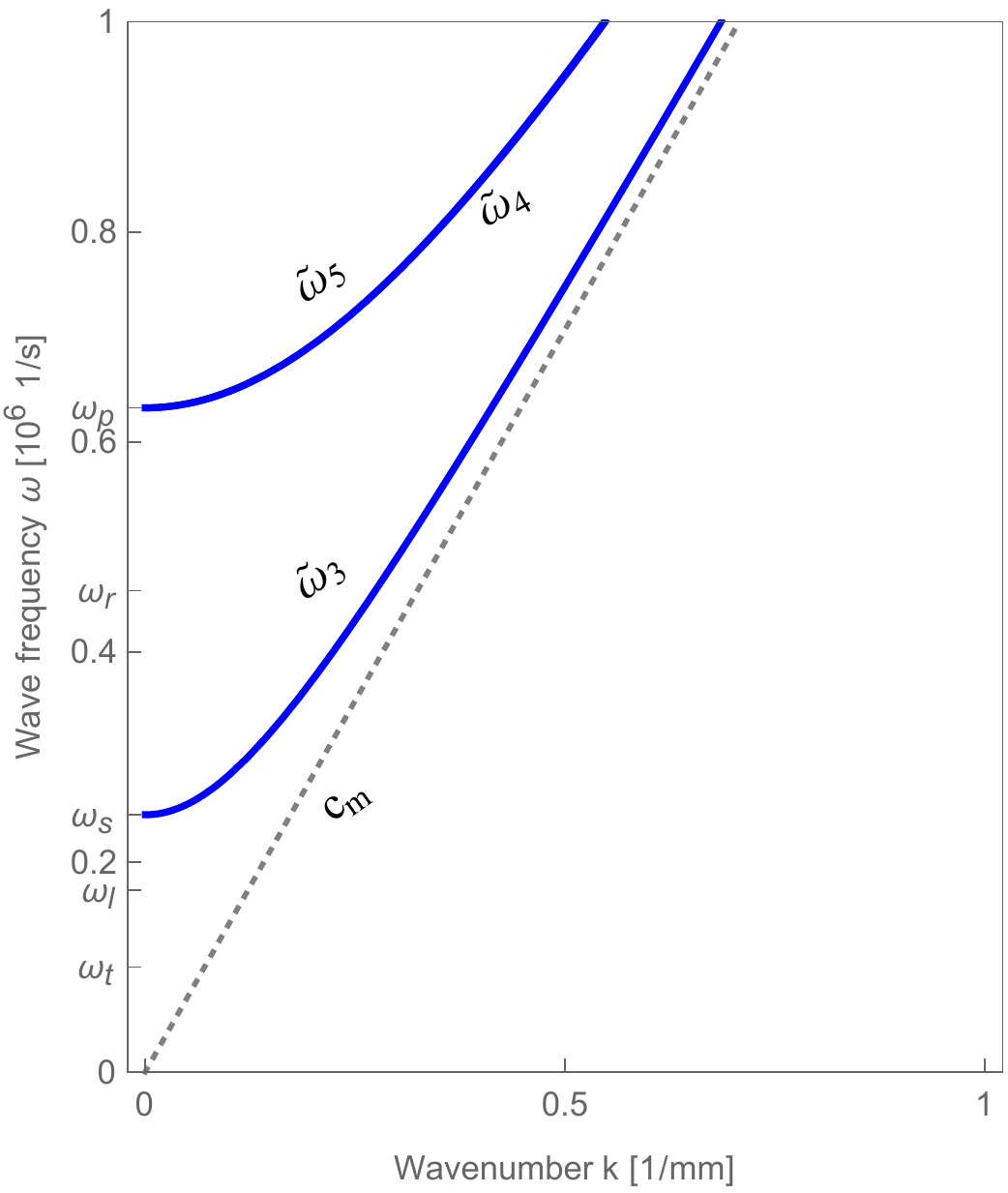} & &
			\includegraphics[scale=0.65]{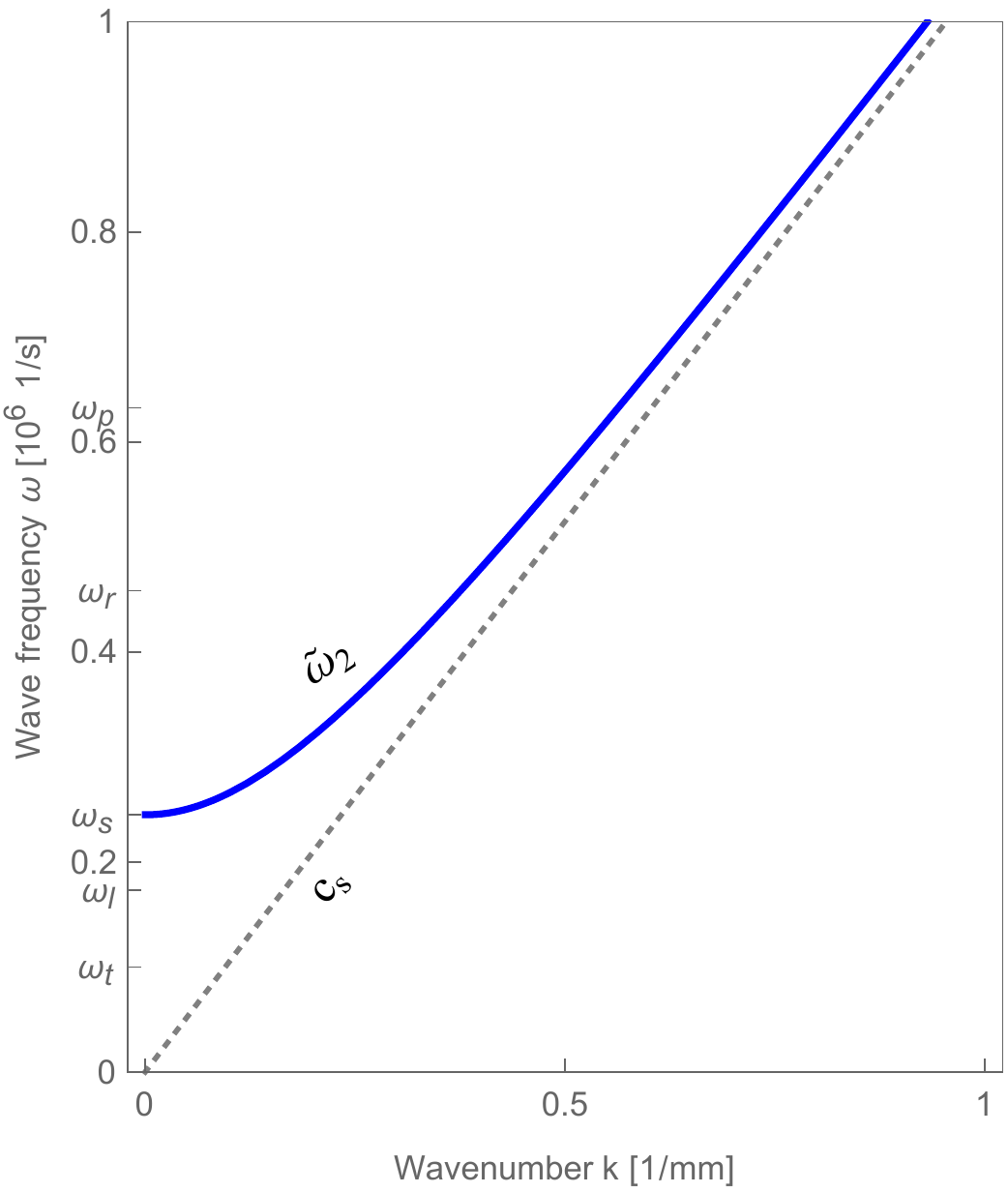}
			\tabularnewline
			(a) &  & (b)  \tabularnewline
		\end{tabular}
		\par\end{centering}
	\caption{\small (a) First mode with cut-off frequency $\omega_p$ and second and third modes which overlap for large values of $|\widetilde{k}|$, all have the same asymptote $c_m$, (b) The fourth mode has an asymptote with slope $c_f$.}
	\label{fig:A2oblique}
\end{figure}
\begin{figure}[h!]
	\begin{centering}
		\begin{tabular}{c}
			\includegraphics[scale=0.65]{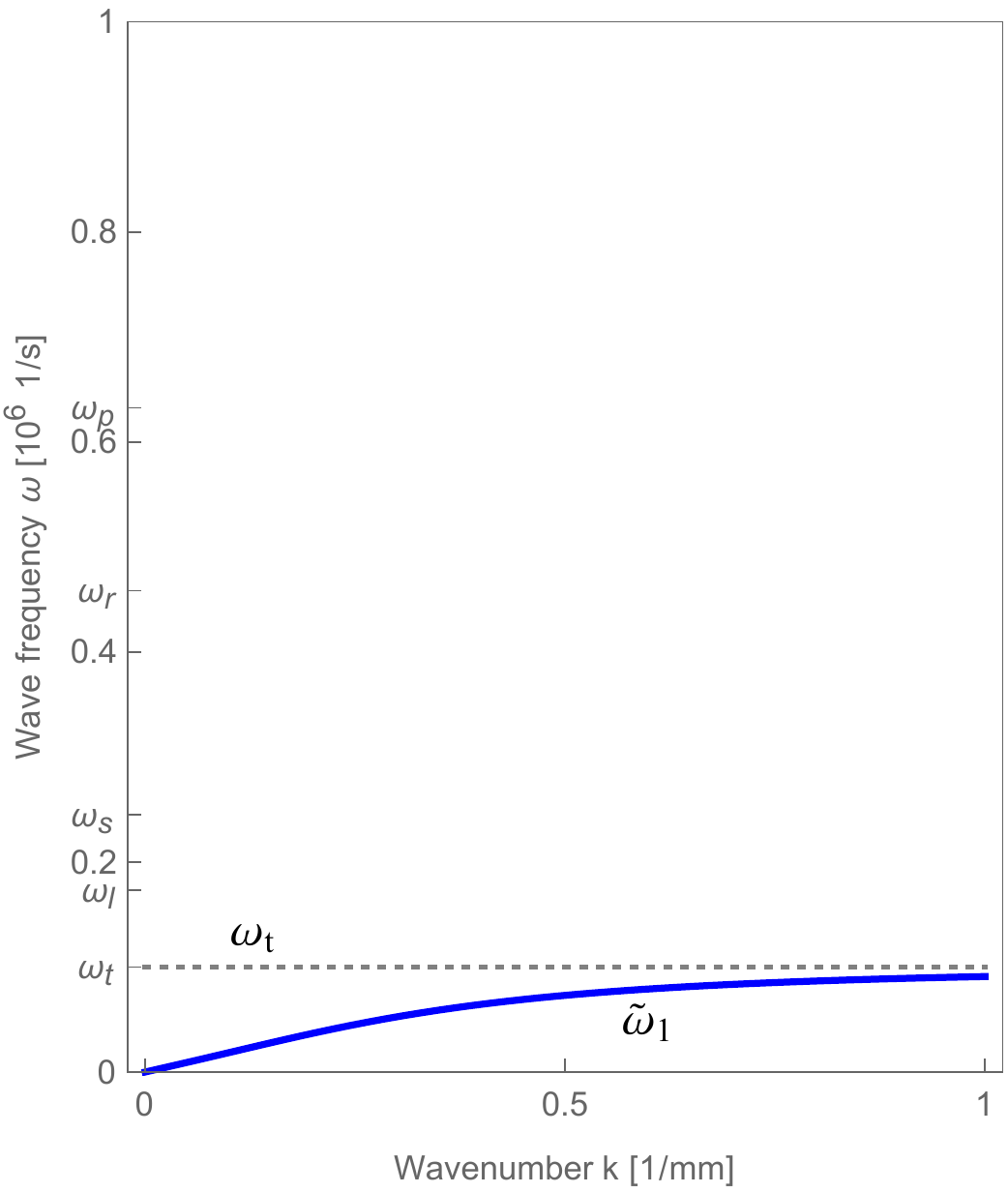} 
		\end{tabular}
		\par
	\end{centering}
	\caption{\small The fourth mode has the horizontal asymptote $\omega_t$.}
	\label{fig:A2horizontal}
\end{figure}

In this case, we have that as $|\widetilde{k}|\to\infty$, the two modes $\widetilde{\omega}_2$ and $\widetilde{\omega}_3$ overlap. Indeed, it can be checked (numerically) that
\[
\lim_{|\widetilde{k}|\to \infty} |\widetilde{\omega}_2(|\widetilde{k}|)-\widetilde{\omega}_3(|\widetilde{k}|)|=0.
\]
Furthermore, we again see that the first, second and third modes have the same asymptote $\omega = c_m |\widetilde{k}|$, a fact which cannot be analytically shown due to the complicated expressions involved; however, numerically we find that
\[
\lim_{|\widetilde{k}|\to \infty} |\widetilde{\omega}_i(|\widetilde{k}|) - c_m |\widetilde{k}||=0,
\]
for $i=1,2,3$.

Figure \ref{fig:dispA2} summarizes the main characteristics of the second set of dispersion curves for the relaxed micromorphic medium. Figures \ref{fig:A2oblique} and \ref{fig:A2horizontal} show in more detail the oblique and horizontal asymptotes for different curves.


\subsubsection{How to find the horizontal asymptotes}\label{sec:horizasymptotes}
In order to be sure that the horizontal asymptotes, computed before by splitting the whole problem of the relaxed model in two uncoupled problems, are $\omega_t$ and $\omega_l$, we follow what was done in \cite{dagostino2017panorama}.

We begin by constructing the following matrix:
\begin{equation}\label{eq:A0}
A_0=\left(
\begin{array}{cccccccccccc}
& & & & & & & 0 & 0 & 0 & 0 & 0 \\
& & & & & & & 0 & 0 & 0 & 0 & 0 \\
& & & & & & & 0 & 0 & 0 & 0 & 0 \\
& & & \mathlarger{\mathlarger{A_1}} & & & & 0 & 0 & 0 & 0 & 0 \\
& & & & & & & 0 & 0 & 0 & 0 & 0 \\
& & & & & & & 0 & 0 & 0 & 0 & 0 \\
& & & & & & & 0 & 0 & 0 & 0 & 0 \\
0 & 0 & 0 & 0 & 0 & 0 & 0 & & & & & \\
0 & 0 & 0 & 0 & 0 & 0 & 0 & & & & & \\
0 & 0 & 0 & 0 & 0 & 0 & 0 & & & \mathlarger{\mathlarger{A_2}} & & \\
0 & 0 & 0 & 0 & 0 & 0 & 0 & & & & & \\
0 & 0 & 0 & 0 & 0 & 0 & 0 & & & & & \\
\end{array}\right),
\end{equation}
which is the full matrix of coefficients of the governing equations \eqref{eq:relaxedeqns} after making the plane-wave ansatz. Furthermore, we define the two auxiliary matrices
\begin{equation}
\widehat{A}_1 = \left(\begin{array}{cc}
A_1 & 0 \\
0 & \idmat_5
\end{array}\right), \quad
\widehat{A}_2 = \left(\begin{array}{cc}
\idmat_7 & 0 \\
0 & A_2
\end{array}\right).
\end{equation}
We immediately see that
\begin{equation}
A_0=\widehat{A}_1 \widehat{A}_2,
\end{equation}
which implies that
\begin{equation}
\det A_0 = \det \widehat{A}_1 \det \widehat{A}_2 = \det A_1 \det A_2.
\end{equation}
Now, the polynomial $p(\kabs,\omega) = \det A_0$ is of degree $18$ in $\kabs$ and only involves even powers of $\kabs$. We can write it as 
\begin{equation}
p(\kabs,\omega) = \sum _{h=0}^{9}c_{2h}(\omega^2)\kabs^{2h}.
\end{equation}

Following the same proof as found in \cite{dagostino2017panorama}, we find again the necessary condition a horizontal asymptote has to satisfy is $c_{18}(\omega_{*}^2)=0$.\footnote{This result follows if we divide by $\kabs^{18}$ and let $\kabs\to \infty$.} This gives us the two solutions 
\begin{equation}\label{eq:horizasymptote}
\omega = \omega_t, \quad \omega = \omega_l.
\end{equation}

\subsection{Decompositions of $A_1$ and $A_2$}\label{appendixRelaxed5}
We recall that 
\begin{align}
A_1 &= \kabs^2 A^R_1-\omega^2 \idmat - i \kabs B^R_1 - C^R_1,\\
A_2 &= \kabs^2 A^R_2-\omega^2 \idmat - i \kabs B^R_2 - C^R_2.
\end{align}
The matrices $A^R_1, B^R_1, C^R_1 , A^R_2, B^R_2, C^R_2$, are presented here. 
\scriptsize
\begin{equation}
A^R_1=\left(\begin{array}{ccccccc}
\xi_1^2\frac{\lame+2\mue}{\rho}+\xi_2^2\frac{\muc+\mue}{\rho}&
\xi_1\xi_2\frac{\mue-\muc+\lame}{\rho}&0&0&0&0&0\\
-\xi_1\xi_2\frac{\mue-\muc+\lame}{\rho}&
\xi_1^2\frac{\muc+\mue}{\rho}+\xi_2^2\frac{\lame+2\mue}{\rho}&0&0&0&0&0\\
0&0&(\xi_1^2+3\xi_2^2)\frac{\mue\Lc}{3\eta}&
\xi_2^2\frac{\mue\Lc}{3\eta}&
(-2\xi_1^2+\xi_2^2)\frac{\mue\Lc}{3\eta}
&-\xi_1\xi_2\frac{\mue\Lc}{3\eta}
& -\xi_1\xi_2\frac{\mue\Lc}{\eta}\\
0&0&\xi_1^2\frac{\mue\Lc}{3\eta}&
(\xi_1^2-+\xi_2^2)\frac{\mue\Lc}{3\eta}&
(\xi_1^2-\xi_2^2)\frac{2\mue\Lc}{3\eta}&
-\xi_1\xi_2\frac{\mue\Lc}{3\eta}
&\xi_1\xi_2\frac{\mue\Lc}{\eta}\\
0&0&
-\xi_1^2\frac{\mue\Lc}{3\eta}&
-\xi_2^2\frac{\mue\Lc}{3\eta}&
\frac{2\mue\Lc}{3\eta}&
-\xi_1\xi_2\frac{2\mue\Lc}{3\eta}&
0\\
0&0&
-\xi_1\xi_2\frac{\mue\Lc}{2\eta}&
-\xi_1\xi_2\frac{\mue\Lc}{2\eta}&
-\xi_1\xi_2\frac{\mue\Lc}{\eta}&
\frac{\mue\Lc}{2\eta}&
(\xi_1^2-\xi_2^2)\frac{\mue\Lc}{2\eta}\\
0&0&
-\xi_1\xi_2\frac{\mue\Lc}{2\eta}&
\xi_1\xi_2\frac{\mue\Lc}{2\eta}&
0&
(\xi_1^2-\xi_2^2)\frac{\mue\Lc}{2\eta}&
\frac{\mue\Lc}{2\eta}
\end{array}\right),
\end{equation}
\small

\begin{equation}\label{BR1}
B^R_1=\left(\begin{array}{ccccccc}
0&0&-\xi_1\frac{2\mue}{\rho}&0&-\xi_1\frac{3\lame+2\mue}{\rho}&-\xi_2\frac{2\mue}{\rho}&-\xi_2\frac{2\mue}{\rho}\\
0&0&0&-\xi_2\frac{2\mue}{\rho}&-\xi_2\frac{3\lame+2\mue}{\rho}&-\xi_1\frac{2\mue}{\rho}&\xi_1\frac{2\mue}{\rho}\\
\xi_1\frac{4\mue}{3\eta}&-\xi_2\frac{2\mue}{3\eta}&0&0&0&0&0\\
-\xi_1\frac{2\mue}{3\eta}&\xi_2\frac{4\mue}{3\eta}&0&0&0&0&0\\
\xi_1\frac{3\lame+2\mue}{3\eta}&\xi_2\frac{3\lame+2\mue}{3\eta}&0&0&0&0&0\\
\xi_2\frac{\mue}{\eta}&\xi_1\frac{\mue}{\eta}&0&0&0&0&0\\
\xi_2\frac{\muc}{\eta}&-\xi_1\frac{\muc}{\eta}&0&0&0&0&0
\end{array}\right),
\end{equation}
\begin{equation}\label{CR1}
C^R_1=\left(\begin{array}{ccccccc}
0&0&0&0&0&0&0\\
0&0&0&0&0&0&0\\
0&0&-\frac{2(\mue+\mumic)}{\eta}&0&0&0&0\\
0&0&0&-\frac{2(\mue+\mumic)}{\eta}&0&0&0\\
0&0&0&0&-\frac{(2\mue+3\lame)+(2\mumic+3\lammic)}{\eta}&0&0\\
0&0&0&0&0&-\frac{2(\mue+\mumic)}{\eta}&0\\
0&0&0&0&0&0&-\frac{2\muc}{\eta}
\end{array}\right).
\end{equation}
Finally, we have that 
\begin{equation}\label{A1Decomp}
A_1= \kabs^2 A^R_1 -\omega^2 \idmat -i\kabs B^R_1 - C^R_1.
\end{equation}
As for $A_2$, it can be written in a similar fashion. We have
\begin{equation}\label{AR2}
A^R_2=\left(\begin{array}{ccccc}
\frac{\muc + \mue}{\rho} & 0 & 0 & 0 & 0\\
0 & \frac{\mue\Lc}{\eta}(\frac{1}{2}\xi_1^2 +\xi_2^2 ) & \frac{\mue\Lc}{2\eta}\xi_1^2 & -\frac{\mue\Lc}{2\eta}\xi_1\xi_2 & \frac{\mue\Lc}{2\eta}\xi_1\xi_2 \\
0& \frac{\mue\Lc}{2\eta}\xi_1^2 & \frac{\mue\Lc}{\eta}(\frac{1}{2}\xi_1^2+\xi_2^2) & \frac{\mue\Lc}{2\eta}\xi_1\xi_2 & -\frac{\mue\Lc}{2\eta}\xi_1\xi_2\\
0& -\frac{\mue\Lc}{2\eta}\xi_1\xi_2 & \frac{\mue\Lc}{2\eta}\xi_1\xi_2 &\frac{\mue\Lc}{\eta}(\xi_1^2 +\frac{1}{2}\xi_2^2 ) & \frac{\mue\Lc}{2\eta}\xi_2^2 \\
0 &  \frac{\mue\Lc}{2\eta}\xi_1\xi_2 &  -\frac{\mue\Lc}{2\eta}\xi_1\xi_2 & \frac{\mue\Lc}{2\eta}\xi_2^2 &\frac{\mue\Lc}{\eta}(\xi_1^2 + \frac{1}{2}\xi_2^2)
\end{array}\right),
\end{equation}
\begin{equation}\label{BR2CR2}
\begin{array}{cc}
B^R_2 = \left( \begin{array}{ccccc}
0 & -\frac{2\mue}{\rho}\xi_1 & \frac{2\muc}{\rho}\xi_1 & -\frac{2\mue}{\rho}\xi_2 &  \frac{2\muc}{\rho}\xi_2 \\
\frac{\mue}{\eta}\xi_1 & 0 & 0 & 0 & 0\\
-\frac{\muc}{\eta}\xi_1 & 0 & 0 & 0 & 0\\
\frac{\mue}{\eta}\xi_2 & 0 & 0 & 0 & 0\\
-\frac{\muc}{\eta}\xi_2 & 0 & 0 & 0 & 0
\end{array}\right),&
C^R_2=\left(\begin{array}{ccccc}
0 & 0 & 0 & 0 & 0\\
0 & -\frac{2(\mue+\mumic)}{\eta} & 0 & 0 & 0\\
0 & 0 & -\frac{2\muc}{\eta} & 0 & 0\\
0 & 0 & 0& -\frac{2(\mue + \mumic)}{\eta} & 0 \\
0 & 0 & 0 & 0 & -\frac{2\muc}{\eta}
\end{array}\right).
\end{array}
\end{equation}

And then, we can write $A_2$ as 
\begin{equation}\label{A2Decomp}
A_2 = \kabs^2 A^R_2 - \omega^2\idmat - i\kabs B^R_2 - C^R_2.
\end{equation}

\subsection{Energy flux matrices}\label{appendixRelaxed6}
\small 
The expression for the energy flux in the relaxed micromorphic model for the case of in-plane motion is given by 
\begin{equation*}
H_1=v^1_{,t}\cdot (H^{11}\cdot v^1_{,1} +H^{12}\cdot v^1_{,2}+H^{13}\cdot v^1), 
\end{equation*}
where 
\footnotesize
\begin{equation}\label{eq:H11}
H^{11}=\left( 
\begin{array}{ccccccc}
-2\mue - \lame  & 0 & 0 & 0 & 0 & 0 & 0 \\
0 & -\mue -\muc & 0 & 0 & 0 & 0 & 0 \\
0 & 0 & -\Lc \mue & -\Lc \mue & \Lc \mue & 0 & 0 \\
0 & 0 & -\Lc \mue & -2\Lc \mue & 0 & 0 & 0 \\
0 & 0 & \Lc \mue & 0 & -2\Lc \mue & 0 & 0\\
0 & 0 & 0 & 0& 0 & -\Lc \mue & -\Lc \mue \\
0 & 0 & 0 & 0& 0 & -\Lc \mue & -\Lc \mue
\end{array}\right),
\end{equation}
\begin{equation}\label{eq:H12}
\begin{array}{cc}
H^{12}= \left(
\begin{array}{ccccccc}
0 & -\lame & 0 & 0 & 0 & 0 & 0 \\
\muc - \mue & 0 & 0 & 0 & 0 & 0 & 0\\
0 & 0 & 0 & 0 & 0 & 0 & 0 \\
0 & 0 & 0 & 0 & 0 & \Lc \mue & -\Lc \mue \\
0 & 0 & 0 & 0 & 0 & \Lc \mue & -\Lc \mue \\
0 & 0 & \Lc \mue & 0 & \Lc \mue & 0 & 0 \\
0 & 0 & \Lc \mue & 0 & \Lc \mue & 0 & 0
\end{array}\right),&
H^{13} = \left(
\begin{array}{ccccccc}
0 & 0 & 2\mue & 0 &  2\mue + 3\lame & 0 & 0 \\
0 & 0 & 0 & 0 & 0 & 2\mue & -2\muc \\
0 & 0 & 0 & 0 & 0 & 0 & 0 \\
0 & 0 & 0 & 0 & 0 & 0 & 0 \\
0 & 0 & 0 & 0 & 0 & 0 & 0 \\
0 & 0 & 0 & 0 & 0 & 0 & 0 \\
0 & 0 & 0 & 0 & 0 & 0 & 0
\end{array} \right).
\end{array}
\end{equation}
\small
In the case of out-of-plane motion, the expression for the energy flux in the relaxed micromorphic model is given by 
\footnotesize
\begin{equation*}
H_1=v^2_{,t}\cdot (H^{21}\cdot v^2_{,1} +H^{22}\cdot v^2_{,2}+H^{23}\cdot v^2), 
\end{equation*}
where 
\begin{equation}\label{H2122}
H^{21}=\left(
\begin{array}{ccccc}
-\muc-\mue & 0 & 0 & 0 & 0 \\
0 & 0 & 0 & 0 & 0 \\
0 & 0 & 0 & 0 & 0 \\ 
0 & 0 & 0 & 0 & 0 \\
0 & 0 & 0 & -\Lc \mue & \Lc \mue \\
0 & 0 & 0 & \Lc \mue & -\Lc \mue
\end{array}\right), \quad 
H^{22}=\left(
\begin{array}{ccccc}
0 & 0 & 0 & 0 & 0  \\
0 & 0 & 0 & 0 & 0 \\
0 & 0 & 0 & 0 & 0 \\ 
0 & 0 & 0 & 0 & 0 \\
0 & \Lc \mue & -\Lc \mue & 0 & 0  \\
0 & -\Lc \mue & \Lc \mue & 0 & 0 
\end{array}\right),
\end{equation}
\begin{equation}\label{H23}
H^{23}=\left(
\begin{array}{ccccc}
0 & 2\mue & -2\muc & 0 & 0  \\
0 & 0 & 0 & 0 & 0 \\
0 & 0 & 0 & 0 & 0 \\ 
0 & 0 & 0 & 0 & 0 \\
0 & 0 & 0 & 0 & 0  \\
0 & 0 & 0 & 0 & 0 
\end{array}\right).
\end{equation}

\end{document}